\documentclass[11pt,letterpaper]{article}

\usepackage[backend=biber,
style=alphabetic,
sorting=nyt,
minalphanames=3,
maxbibnames=99]{biblatex}

\usepackage{amsmath,amsthm,amsfonts,amssymb,thmtools}
\usepackage{mathtools}
\usepackage{physics}
\usepackage{braket}
\usepackage{array}
\allowdisplaybreaks

\usepackage[margin=1in]{geometry}
\usepackage{xspace}
\usepackage{xcolor}
\usepackage{subcaption}
\usepackage{enumitem}
\usepackage{microtype}
\usepackage{csquotes}

\usepackage{hyperref}
\hypersetup{
    colorlinks=true,
    linkcolor=black,
    citecolor=black,
    filecolor=black,
    urlcolor=black,
}
\usepackage[capitalize]{cleveref}
\crefname{step}{step}{steps}
\Crefname{step}{Step}{Steps}

\usepackage[skins,many]{tcolorbox}

\newtheorem{theorem}{Theorem}[section]
\newtheorem{lemma}[theorem]{Lemma}
\newtheorem{corollary}[theorem]{Corollary}

\newtheorem{observation}[theorem]{Observation}
\theoremstyle{definition}
\newtheorem{definition}[theorem]{Definition}
\theoremstyle{remark}
\newtheorem{remark}[theorem]{Remark}

\newcommand{\local}{LOCAL\xspace}
\newcommand{\qlocal}{quantum-\local}
\newcommand{\detlocal}{det-\local}
\newcommand{\randlocal}{rand-\local} 
\newcommand{\slocal}{S\local}
\newcommand{\olocal}{online-\local} 

\newcommand{\dep}{dependence\xspace}
\newcommand{\dept}{dependent\xspace}
\newcommand{\boundep}{bounded-\dep}

\newcommand{\boundept}{bounded-\dept}

\newcommand{\nonsign}{non-signaling\xspace}

\newcommand{\myO}[1]{O\mathopen{}\left(#1\right)\mathclose{}}

\DeclareMathOperator{\dist}{dist} 
\newcommand{\neighborhood}{\NN}
\DeclareMathOperator{\poly}{poly} 

\newcommand{\floor}[1]{\left \lfloor #1 \right \rfloor}
\newcommand{\view}{\VV}
\DeclareMathOperator{\diam}{diam}

\newcommand{\nats}{\mathbb{N}}

\renewcommand{\AA}{\mathcal{A}}

\newcommand{\CC}{\mathcal{C}}

\newcommand{\EE}{\mathcal{E}}
\newcommand{\FF}{\mathcal{F}}

\newcommand{\HH}{\mathcal{H}}
\newcommand{\II}{\mathcal{I}}

\newcommand{\LL}{\mathcal{L}}

\newcommand{\NN}{\mathcal{N}}

\newcommand{\PP}{\mathcal{P}}

\newcommand{\VV}{\mathcal{V}}

\newcommand{\st}{\ \middle| \ }

\DeclareMathOperator{\oupt}{out}
\newcommand{\localVar}{\mathrm{x}}
\newcommand{\problem}{\Pi}
\newcommand{\maxDeg}{\Delta}
\newcommand{\algo}{\AA}
\newcommand{\outcome}{\mathrm{O}}

\makeatletter
\DeclareMathOperator{\regularlog}{log}
\renewcommand{\log}{\protect\@ifstar{\regularlog^*}{\regularlog}}
\makeatother

\newcommand{\Vin}{\VV_{\text{in}}}
\newcommand{\Ein}{\EE_{\text{in}}}
\newcommand{\Vout}{\VV_{\text{out}}}
\newcommand{\Eout}{\EE_{\text{out}}}
\newcommand{\inptLbl}{\ell_\text{in}}
\newcommand{\outLbl}{\ell_\text{out}}
\newcommand{\lbl}{\ell}

\newtcolorbox{myframe}[2][]{%
	breakable,enhanced,colback=white,colframe=black,coltitle=black,
	sharp corners,boxrule=0.4pt,
	fonttitle=\itshape,
	attach boxed title to top left={yshift=-0.3\baselineskip-0.4pt,xshift=2mm},
	boxed title style={tile,size=minimal,left=0.5mm,right=0.5mm,
		colback=white,before upper=\strut},
	title=#2,#1
}
\newcommand{\ltreelike}{\mathsf {tree}}
\newcommand{\lparent}{\mathsf {P}}
\newcommand{\lleft}{\mathsf {L}}
\newcommand{\lright}{\mathsf {R}}

\newcommand{\llch}{\ensuremath{\mathsf {Ch_L}}}
\newcommand{\lrch}{\ensuremath{\mathsf {Ch_R}}}
\newcommand{\lerror}{\mathsf {Error}}

\DeclareMathOperator{\iterghz}{GHZ}

\newcommand{\lbadgraph}{\mathsf {badGraph}}
\newcommand{\lbadtree}{\mathsf {badTree}}
\newcommand{\lproper}{\mathsf {proper}}
\newcommand{\lpromise}{\mathsf {promise}}
\newcommand{\lbadgadget}{\mathsf {badOctopus}}
\newcommand{\linearizable}{\mathsf {linearizable}}

\newcommand{\lgadget}{\mathsf {octopus}}

\newcommand{\gadgetN}{\mathsf {head}} 
\newcommand{\gadgetP}{\mathsf {port}} 
\newcommand{\nplink}{\mathsf {hp_{link}}}
\newcommand{\pnlink}{\mathsf {ph_{link}}}

\newcommand{\inter}{\mathsf{inter\text{-}octopus}}
\newcommand{\proper}{\mathsf {proper}}
\newcommand{\iplink}{\mathsf {ip_{link}}}
\newcommand{\pilink}{\mathsf {pi_{link}}}
\newcommand{\badgraph}{\mathsf {badGraph}}
\newcommand{\linter}{\mathsf{inter}}
\newcommand{\lintra}{\mathsf{intra}}

\newcommand{\lfirst}{\mathsf {first}}
\newcommand{\lother}{\mathsf {other}}

\newcommand{\compression}{\mathsf{compression}}

\begin{document}
\begin{flushleft}
    \huge\bf
    Distributed Quantum Advantage in \\
    Locally Checkable Labeling Problems\footnote{This work was originally presented at SODA 2026 \cite{balliu2026}}
\end{flushleft}
\smallskip
\begin{flushleft}
    \setlength{\parskip}{3pt}
    \textbf{Alkida Balliu} · Gran Sasso Science Institute

    \textbf{Filippo Casagrande} · Gran Sasso Science Institute

    \textbf{Francesco d'Amore} · Gran Sasso Science Institute

    \textbf{Massimo Equi} · Aalto University

    \textbf{Barbara Keller} · Aalto University

    \textbf{Henrik Lievonen} · Aalto University

    \textbf{Dennis Olivetti} · Gran Sasso Science Institute

    \textbf{Gustav Schmid} · University of Freiburg

    \textbf{Jukka Suomela} · Aalto University

\end{flushleft}
\smallskip

\paragraph{Abstract.}

    In this paper, we present the first known example of a locally checkable labeling problem (LCL) that admits asymptotic distributed quantum advantage in the LOCAL model of distributed computing: our problem can be solved in $O(\log n)$ communication rounds in the quantum-LOCAL model, but it requires $\Omega(\log n \cdot \log^{0.99} \log n)$ communication rounds in the classical randomized-LOCAL model.

    We also show that distributed quantum advantage cannot be arbitrarily large: if an LCL problem can be solved in $T(n)$ rounds in the quantum-LOCAL model, it can also be solved in $\tilde O(\sqrt{n T(n)})$ rounds in the classical randomized-LOCAL model. In particular, an LCL problem that is strictly global classically is also almost-global in quantum-LOCAL.

    This solves a major open question at the intersection of distributed graph algorithms and quantum computing. LCL problems [Naor and Stockmeyer, STOC 1993] have been extensively studied in the past decade and they are by now very well-understood in the classical LOCAL model, yet whether any of them admits a genuine quantum advantage has remained open. Coiteux-Roy et al.\ [STOC 2024] showed that for \emph{some specific} LCL problems, \emph{if} quantum helps, it cannot help by much. Akbari et al.\ [STOC 2025] showed that, for \emph{some} LCLs, in \emph{rooted trees}, quantum-LOCAL and randomized-LOCAL have the same power. Balliu et al.\ [STOC 2025] showed that quantum helps in solving locally checkable problems faster when the maximum degree of the graph is super-constant; however, this does not give any asymptotic quantum advantage for LCLs. The above-mentioned works repeatedly asked the key question of whether quantum helps for LCLs. We solve this open question by giving the first example of an LCL problem that admits a super-constant distributed quantum advantage, and by also giving the first result that puts limits on distributed quantum advantage for LCLs in general graphs.

    Our second result also holds for $T(n)$-dependent probability distributions. As a corollary, if there exists a \emph{finitely dependent distribution} over valid labelings of some LCL problem $\Pi$, then the same problem $\Pi$ can also be solved in $\tilde O(\sqrt{n})$ rounds in the classical randomized-LOCAL and deterministic-LOCAL models. That is, finitely dependent distributions cannot exist for global LCL problems.

\thispagestyle{empty}
\setcounter{page}{0}
\newpage
\section{Introduction}\label{sec:introduction}

Does quantum computation and communication help with solving graph problems in the distributed setting? We study this question in the usual LOCAL model of distributed computing, more precisely comparing these two settings (see \cref{sec:preliminaries} for precise definitions):
\begin{itemize}
    \item \textbf{Randomized-LOCAL model:} Each node of the input graph is a classical computer (that can store an arbitrary number of classical bits), and each edge is a classical communication channel (that can transmit an arbitrary number of classical bits per communication round). Each node is initialized with its own independent random bit string.
    \item \textbf{Quantum-LOCAL model:} Each node of the input graph is a quantum computer (that can store an arbitrary number of qubits), and each edge is a quantum communication channel (that can transmit an arbitrary number of qubits per communication round). Each node is initialized with its own unentangled qubits.
\end{itemize}
Given a graph problem $\Pi$ (e.g.\ graph coloring), we say that it can be solved in $T(n)$ rounds if there is a distributed algorithm $A$ such that in any $n$-node graph after $T(n)$ communication rounds each node terminates and outputs its own part of the solution (e.g.\ its own color).

\subsection{Prior work on distributed quantum advantage}

It is not at all obvious if quantum-LOCAL could possibly admit any advantage over randomized-LOCAL; after all, we did not put any limits on the amount of local computation, local storage, or number of bits communicated per round. Could it be the case that any quantum-LOCAL algorithm can be simulated with a classical algorithm in the same number of rounds (at least asymptotically), if we just encode the local state of qubits with an exponentially larger number of classical bits?

Surprisingly, this turns out not to be the case: \textcite{legall2019} show that there are graph problems that can be solved in $O(1)$ rounds in quantum-LOCAL but that require $\Omega(n)$ rounds in classical randomized-LOCAL. However, the problem studied in \cite{legall2019} is very different from problems commonly studied in the theory of distributed graph algorithms. In particular, their problem has an inherently \emph{global} specification.

\subsection{Prior work on LCL problems}

In recent years, a large body of literature has focused on \textbf{locally checkable labeling problems}, or LCLs in brief, first introduced by \textcite{naor1995}. These are graph problems in which valid solutions can be specified by listing a finite set of valid labeled neighborhoods. LCLs strike a balance between being broad enough so that they contain a large number of interesting graph problems and being narrow enough so that it is possible to prove strong theorems that apply to all LCL problems, see e.g.\ \cite{naor1995,chang19hierarchy,chang_kopelowitz_pettie2019exp_separation,brandt16lll,balliu18lcl-complexity,balliu20almost-global,balliu20lcl-randomness,suomela-2020-landscape,dahal_et_al:LIPIcs.DISC.2023.40}. There are numerous results about LCL problems in the classical LOCAL model, yet we do not know if \emph{all} of these results hold also in the quantum-LOCAL model. Hence, this is the key question that we study: \textbf{do any LCL problems admit a distributed quantum advantage?}

There are several papers that have so far delivered mainly negative results: we have learned about cases in which quantum cannot help, at least not much \cite{coiteuxroy2023,dhar24rand,akbari2024,balliu-coupette-etal-2025-new-limits-on-distributed}, and we have also learned about barriers for studying such questions \cite{akbari2024}. The main exception is the very recent work \cite{balliu2024quantum} that showed the following separation result: there is a family of LCL problems $\iterghz(\Delta)$ parameterized by maximum degree $\Delta$ such that $\iterghz(\Delta)$ can be solved in $O(1)$ rounds in quantum-LOCAL but it requires $\Omega(\Delta)$ rounds in randomized-LOCAL. However, this does not yield a super-constant separation for any fixed LCL problem: for any fixed $\Delta$, problem $\iterghz(\Delta)$ is an LCL that can be solved in $O(1)$ rounds (albeit with very different constants) in both quantum-LOCAL and randomized-LOCAL.

Is it possible to exhibit a single LCL problem $\Pi$ such that $\Pi$ is solvable in $T(n)$ rounds in quantum-LOCAL but requires $\omega(T(n))$ rounds in classical models? This is a \textbf{key open question} that has been mentioned repeatedly in the literature, see e.g.\ \cite{coiteuxroy2023,akbari2024,balliu2024quantum,suomela-open}.

\subsection{Contribution 1: first LCL problem with a quantum advantage}

In this work we show that the answer is yes: quantum-LOCAL is strictly stronger than randomized-LOCAL for LCL problems. More precisely, in \cref{sec:high-level-lcl,sec:lcl-details} we prove:
\begin{theorem}\label{thm:intro:sep}
    There is an LCL problem $\Pi$ such that the round complexity of $\Pi$ is $O(\log n)$ in quantum-LOCAL but $\Omega(\log n \cdot \log^{0.99} \log n)$ in randomized-LOCAL.
\end{theorem}
This brings us both good news and bad news: The good news is that this demonstrates that distributed quantum computing indeed helps even when restricted to the family of LCL problems. The bad news is that there is now no longer hope that we could prove a theorem that shows that quantum-LOCAL and randomized-LOCAL are asymptotically equally strong for LCL problems, enabling us to lift a large body of prior work from classical models to the quantum model---to fully characterize an LCL problem, we may need to separately classify its locality in deterministic-LOCAL, randomized-LOCAL, and quantum-LOCAL.

\subsection{Contribution 2: limits on quantum advantage for LCLs}

While \cref{thm:intro:sep} resolves the open question, the separation that we have between the classical and quantum models is tiny, especially if we compare this with the non-LCL problem from \cite{legall2019} that admits an $O(1)$-round quantum algorithm and requires $\Omega(n)$ rounds in classical models. Could we demonstrate such a separation with LCL problems?

In this work we show that the answer is no: there is no LCL with a linear-in-$n$ gap between classical and quantum round complexities. Formally, in \cref{sec:quantum-advantage} we show:
\begin{theorem}\label{thm:intro:sim}
    Let $\Pi$ be any LCL problem that can be solved in $T(n)$ rounds in quantum-LOCAL. Then $\Pi$ can be solved in $O(\sqrt{n T(n)} \poly \log n)$ rounds in randomized-LOCAL.
\end{theorem}
In particular, if the classical round complexity is $\Omega(n)$, i.e., the problem is inherently global, then also in the quantum-LOCAL model we need $\tilde\Omega(n)$ rounds (here we are using $\tilde\Omega$ and $\tilde O$ to hide polylogarithmic factors).

The main open question after this work is narrowing down the gap between \cref{thm:intro:sep,thm:intro:sim}: can we construct LCLs where quantum-LOCAL helps more than some doubly-logarithmic factors? Currently, there is a larger gap between deterministic-LOCAL and randomized-LOCAL than randomized-LOCAL and quantum-LOCAL, which seems counterintuitive: is access to classical randomness already almost as good as the ability to manipulate qubits?

While we do not know yet if \cref{thm:intro:sim} is tight, in \cref{ssec:intro-sim-ideas} we will see that there is a direct generalization of \cref{thm:intro:sim} to stronger (super-quantum) models, and there the result turns out to be close to the best possible, in the sense that we may have a gap of $\tilde{\Omega}(\sqrt{n})$ between classical models and super-quantum models. Hence to strengthen the claim of \cref{thm:intro:sim}, we need techniques that do \emph{not} generalize far beyond quantum-LOCAL.

\subsection{Key ideas in the proof of Theorem~\ref{thm:intro:sep}}\label{ssec:key-ideas-separation}

In the proof of \cref{thm:intro:sep}, the key novel ingredient is the observation that the $\iterghz(\Delta)$ problem from \cite{balliu2024quantum} can be \emph{linearized} (with high-degree nodes replaced by paths in a way that preserves local checkability), and then further \emph{padded} with the technique from \cite{balliu20lcl-randomness} to construct a genuine LCL problem that admits a quantum advantage. We summarize here the key elements of the proof.

\paragraph{LCL problems.}
As our goal is to construct an LCL problem that exhibits distributed quantum advantages, let us recall what kind of entities LCL problems are. We will present a formal definition in \cref{sec:preliminaries}, but the following informal description suffices for now: An LCL problem $\Pi$ can be specified by listing a \textbf{finite set} of valid labeled radius-$r$ neighborhoods $\mathcal{N}$ for some constant $r$. The task is to label nodes and/or edges of the input graph $G$, and a solution is feasible if for every node $v$, the radius-$r$ neighborhood of $v$ is isomorphic to one of the neighborhoods in $\mathcal{N}$.

In particular, this implies that the set of node labels and edge labels has to be also finite, and $G$ must be a bounded-degree graph in order to admit a solution to $\Pi$. This leads to another equivalent way of characterizing LCL problems: they are graph problems for bounded-degree graphs, where the task is to label nodes and/or edges with labels from a finite alphabet, and the validity of a solution can be verified in a distributed manner by checking all constant-radius local neighborhoods.

The restrictions may sound at first technical, but this exact definition originally introduced by \cite{naor1995} has been tremendously successful: this is a problem family that is now very well-understood especially for the classical LOCAL model. What we seek to demonstrate is that inside this well-understood yet restrictive family there is indeed a problem $\Pi$ that admits a super-constant quantum advantage.

\paragraph{Starting point: the iterated GHZ problem.}
Let us next recall the key elements of the problem family $\iterghz(\Delta)$ from \cite{balliu2024quantum}. Here $\iterghz(\Delta)$ is a family of LCL problems, parameterized by $\Delta$. The problem family is defined so that hard instances are bipartite graphs, where one part consists of white nodes of degree $\Delta$ and the other part consists of black nodes of degree $3$. White nodes represent \emph{players} and black nodes represent \emph{games}---more precisely, they are instances of the GHZ game \cite{GHZineq,mermin1990quantum}, which is a $3$-party game where players can win the game without communication if they hold entangled qubits. The games are labeled with colors $1,2,\dotsc,\Delta$, so that each player is adjacent to one game of each color.
\begin{center}
    \includegraphics[page=1]{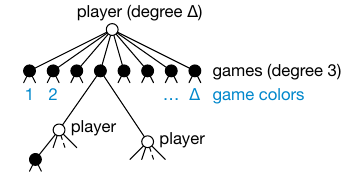}
\end{center}
Crucially, for each player, the input that they have in game of color $i$ is the output they got from game of color $i-1$. The key idea is that for a classical algorithm to win in a color-$i$ game, it first has to know the outputs of color-${(i-1)}$ games, and spend additional communication rounds after that to coordinate a feasible solution for the color-$i$ game. Such a strategy is inherently sequential, taking $\Omega(\Delta)$ rounds, and it can be shown with the round elimination technique \cite{Brandt2019automatic} that this is indeed the best that a classical algorithm can do (at least in neighborhoods that are tree-like). However, a quantum algorithm can do much better: we can use one communication round to establish a set of shared qubits for each game, and then the players can do the rest locally, by applying appropriate measurements to the shared qubits; this is $O(1)$ rounds independently of~$\Delta$.

At this point it is also good to recall that in the classical LOCAL model, time and distances are interchangeable: a $T$-round algorithm is equivalent to a function that maps radius-$T$ neighborhoods to local outputs (and vice versa). So the key point is this: any classical algorithm that solves $\iterghz(\Delta)$ has to see up to distance $\Omega(\Delta)$, in neighborhoods that locally look like a two-colored tree in which all white nodes have degree $\Delta$ and black nodes have degree $3$ (we will call it a $(\Delta,3)$-biregular tree).

Now $\iterghz(\Delta)$ is a well-defined graph problem even if we plug in e.g.\ $\Delta = \sqrt{n}$, but then we cannot claim that any classical algorithm for solving it requires $\Omega(\sqrt{n})$ rounds, simply because we can no longer construct a hard instance that would be a $(\sqrt{n},3)$-biregular tree of depth $\sqrt{n}$ with at most $n$ nodes, which certainly does not exist. However, a more modest choice of $\Delta \approx \log n / \log \log n$ makes sense and results in a graph problem that admits an $O(1)$-round quantum algorithm but requires $\Omega(\log n / \log \log n)$ rounds for any classical deterministic algorithm (the case of classical \emph{randomized} algorithms is more involved, but let us put it aside for now for the purposes of this informal overview). However, this is \textbf{not an LCL problem}. It is some labeling problem in which validity of a solution is locally verifiable, but it cannot be described with any finite set of valid neighborhoods; in particular, we will need super-constant degrees in the input graph.

Yet this is a useful source of inspiration for the present work. We aim at engineering a graph problem $\Pi$ that captures the spirit of ``what would happen if we take $\iterghz(\Delta)$ and plug in something like $\Delta \approx \log n / \log \log n$,'' and at the same time $\Pi$ is a genuine LCL problem in the strict sense of \cite{naor1995}, without any cheating.

\paragraph{Linearizing the iterated GHZ problem.}

The key new observation is that $\iterghz(\Delta)$ happens to have a very convenient property: for each player, the set of valid labels on the incident edge number $i$ only depends on the labels of the incident edge number $i-1$. In particular, there is a finite automaton that processes the labels of the incident edges in the order specified by the colors of the games and determines if a solution is valid from the player's perspective. This makes it possible to \textbf{linearize} the problem so that we can replace a player of degree $\Delta$ with a path of $\Delta$ nodes. This way we can (1)~eliminate high-degree nodes, (2)~preserve the property that we have a finite set of labels, and (3)~preserve the property that the feasibility of a solution can be verified locally.
\begin{center}
    \includegraphics[page=2]{figs.pdf}
\end{center}
Therefore we can define an LCL problem $\Pi_1$, where all nodes have low degrees, and a path with $\Delta$ white nodes behaves as if it was a degree-$\Delta$ player in the original $\iterghz(\Delta)$ problem. However, while doing this translation, we lose all control over the degrees of the players. What if the adversary constructs an input where there is a white path of length $\Theta(n)$? Even a quantum algorithm would have to spend $\Theta(n)$ rounds to just emulate the behavior of a single $\Theta(n)$-degree player, and hence there is no room for quantum advantage (as \emph{any} problem can be solved in $O(n)$ rounds by brute force).

We circumvent this challenge with a commonly-used solution: instead of a path we use a tree-like construction (a balanced binary tree where each level forms a path). In particular, this way the white nodes that represent a $\Delta$-degree player will be within distance $O(\log \Delta)$ from each other, and even in the worst case the overhead of simulating one such node is $O(\log n)$. We can furthermore add additional rules that ensure that if the adversary presents a graph that does not represent a valid instance, we can nevertheless produce a valid labeling in $O(\log n)$ rounds; let us call the resulting LCL problem $\Pi_2$.
\begin{center}
    \includegraphics[page=3]{figs.pdf}
\end{center}

\paragraph{Padding.}
Unfortunately, $\Pi_2$ as constructed above does not admit any distributed quantum advantage. To construct a hard instance for classical algorithms, we would need to construct a $(\Delta,r)$-biregular tree of depth $\Omega(\Delta)$, and then replace each degree-$\Delta$ node with our tree-like construction. But we will have to use $\Delta = O(\log n / \log \log n)$ so that we do not run out of nodes. To construct a worst-case instance, we would now take an $O(\Delta)$-deep $(\Delta,r)$-biregular tree and replace each white degree-$\Delta$ node with a tree-like construction of diameter $O(\log \Delta)$. This replacement process increases the diameter of the graph from $O(\Delta)$ to $O(\Delta) \cdot O(\log \Delta)$. But the diameter of the entire graph would be still bounded by $O(\Delta) \cdot O(\log \Delta) = O(\log n)$, and a classical algorithm could solve it trivially in $O(\log n)$ rounds by gathering everything. On the other hand, a quantum algorithm needs $\Omega(\log n)$ rounds to merely simulate one player in the worst case, so we gained nothing.

We circumvent this issue by borrowing the idea of \textbf{padding} from \cite{balliu20lcl-randomness}. In essence, we replace each edge between a player (white node) and a game (black node) with another tree-like construction, where the adversary gets to choose the height of the tree.
\begin{center}
    \includegraphics[page=4]{figs.pdf}
\end{center}
It turns out that this is enough to construct an LCL problem $\Pi$ that exhibits a distributed quantum advantage! In essence, the reason for this is as follows:
\begin{itemize}
    \item Quantum algorithms: In the worst case the adversary can add padding of depth $O(\log n)$, and the total overhead that we have in simulating the original $O(1)$-round quantum algorithm is $O(\log n)$ rounds to simulate the activities of the players plus $O(\log n)$ rounds to transmit information across a padded edge, in total $O(\log n)$ rounds.
    \item Classical algorithms: The adversary can select a suitable super-constant $\Delta$ and use trees of height $O(\log \Delta)$ to represent players and trees of height $O(\log n)$ for padding. We can start with a $\iterghz(\Delta)$-instance with e.g.\ $\Theta(n^{1/3})$ players, and we can then afford to use e.g.\ $\Theta(n^{1/3})$ nodes for each tree that we use in padding. This results in an instance in which the algorithm has to see beyond $\Delta$ padded edges, in total up to distance $\Omega(\Delta \log n)$, which is worse than the complexity of the quantum algorithm for a super-constant $\Delta$.
\end{itemize}
This is the essence of the construction. We give a more formal high-level overview in \cref{sec:high-level-lcl} and the technical details are postponed to \cref{sec:lcl-details}.

\paragraph{Generalizations and future work.}
We note here that exactly the same framework is applicable to a wide range of problems beyond $\iterghz(\Delta)$, as we are only making essential use of the fact that the validity of a solution for high-degree nodes can be verified with a finite automaton that processes edges in a sequential order, and hence we can linearize the problem. In particular, the widely-studied maximal matching problem has the same property.

Therefore we expect that the same construction will find use in engineering LCL problems that separate other models. For example, the maximal matching problem would be a good candidate for proving a separation between the quantum-LOCAL model and the SLOCAL model \cite{ghaffari2017}.

\subsection{Key ideas in the proof of Theorem~\ref{thm:intro:sim}}\label{ssec:intro-sim-ideas}

To prove \cref{thm:intro:sim}, we first take a step from distributed quantum algorithms to bounded-dependence distributions, and then observe that while such distributions \emph{cannot} be efficiently sampled in the randomized-LOCAL model, we can nevertheless do a combination of \emph{partial sampling and brute-force completion} to solve the same LCL problem.

\paragraph{Bounded-dependence distributions.}
The statement of \cref{thm:intro:sim} refers to the quantum-LOCAL model, but it will be much more convenient to work in the bounded-dependence model \cite{akbari2024}, which is a generalization of finitely-dependent distributions \cite{holroyd2016,holroyd2018,holroyd2024}.

Informally, a $T$-dependent distribution is a probability distribution $P$ over labelings of some graph $G$ where the following holds: if we choose two sets of nodes $X$ and $Y$ such that their $T$-radius neighborhoods do not overlap, then $P$ restricted to $X$ is independent of $P$ restricted to $Y$.
\begin{center}
    \includegraphics[page=5]{figs.pdf}
\end{center}
While it is easy to see that the output of a $T$-round randomized-LOCAL algorithm results in a $T$-dependent distribution, it also holds that the output of a $T$-round quantum-LOCAL algorithm results in a $T$-dependent distribution. Hence it suffices to show that if we have a $T$-dependent distribution $P$ over valid solutions of some LCL problem $\Pi$, we can construct a randomized-LOCAL algorithm that solves the same problem in $\tilde O(\sqrt{nT})$ rounds.

\paragraph{Sampling is hard.}
It would be tempting to now design a randomized-LOCAL algorithm that produces a sample from distribution $P$. However, this is \emph{not} possible. Indeed, if we could sample efficiently from any $T$-dependent distribution, we would also have a classical algorithm that efficiently simulates the distributed quantum algorithm from \cite{legall2019}, which is known to require $\Omega(n)$ rounds.

\paragraph{Partial sampling is easy.}
We will exploit the clustering algorithm from \cite{coiteuxroy2023}. In particular, with it we can partition nodes into well-separated clusters where the diameter of each cluster is $\tilde O(\sqrt{nT})$, the number of unclustered leftover nodes is $O(\sqrt{n/T})$, and clusters can be constructed in $\tilde O(\sqrt{nT})$ rounds in the randomized-LOCAL model.
\begin{center}
    \includegraphics[page=6]{figs.pdf}
\end{center}
We will split the leftover nodes into components that are at distance $\Omega(T)$ from each others; in the worst case all leftover nodes are in the same component and it will have diameter $D = O(\sqrt{nT})$.
\begin{center}
    \includegraphics[page=7]{figs.pdf}
\end{center}
Now if $X$ and $Y$ are two components of leftover nodes, they are sufficiently far from each other so that $P$ restricted to $X$ is independent of $P$ restricted to $Y$. Hence we can sample from $P$ independently in each such component, and the joint distribution will be a sample from $P$ restricted to all leftover nodes.

To recap, so far we have used $\tilde O(\sqrt{nT})$ rounds to compute a clustering, and another $O(\sqrt{nT})$ rounds to sample from $P$ in each component of leftover nodes. All leftover nodes have now committed to an output.
\begin{center}
    \includegraphics[page=8]{figs.pdf}
\end{center}

\paragraph{Completability helps us next.}
Now it would be tempting to switch to the clusters and locally sample from $P$ conditioned on what has already been sampled, with the hope of globally getting a sample from $P$.
But as we discussed above, this \emph{cannot work}, as it would let us simulate any quantum algorithms too efficiently in the classical models.

This is the point in which we will have to use the assumption that $P$ is an output distribution that solves the LCL problem~$\Pi$ with high probability.
In the previous step, we have produced a sample $S$ from $P$ restricted to the leftover nodes.
Equivalently, we can imagine that $S$ is constructed by the following process: we have first sampled a labeling $S^*$ of the entire graph from $P$, and then discarded the labels inside the clusters (keeping only the values of the leftover nodes).
Now with high probability $S^*$ is a valid solution to $\Pi$, so with high probability our partial solution $S$ can be \textbf{extended} to some globally valid solution of $\Pi$.
Furthermore, we can find such a completion independently for each cluster, e.g., simply by brute force: collect the entire cluster and the partial solution at its boundaries and find a valid completion.
Here we exploit the fact that $\Pi$ is a locally checkable problem: any completion that is locally valid in each cluster is also globally valid.
As the clusters had diameter $\tilde O(\sqrt{nT})$, this step takes also $\tilde O(\sqrt{nT})$ rounds, and the total running time is also $\tilde O(\sqrt{nT})$.

Put together, given a $T$-round quantum-LOCAL algorithm for some LCL problem $\Pi$, we can construct a randomized-LOCAL algorithm that solves the same problem in $\tilde O(\sqrt{nT})$ rounds. We will give the details of the proof in \cref{sec:quantum-advantage}.

\paragraph{Generalizations and future work.}
As we took a route through bounded-dependence distributions, we immediately get a stronger result: $T$-dependent distributions over LCLs can be simulated in randomized-LOCAL in $\tilde O(\sqrt{nT})$ rounds. In particular, this shows that there cannot exist an LCL that admits a finitely dependent distribution (i.e., $O(1)$-dependent distribution) but requires $\Omega(n)$ rounds to solve with classical algorithms---previously this was only known in trees \cite{akbari2024,dhar24rand}.

Our general proof strategy can be pushed further, towards stronger models. It works directly in the \emph{com\-ponent-wise version} of the online-LOCAL model \cite{akbari_et_al:LIPIcs.ICALP.2023.10,akbari2024}. Furthermore, online-LOCAL algorithms with locality $O(1)$ can be turned into component-wise algorithms \cite{akbari2024}, and hence we also get a result that puts limits on the gap between the online-LOCAL and classical LOCAL models. Our simulation result is also close to the best possible, as a much stronger result would contradict with the known bounds for the problem of $3$-coloring bipartite graphs, which is a problem that can be solved with locality $O(\log n)$ in (component-wise) online-LOCAL yet requires $\Omega(\sqrt n)$ rounds in the classical randomized-LOCAL model \cite{akbari_et_al:LIPIcs.ICALP.2023.10,coiteuxroy2023,akbari2024,chang23_tight_arxiv,brandt17grid-lcl}.
We will discuss the extension of our result to component-wise \olocal more in \cref{app:reduction}.

\section{Preliminaries}\label{sec:preliminaries}

\paragraph{Graphs.}

We work with simple undirected graphs unless otherwise specified.
Let \(G = (V,E)\) be any graph.
If the set of nodes \(V\) and the set of edges \(E\) are not specified, we refer to them by the notation \(V(G)\) and \(E(G)\), respectively.
For any subset of nodes \(A \subseteq V\), the subgraph of \(G\) induced by \(A\) is denoted by \(G[A]\).
The distance between any two nodes \(u,v \in V\) is the number of edges composing a shortest path between \(u\) and \(v\), and is denoted by \(\dist_G(u,v)\).
The notion of distance can be easily extended to subsets of nodes: 
Given any node \(u \in V\) and any two subsets \(A,B \subseteq V\), the distance between \(u\) and \(A\) is \(\dist_G(u,A) = \min_{v \in A}\{\dist_G(u,v)\}\) and the distance between \(A\) and \(B\) is \(\dist_G(A,B) = \min_{u \in A, v \in B}\{\dist_G(u,v)\}\).
When the graph is clear from the context, we omit the suffix and write only \(\dist()\) instead of \(\dist_G()\).

For any non-negative integer \(T\), the radius-\(T\) (closed) neighborhood of a node \(u\) in a graph \(G\) is the set \(\neighborhood_T[u] = \left\{v \in V \st \dist(u,v) \le T \right\}\).
More generally, the radius-\(T\) neighborhood of any subset of node \(A \subseteq V\) is \(\neighborhood_T[A] = \cup_{u \in A} \neighborhood_T[u]\).

We adopt the nomenclature \emph{half-edge} to refer to any pair \((v,e)\) where \(v \in V\) and \(e \in E\) is an edge incident to \(v\).
In order to be able to define the class of problems we consider, we need to define what a labeled graph is.

\begin{definition}[Labeled graph]\label{def:preliminaries:labeled-graph}
    Let \(\VV\) and \(\EE\) be sets of labels.
    A graph \(G = (V,E)\) is said to be \((\VV,\EE)\)-\emph{labeled} if the following statements are satisfied:
    \begin{enumerate}[noitemsep]
        \item Each node \(v \in V\) is assigned a label from \(\VV\);
        \item Each half-edge \((v,e) \in V \times E\) satisfying $v \in e$ is assigned a label from \(\EE\).
    \end{enumerate}
\end{definition}

We also define the notion of \emph{centered graph}, that is, a graph with a distinguished node that we call the \emph{center}.

\begin{definition}[Centered graph]
    Let \(H = (V,E)\) be any graph, and let \(v_H \in V_H\) be any node of \(H\).
    The pair \((H,v_H)\) is said to be a \emph{centered graph} and \(v_H\) is said to be the \emph{center} of \((H,v_H)\).
\end{definition}

We will focus on problems that ask to produce labelings over graphs that satisfy some constraints in all radius-\(r\) neighborhoods of any node.
We remind the reader that the \emph{eccentricity} of a node \(v\) of a graph \(H\) is the maximum distance between \(v\) and any other node of \(H\).

\begin{definition}[Set of constraints]\label{def:preliminaries:set-of-constraints}
    Let \(r, \maxDeg \in \nats\) be constants, \(I\) a finite set of indices, and
    \((\VV,\EE)\) a tuple of finite label sets.
    Let \(\CC\) be a finite set of centered graphs \(\{(H_i,v_{H_i})\}_{i \in I}\) where each \(H_i\) is a \((\VV,\EE)\)-labeled graph of degree at most \(\maxDeg\) and \(v_{H_i}\) has eccentricity at most \(r\).
    Then, \(\CC\) is said to be an \((r,\maxDeg)\)-\emph{set of constraints} over \((\VV,\EE)\).
\end{definition}

With the following definition, we explain what it means to satisfy a set of constraints.

\begin{definition}[Labeled graph satisfying a set of constraints]\label{def:preliminaries:constraint-satisfaction}
    Let \(G = (V,E)\) be a \((\VV,\EE)\)-labeled graph for some finite sets of labels \(\VV\) and \(\EE\).
    Let \(\CC\) be an \((r,\maxDeg)\)-set of constraints over \((\VV,\EE)\) for some finite \(r,\maxDeg \in \nats\).
    The graph \(G\) satisfies \(\CC\) if the following statement is satisfied:
    \begin{itemize}
        \item For every node \(v \in V\), the \((\VV,\EE)\)-labeled graph \(G[\neighborhood_r[v]]\) is such that \((G[\neighborhood_r[v]],v) \in \CC\).
    \end{itemize}
\end{definition}

We are now ready to define the class of problems of interest.

\begin{definition}[Locally checkable labeling (LCL) problems]\label{def:preliminaries:lcl-problems}
    Let \(r, \maxDeg \in \nats\) be constants.
    A locally checkable labeling (LCL) problem \(\problem\) is a tuple \((\Vin, \Ein, \Vout, \Eout, \CC)\) where \(\Vin\), \(\Ein\), \(\Vout\), and \(\Eout\) are finite label sets and \(\CC\) is an \((r,\maxDeg)\)-set of constraints over \((\Vin \times \Vout, \Ein \times \Eout)\).
\end{definition}

We now explain what it means to \emph{solve} an LCL problem \(\problem\).
We are given as input a \((\Vin,\Ein)\)-labeled graph \(G = (V,E)\).
For any node \(v \in V\) and half-edge \((v,e) \in V \times E\), let us denote the input label of \(v\) by \(\inptLbl(v)\), and the input label of \((v,e)\) by \(\inptLbl((v,e))\). We want to produce an output labeling on \(G\) so that we obtain a \((\Vout,\Eout)\)-labeled graph.
Let us denote by \(\outLbl(v)\) and by \(\outLbl((v,e))\) the output labels of \(v\) and \((v,e)\), respectively.
Consider the \((\Vin \times \Vout, \Ein \times \Eout)\)-labeled graph \(G\) where the labeling is defined as follows: the label of each node \(v\) is \(\lbl(v) = (\inptLbl(v),\outLbl(v))\), and the label of each half-edge \((v,e)\) is \(\lbl((v,e)) = (\inptLbl((v,e)),\outLbl((v,e)))\).
Now, the \((\Vin \times \Vout, \Ein \times \Eout)\)-labeled graph \(G\) must satisfy the \((r,\maxDeg)\)-set of constraints \(\CC\) over \((\Vin \times \Vout, \Ein \times \Eout)\) according to \cref{def:preliminaries:constraint-satisfaction}.

\paragraph{The \local model.} 
In the \local model of computing,
we are given a distributed system of \(n\) processors/nodes.
The processors are connected through a communication network that is modeled by a graph \(G = (V,E)\). 
Computation proceeds in synchronous rounds: in each round, nodes simultaneously exchange messages with their neighbors, perform some local computation, and update their states variables based on the received messages.
Note that communication bandwidth is unconstrained (i.e., messages can be of any size) and nodes are capable of arbitrary local computation, and it is assumed that all processes are fault-free and messages cannot be corrupted.
At the beginning of computation, all processors are identical and start with identical copies of the same state variables,
except for a distinguished local variable \(\localVar(v)\) which stores input data for node \(v\).
Input data for a node \(v\) encodes the number \(n\) of nodes in the network, a unique identifier from the set \([n^c] = \{1,2,\ldots,n^c\}\)  where \(c \ge 1\) is a fixed constant, and possible inputs defined by the problem of interest (we assume nodes store both input node labels and input half-edge labels).
If computation is randomized, \(\localVar(v)\) also encodes a private, infinite random bit string that is independent of the random bits of all other nodes.
In this case, we refer to the model as the \emph{randomized \local model} as opposed to the \emph{deterministic \local model}.
The goal of the computation is to assign to each node \(v\) an output label \(\outLbl(v)\) and computation ends when all nodes have decided on their output labels.
The running time of an algorithm is the number of communication rounds required to solve a problem.
If computation is randomized, we also ask that the algorithm solves the problem of interest with probability at least \(1 - 1/\poly(n)\), where \(\poly{n}\) is any polynomial function in \(n\).
Since computation and message size is unbounded, we can look at \(T\)-rounds \local algorithms as functions mapping radius-\(T\) neighborhoods of nodes to output labels (in the deterministic case) or probability distributions over output labels (in the randomized case).
That is why we refer to the parameter \(T\) as the \emph{locality} of the algorithm.

\paragraph{The \qlocal model.}
The \qlocal model of computing is similar to the deterministic LOCAL
model above, where we replace the classical processors with quantum processors and the classical communication links with quantum communication links. More precisely, the quantum
processors can manipulate local states consisting of an unbounded number of qubits with arbitrary unitary transformations, while the communication links are quantum communication channels, and the local outputs can be the result of any quantum measurement.
As in the deterministic \local model, adjacent nodes can exchange any number of
qubits.
We refer to the randomized \local model of computing as \emph{classical \local}, as opposed to the \qlocal model of computing.
Akin to the classical \local model, also in \qlocal we ask that a problem is solved with probability at least \(1-1/\poly(n)\).
This is just an intuitive definition of the model which is sufficient for our purposes, but the interested reader can find  a formal definition in \cite{gavoille2009}.
\section{Quantum advantage: high level ideas}\label{sec:high-level-lcl}
In this section, we provide a high-level explanation of how the LCL problem $\Pi$ that exhibits quantum advantage is defined, and we explain why its quantum and classical complexities differ. All the formal details are deferred to \Cref{sec:lcl-details}.

\paragraph{Proper instances.} In \Cref{ssec:key-ideas-separation}, we explained how we would like a graph to be in order to have an instance that is easy for a quantum algorithm but hard for any classical algorithm. We call the family of graphs explained in \Cref{ssec:key-ideas-separation} \emph{proper instances}, i.e., these graphs are graphs that can be obtained by starting from an arbitrary (not necessarily connected) graph and then replacing each node, and each edge, with a tree-like construction.

\paragraph{LCLs must be promise-free.} However, a standard LCL must be defined on \emph{any} graph, that is, it must not require the promise that the given input graph satisfies some specific conditions. That is, our LCL problem $\Pi$ must be well-defined even in graphs that do not look at all like proper instances. In particular, we want a problem that, on any graph, even on those that look completely different from a proper instance, should still be solvable in $O(\log n)$ rounds by a quantum algorithm.

\paragraph{Overview of the problem \boldmath $\Pi$.} On a high level, our (promise-free) LCL problem $\Pi$ will be defined such that it requires solving two problems, called $\Pi^{\lbadgraph}$ and $\Pi^{\lpromise}$.
More in detail, in $O(\log n)$ classical deterministic rounds, nodes can produce a labeling that marks \emph{invalid} parts of the graph, that is, parts of the graph that do not look like subgraphs of a proper instance. On the other hand, the problem will be defined such that marking as invalid the parts of the graph that are actually valid is not possible. We will formally define this labeling as an LCL called $\Pi^{\lbadgraph}$. Then, the problem $\Pi$ will require to solve some other problem $\Pi^{\lpromise}$ in parts of the graph that are not marked.

\paragraph{How we define \boldmath $\Pi^{\lpromise}$.} The problem $\Pi^{\lpromise}$ will satisfy the following properties.
\begin{itemize}[noitemsep]
    \item It is an LCL \emph{with promise}, that is, it is guaranteed that the input graph is a proper instance.
    \item In proper instances, the complexity is $O(\log n)$ for quantum algorithms but $\omega(\log n)$ for classical randomized algorithms.
\end{itemize}

\paragraph{How we define \boldmath $\Pi^{\lbadgraph}$.}
The problem $\Pi^{\lbadgraph}$ will satisfy the following properties.
\begin{itemize}[noitemsep]
    \item This problem is \emph{promise-free}, that is, any input graph is allowed.
    \item There is a special output label for $\Pi^{\lbadgraph}$ called $\bot$.
    \item In $O(\log n)$ classical deterministic rounds it is possible to produce a solution for $\Pi^{\lbadgraph}$ such that the subgraph induced by nodes labeled $\bot$ is a proper instance.
\end{itemize}

\paragraph{How we define \boldmath $\Pi$.}
Our LCL problem $\Pi$ will then be defined as follows.
\begin{itemize}[noitemsep]
    \item Nodes need to solve $\Pi^{\lbadgraph}$.
    \item On the subgraph induced by nodes that output $\bot$ for $\Pi^{\lbadgraph}$, nodes need to solve $\Pi^{\lpromise}$.
\end{itemize}

\paragraph{Quantum upper bound.}
Then, we will show that our problem $\Pi$ can be solved in $O(\log n)$ quantum rounds, as follows.
\begin{itemize}
    \item Nodes spend $O(\log n)$ classical deterministic rounds to mark invalid parts of the graph, guaranteeing that the subgraph induced by nodes labeled $\bot$ is a proper instance.
    \item Nodes spend $O(\log n)$ quantum rounds to solve $\Pi^{\lpromise}$ in the subgraph induced by nodes labeled $\bot$.
\end{itemize}

\paragraph{Classical lower bound.}
Finally, we will show that our problem $\Pi$ requires $\omega(\log n)$ classical randomized rounds, as follows.
\begin{itemize}
    \item As a lower bound graph, we consider a graph that is a proper instance. This implies that, on this graph, any solution for $\Pi^{\lbadgraph}$ must label all nodes $\bot$.
    \item The (proper) instance is constructed as follows. Let $G$ be a lower bound graph for $\iterghz(\Delta)$, where $\Delta$ is an appropriately chosen degree, as a function of $n$. We replace nodes and edges with tree-like constructions, as follows.
    \begin{itemize}
        \item Tree-like constructions used to replace \emph{nodes} will have an appropriate height to guarantee that they have $\Delta$ leaves.
        \item Tree-like constructions used to replace \emph{edges} will have height $\Theta(\log n)$.
    \end{itemize}
    \item We show that if $\Pi$ can be solved in $O(\log n)$ randomized classical rounds on these graphs, then we reach a contradiction with the lower bound for $\iterghz(\Delta)$. For this purpose, we will exploit the fact that, informally, each edge has been stretched by a $\Theta(\log n)$ factor.
\end{itemize}

\subsection{The tree-like gadget}
A basic building block for our construction is a structure called \emph{tree-like gadget}, that has already been used in different ways in several works related to LCLs (see, e.g., \cite{balliu20lcl-randomness,balliu2024shared-randomness,balliu21lcl-congest}). An example of tree-like gadget is shown in \Cref{fig:treelike}. In \Cref{sec:tree-like gadget}, we will formally define tree-like gadgets and state their properties.

\paragraph{Local checkability of tree-like gadgets.}
Informally, a tree-like gadget is a perfect binary tree in which all nodes that are at the same depth are connected via a path. It is known from previous work that this structure is \emph{locally checkable}, that is, there exist a set of labels $\Sigma^{\ltreelike}$ and a set of constraints $C^{\ltreelike}$ (checkable by inspecting the $O(1)$-radius neighborhood of each node) satisfying the following properties.
\begin{itemize}
    \item Any tree-like gadget can be labeled with labels from  $\Sigma^{\ltreelike}$  such that the constraints $C^{\ltreelike}$ are satisfied on all nodes.
    \item For any connected graph that is not a tree-like gadget, and for any labeling from $\Sigma^{\ltreelike}$, there exists at least one node for which the constraints $C^{\ltreelike}$ are not satisfied.
\end{itemize}

\begin{figure}[t]
	\centering
	\includegraphics[page=1,width=0.7\textwidth]{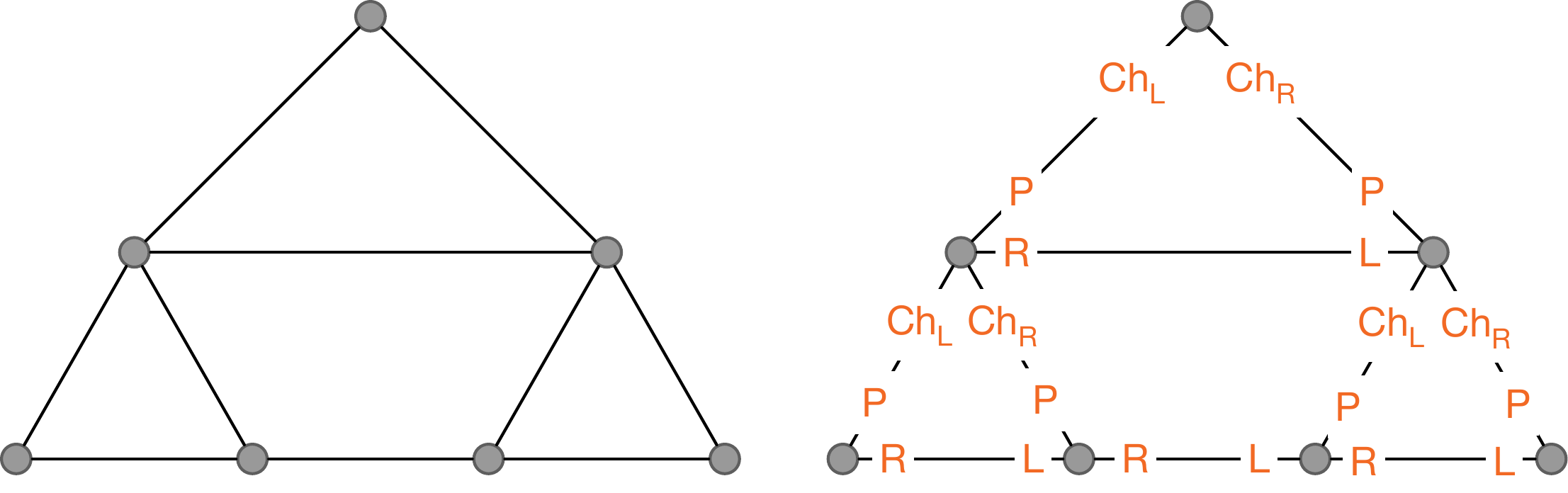}
	\caption{On the left, a tree-like gadget. On the right, a tree-like gadget labeled with labels from $\Sigma^{\ltreelike}$ such that the constraints $\CC^{\ltreelike}$ are satisfied.}
	\label{fig:treelike}
\end{figure}

\paragraph{The LCL problem \boldmath $\Pi^{\lbadtree}$.}
While local checkability is a very useful property, in order to define our problem we need an even stronger property: if there is an error \emph{somewhere in the graph} we want \emph{all nodes to be able to quickly detect the error}. In previous works, it is shown that tree-like gadgets indeed satisfy this stronger property. More in detail, there exists an LCL problem $\Pi^{\lbadtree}$ satisfying the following properties.
\begin{itemize}
    \item There is a special output label for $\Pi^{\lbadtree}$, called $\bot$.
    \item For any connected graph $G$, there exists a solution for $\Pi^{\lbadtree}$ satisfying that, if $G$ is not a tree-like gadget, then no node of $G$ is labeled $\bot$.
    \item Such a solution can be computed in $O(\log n)$ deterministic classical rounds.
\end{itemize}
On a high level, the problem $\Pi^{\lbadtree}$ is defined as follows.
\begin{itemize}
    \item Labels from $\Sigma^{\ltreelike}$ are an input for the problem $\Pi^{\lbadtree}$.
    \item Nodes that do not satisfy the constraints $C^{\ltreelike}$ can directly output an error.
    \item It is possible to prove that, if there is an error somewhere, all other nodes must see an error within distance $O(\log n)$, and $\Pi^{\lbadtree}$ allows outputting \emph{pointers} to prove that there is an error somewhere.
    \item The constraints on the output of $\Pi^{\lbadtree}$ are carefully defined so that, on tree-like gadgets that are properly labeled with labels from $\Sigma^{\ltreelike}$ (i.e., the constraints $C^{\ltreelike}$ are satisfied on all nodes), \emph{no cheating is allowed}, that is, it is not possible to output pointers in some carefully crafted way that gives a valid output for $\Pi^{\lbadtree}$.
    \item Nodes that do not output errors or pointers must output $\bot$.
\end{itemize}
For reasons that will become clear later, nodes will have an additional input for $\Pi^{\lbadtree}$, which is either $0$, i.e., unmarked nodes, or $1$, i.e., marked nodes, and $\Pi^{\lbadtree}$ is defined such that marked nodes are always allowed to output an error.
An example of valid solution of $\Pi^{\mathrm{badtree}}$ is depicted in \Cref{fig:badtree-solution}.
\begin{figure}[t]
	\centering
	\includegraphics[page=2, width=0.4\textwidth]{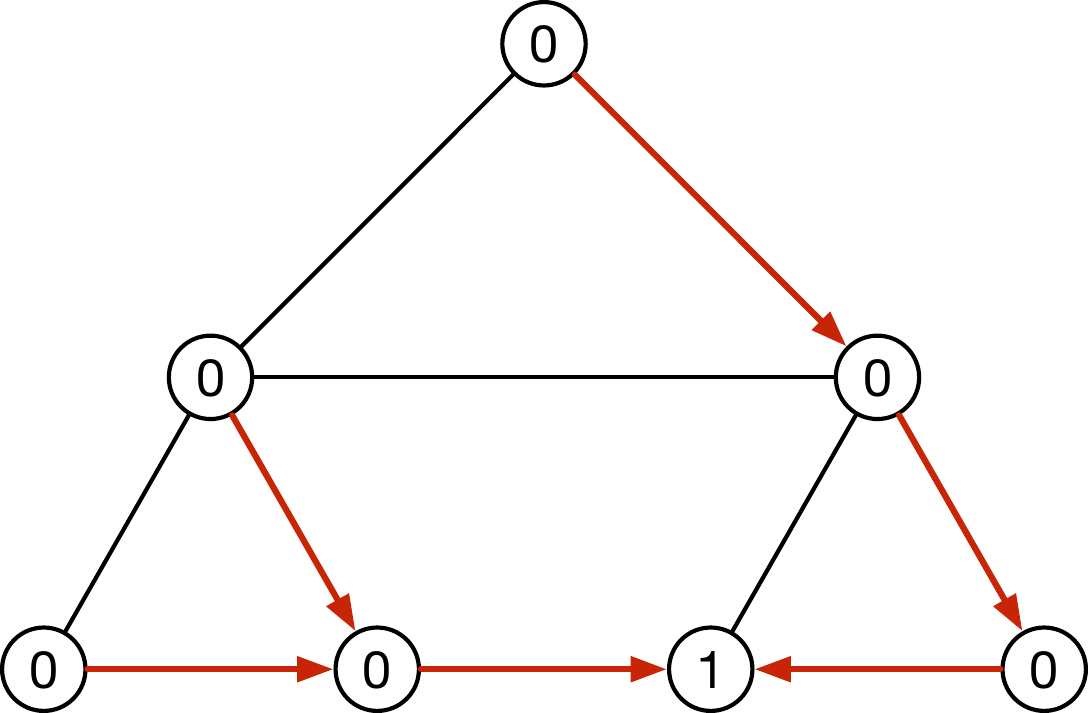}
	\caption{An example of solution of $\Pi^{\lbadtree}$ in which, since one node is marked (i.e., it has input $1$), all nodes output something different from $\bot$ (i.e., they all output \emph{pointers}). Informally, the problem $\Pi^{\lbadtree}$ allows outputting pointers, and in order to not allow cheating (i.e., by creating local cycles with pointers), some specific sequences of pointers are forbidden in the definition of $\Pi^{\lbadtree}$. For example, a pointer going up or down is not allowed to appear after a pointer going left or right, and pointing to the left child is always forbidden.}
	\label{fig:badtree-solution}
\end{figure}

\paragraph{How we will use \boldmath $\Pi^{\lbadtree}$.}
We will define a more complex gadget, called \emph{octopus}, which consists of many tree-like structures connected in a specific way. We will define a problem $\Pi^{\lbadgadget}$ that allows to mark invalid octopus gadgets. In the definition of $\Pi^{\lbadgadget}$, we will use $\Pi^{\lbadtree}$ as a subproblem. Informally, our final problem $\Pi$ will satisfy that, if the adversary really wants nodes to solve $\Pi^{\lpromise}$ (i.e., to make the problem non-trivial), then it needs to properly label tree-like gadgets and octopus gadgets, so that the output for the problems $\Pi^{\lbadtree}$ and $\Pi^{\lbadgadget}$ needs to be $\bot$ everywhere.

\subsection{The octopus gadget}
In \Cref{sec:octopus gadget}, we will formally define the notion of \emph{octopus gadget}, and we will prove some useful properties of this object.
Informally, an octopus gadget is a graph that can be obtained as follows. Let $\Delta$ be an integer power of $2$. Start with a tree-like gadget $G_0$ that has $\Delta$ leaves, and with additional $\Delta$ tree-like gadgets $G_1,\ldots, G_\Delta$. For each $1 \le i \le \Delta$, add an edge between the $i$-th leaf of the gadget $G_0$ and the root of the gadget $G_i$. The tree-like gadget $G_0$ is called \emph{head} gadget, while all the other tree-like gadgets are called \emph{port} gadgets. The formal definition of octopus gadget in \Cref{sec:octopus gadget} will be more general: to each leaf of the head gadget we will connect either $1$ or $2$ port gadgets. In this way, we will be able to handle also cases in which $\Delta$ is not a power of $2$.
An example of octopus gadget is depicted in \Cref{fig:octopus-labeled}.

\begin{figure}[t]
	\centering
	\includegraphics[page=3,width=0.5\textwidth]{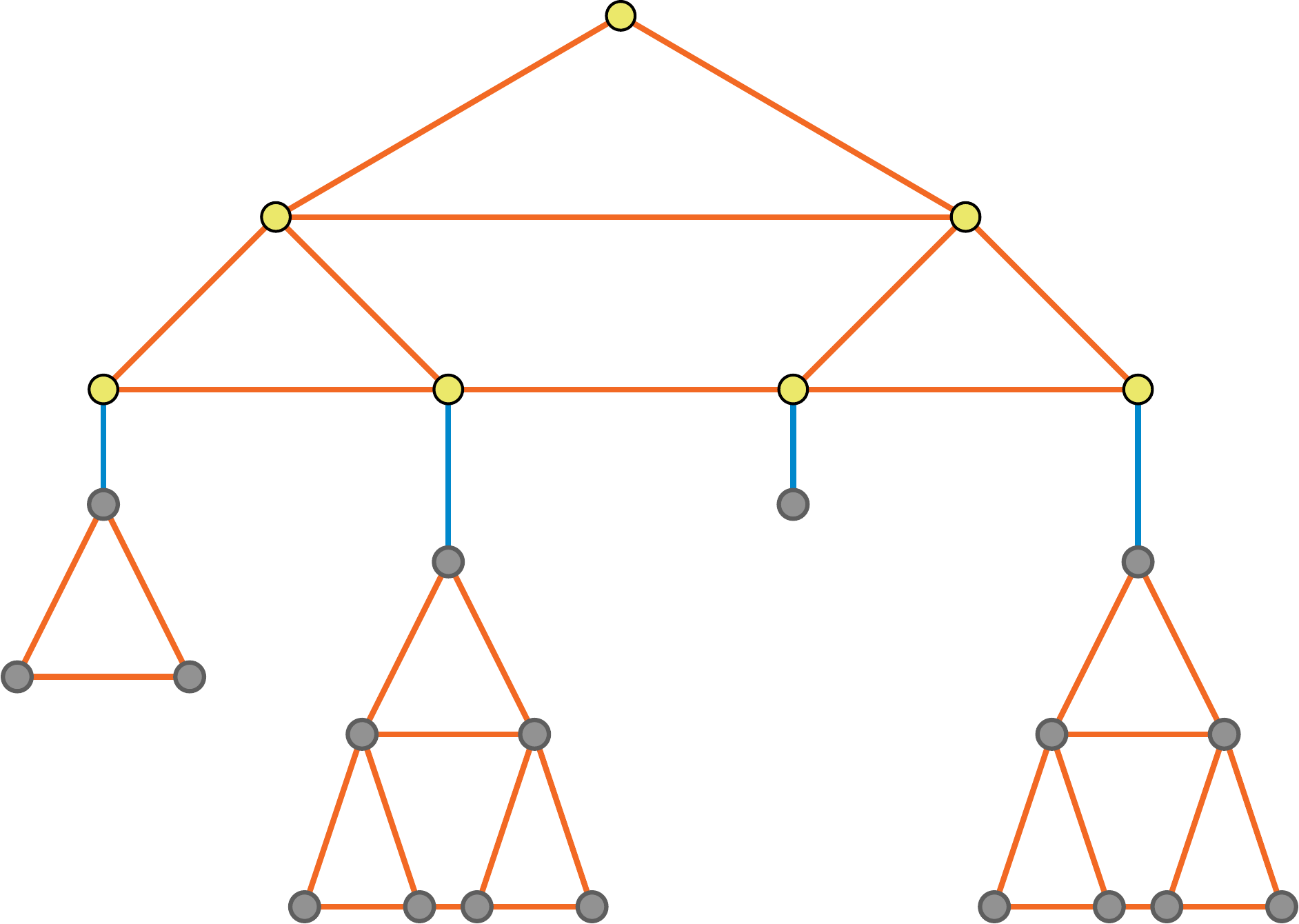}
	\caption{An example of octopus gadget. The connected component induced by yellow nodes and their incident orange edges is the head gadget. Each connected component induced by gray nodes and their incident orange edges is a port gadget. Blue edges properly connect these tree-like gadgets to create a proper octopus gadget. Nodes and edges will be labeled such that this structure is locally checkable. For example, blue edges will have special labels to denote that they do not belong to the tree-like gadgets.}
	\label{fig:octopus-labeled}
\end{figure}

\paragraph{Local checkability of octopus gadgets.}
We will prove that, similarly to the case of tree-like gadgets, also octopus gadgets are locally checkable. More in detail, we will prove that there exists a set of input labels $\Sigma^{\lgadget}$ and a set of constraints $C^{\lgadget}$ over these labels
satisfying the following properties.
\begin{itemize}
    \item Any octopus gadget can be labeled with labels from  $\Sigma^{\lgadget}$  such that the constraints $C^{\lgadget}$ are satisfied on all nodes.
    \item For any connected graph that is not an octopus gadget, and for any labeling from $\Sigma^{\lgadget}$, there exists at least one node for which the constraints $C^{\lgadget}$ are not satisfied.
\end{itemize}
In order to define $\Sigma^{\lgadget}$ and $C^{\lgadget}$, we will extend $\Sigma^{\ltreelike}$ and $C^{\ltreelike}$, as follows.
\begin{itemize}
    \item Each edge will be additionally labeled to denote whether the edge is part of a tree-like gadget, i.e., \emph{internal edge}, or if it is an edge connecting different tree-like gadgets, i.e., \emph{external edge}.
    \item Each node will be additionally labeled to denote whether it is part of the head gadget or if it is part of a port gadget. 
    \item $C^{\lgadget}$ will be defined such that, in the subgraph induced by internal edges, the constraints $C^{\ltreelike}$ must hold, and such that nodes belonging to the same tree-like gadget must either be all labeled head or all labeled port.
    \item $C^{\lgadget}$ will contain additional constraints to make sure that external edges are properly connected.
\end{itemize}

\paragraph{The LCL problem \boldmath $\Pi^{\lbadgadget}$.}
As in the case of tree-like gadgets, we need the stronger property that, if there is some error somewhere, then all nodes are able to prove this fact in $O(\log n)$ deterministic classical rounds, such that no cheating is allowed (i.e., in properly labeled octopus gadgets, the output must be $\bot$ everywhere). For this purpose, we define an LCL problem $\Pi^{\lbadgadget}$. On a high level, this problem is defined by requiring nodes to solve  $\Pi^{\lbadtree}$ \emph{multiple times}, each time having potentially more \emph{marked} nodes (and hence allowed to directly output \emph{error}). More in detail, we define $\Pi^{\lbadgadget}$ as follows.
\begin{itemize}
    \item Nodes that do not satisfy some of the constraints of $C^{\lgadget}$ are \emph{marked}.
    \item Nodes need to solve $\Pi^{\lbadtree}$ in the subgraph induced by internal edges. Observe that it could be the case that some node is in a valid tree-like gadget, but some constraint of $C^{\lgadget}$ is not satisfied on the node (e.g., because some external edge is wrongly connected). In this case, since the node is marked, it is allowed to produce an error when solving $\Pi^{\lbadtree}$. Let $O_1$ be the output labeling produced by the nodes at this point.
    \item Nodes that are labeled to be part of a \emph{head} gadget need to solve $\Pi^{\lbadtree}$ \emph{again}. This time, however, a node is marked also if it has an incident external edge connecting it to a node that did not output $\bot$ in $O_1$. Let $O_2$ be the output labeling produced by the nodes at this point.
    \item Nodes that are labeled to be part of a \emph{port} gadget need to solve $\Pi^{\lbadtree}$ \emph{again}. Similarly as before, this time, a node is marked also if it has an incident external edge connecting it to a node that did not output $\bot$ in $O_2$.
\end{itemize}
Thanks to our triple usage of $\Pi^{\lbadtree}$, we will be able to handle bad cases like the following. Assume that there is a broken port gadget connected to a valid head gadget, and all other port gadgets connected to the head gadget are valid. In this case, we want $\Pi^{\lbadgadget}$ to allow a solution in which no node outputs $\bot$. Indeed, $\Pi^{\lbadgadget}$  allows the following solution.
\begin{itemize}
    \item In the first instance of $\Pi^{\lbadtree}$, all nodes part of the broken port gadget can output something different from $\bot$, while all other nodes output $\bot$.
    \item In the second instance of $\Pi^{\lbadtree}$, the node of the head gadget that is connected to the broken gadget will be marked, and hence all nodes of the head gadget can output something different from $\bot$.
    \item In the third instance of  $\Pi^{\lbadtree}$, all the roots of the port gadgets will be marked, and hence all nodes can output something different from $\bot$.
\end{itemize}
In fact, we will prove that in \emph{all} graphs that are not a valid octopus gadget, nodes can compute a solution in $O(\log n)$ deterministic classical rounds such that no node outputs $\bot$, while in valid octopus gadgets, we still maintain the property that the only allowed output is $\bot$ everywhere. \Cref{fig:invalid-octopus} shows an example of a solution of $\Pi^{\lbadgadget}$ in a broken octopus gadget.

\begin{figure}[t]
	\centering
	\includegraphics[page=5,width=\textwidth]{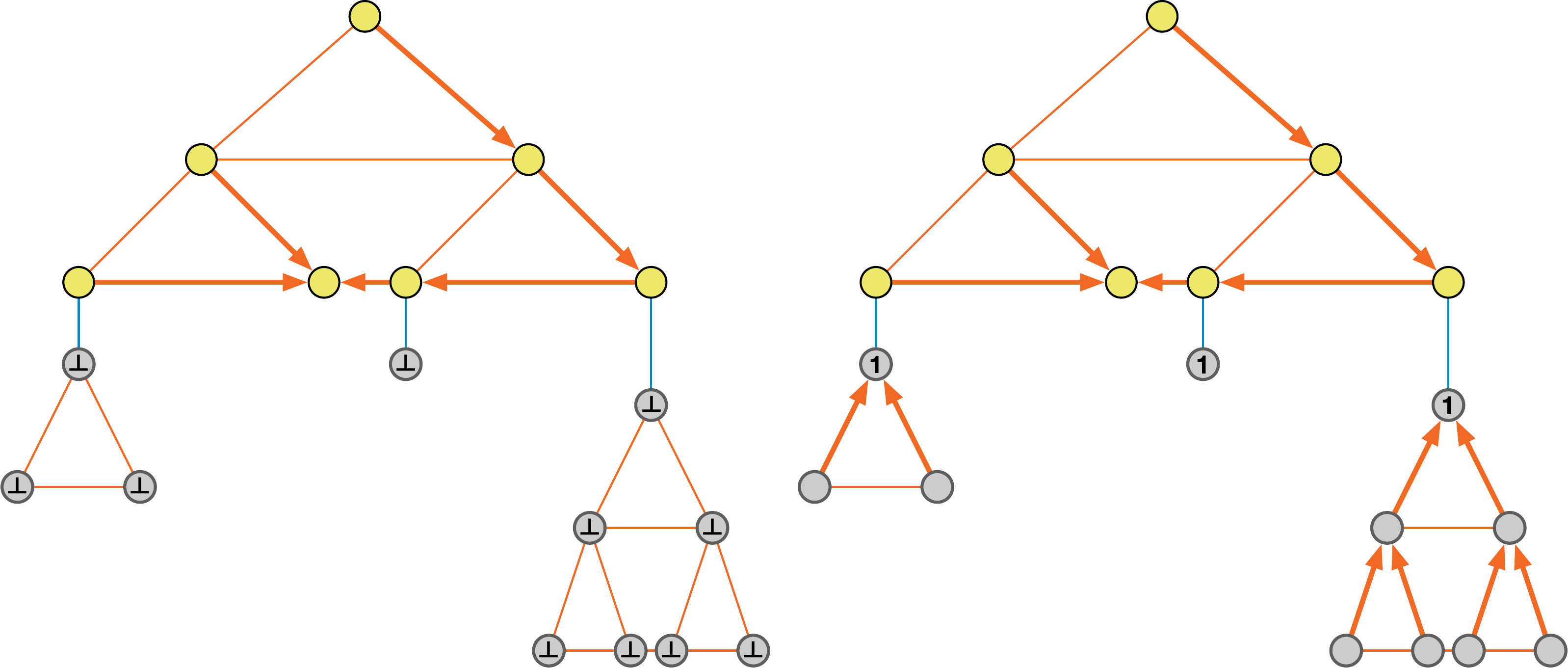}
	\caption{This graph contains a yellow node that has no children (i.e., no incident edges labeled $\llch$ and $\lrch$) and is missing a blue edge, i.e., this graph is not a valid octopus gadget. On the left, the output of the nodes in the first instance of $\Pi^{\lbadtree}$ is shown, where port gadgets do not see any error and output $\bot$, but the nodes that are part of the broken head gadget prove that there is some error by using pointers. On the right, the output of the nodes in the next instance of $\Pi^{\lbadtree}$ is shown, where now the roots of the port gadgets are marked, and hence all the nodes in the port gadgets can now prove that there is an error.}
	\label{fig:invalid-octopus}
\end{figure}

Similarly as in the case of $\Pi^{\lbadtree}$, also in the case of $\Pi^{\lbadgadget}$ nodes will have an additional input for $\Pi^{\lbadgadget}$, which is either $0$, i.e., unmarked nodes, or $1$, i.e., marked nodes, and $\Pi^{\lbadgadget}$ is defined such that nodes that are marked for $\Pi^{\lbadgadget}$ are also marked for $\Pi^{\lbadtree}$. We will use this fact later when defining $\Pi^{\lbadgraph}$.

\subsection{The family of proper instances}\label{sec:intro_proper_instances}
In \Cref{sec:graph family}, we will formally define the family of \emph{proper instances}. Informally, a proper instance is a graph $G'$ that can be obtained as follows. Let $G = (U,V,E)$ be a bipartite graph that is not necessarily connected and that could contain parallel edges. Assume each node $u \in U$ has a degree that is an integer power of $2$ (we will get rid of this assumption, and only require $u$ to not have degree $0$, in \Cref{sec:graph family}). For each node $u \in U$, we put an octopus gadget $g(u)$ in $G'$, and for each node $v \in V$, we put a node $f(v)$ in $G'$. These latter types of nodes are called \emph{inter-octopus nodes}. Then, let $\{u,v\}$ be the $i$-th edge incident to $u$, for some given order. We add an edge in $G'$ between the left-most leaf of the $i$-th port gadget of $g(u)$ and $f(v)$. We call these edges \emph{inter-octopus edges}. \Cref{fig:proper instance} shows an example of a proper instance.

\begin{figure}[t]
	\centering
	\includegraphics[page=4,width=0.7\textwidth]{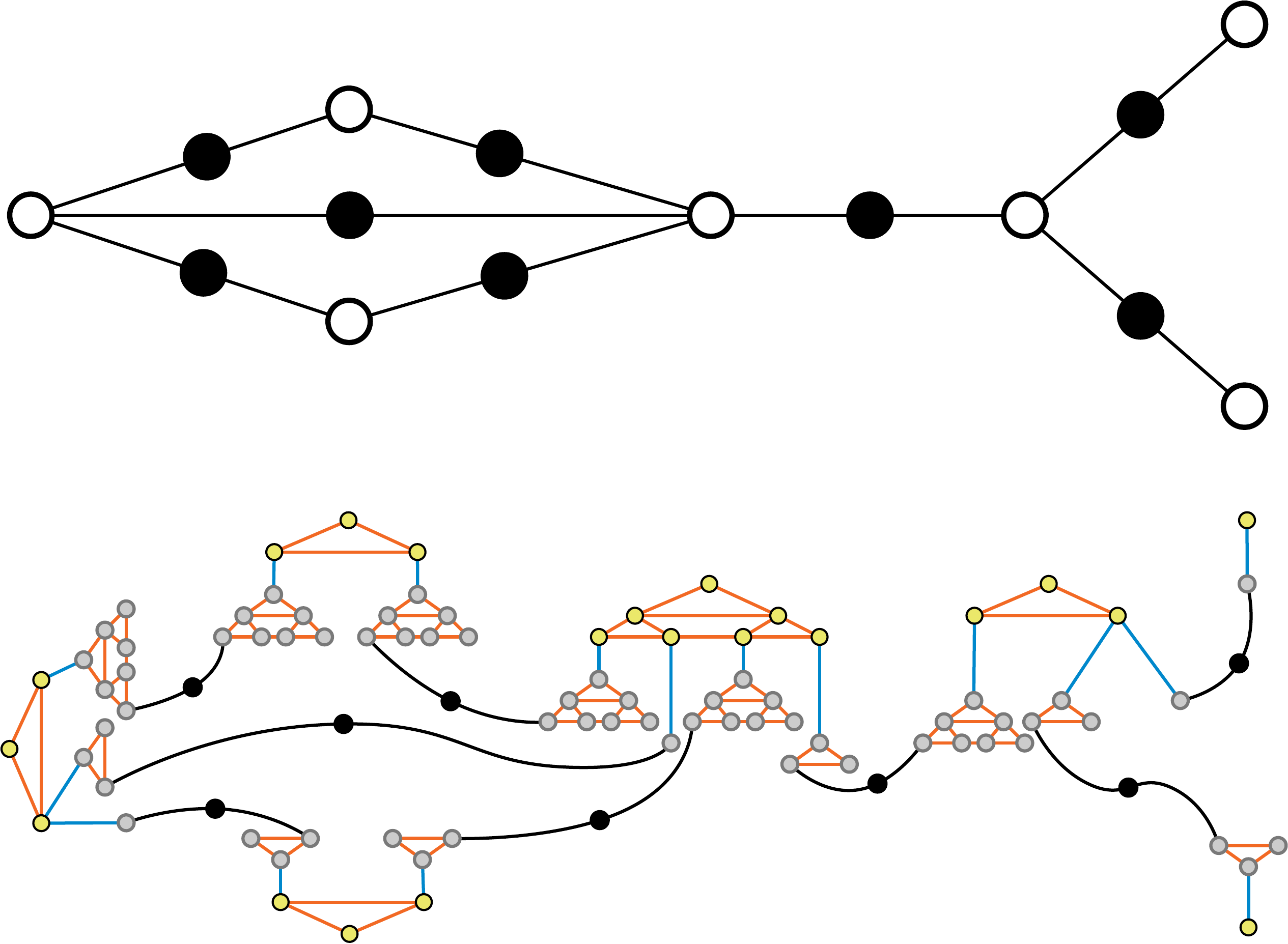}
	\caption{A graph $G$ (above) and a possible proper instance $G'$ (below) that can be constructed starting from $G$. Black nodes in $G'$ correspond to inter-octopus nodes, and black edges correspond to inter-octopus edges.}
	\label{fig:proper instance}
\end{figure}

\paragraph{Local checkability of proper instances.}
Similarly as in the case of tree-like gadgets and octopus gadgets, we will define a set of input labels $\Sigma^{\lproper}$ and a set of constraints $C^{\lproper}$ satisfying the following properties.
\begin{itemize}
    \item Any proper instance can be labeled with labels from  $\Sigma^{\lproper}$  such that the constraints $C^{\lproper}$ are satisfied on all nodes.
    \item For any connected graph that is not a proper instance, and for any labeling from $\Sigma^{\lproper}$, there exists at least one node for which the constraints $C^{\lproper}$ are not satisfied.
\end{itemize}
Similarly as in the case of octopus gadgets, in order to define $\Sigma^{\lproper}$ and $C^{\lproper}$, we will extend $\Sigma^{\lgadget}$ and $C^{\lgadget}$. For example, we will use additional labels for labeling inter-octopus edges. These edges will have additional constraints, to ensure that they are properly connected to octopus gadgets and to inter-octopus nodes.

\paragraph{The LCL problem \boldmath $\Pi^{\lbadgraph}$.}
Similarly as in the case of tree-like gadgets and octopus gadgets, we would like nodes to be able to quickly prove that the graph is not a proper instance, if that is the case. Unfortunately, this seems to be not possible. However, a weaker guarantee is sufficient. We will define a problem $\Pi^{\lbadgraph}$ satisfying the following properties.
\begin{itemize}
    \item There is a special output label for $\Pi^{\lbadgraph}$, called $\bot$.
    \item For any graph $G$, there exists a solution for $\Pi^{\lbadgraph}$ satisfying that the subgraph induced by nodes labeled $\bot$ is a proper instance.
    \item Such a solution can be computed in $O(\log n)$ deterministic classical rounds.
\end{itemize}
This problem will be defined by first marking nodes that do not satisfy $C^{\lproper}$, and then requiring the nodes to solve $\Pi^{\lbadgadget}$ on the subgraph induced by nodes that are not inter-octopus. Moreover, there will be a special output allowed for inter-octopus nodes that do not satisfy $C^{\lproper}$.

\subsection{The LCL problem \texorpdfstring{\boldmath $\Pi$}{Pi}}
We will formally define our LCL problem $\Pi$ in \Cref{sec:lcl-problem}. Informally, it is defined as follows.
\begin{itemize}
    \item The input given to $\Pi$ becomes an input for the problem $\Pi^{\lbadgraph}$.
    \item Nodes need to solve $\Pi^{\lbadgraph}$.
    \item Let $G'$ be the subgraph induced by the nodes outputting $\bot$. On $G'$, nodes need to solve $\Pi^{\lpromise}$.
\end{itemize}
In other words, on a high level, we use $\Pi^{\lbadgraph}$ to transform the promise-problem $\Pi^{\lpromise}$ into a promise-free LCL problem $\Pi$. That is, all the machinery that we have defined until this point are what we use to make $\Pi^{\lpromise}$ promise-free. More in detail, everything is defined so that on proper instances, the only valid solution for $\Pi^{\lbadgraph}$ is the one assigning $\bot$ to all nodes, and hence it is required to solve $\Pi^{\lpromise}$ on all nodes. Hence, a worst-case instance for $\Pi^{\lpromise}$ can be translated into a worst-case instance for $\Pi$. 

In order to make $\Pi^{\lproper}$ promise-free, we pay a cost: it takes $O(\log n)$ rounds to mark invalid parts of the graph, and hence to restrict to a subgraph that is a proper instance. 
Since this cost is still lower than the randomized classical complexity of $\Pi^{\lpromise}$, we are still able to prove a promise-free separation between quantum and classical randomized.

\subsection{The problem \texorpdfstring{\boldmath $\Pi^\lpromise$}{Pi promise}}\label{sec:higl_level-pi_promise}
Before giving an overview of the problem $\Pi^\lpromise$, we first discuss a family of problems that can be used, as a starting point, in order to define $\Pi^\lpromise$. In fact, our construction does not only work for a specific problem, but it works for an entire \emph{family} of problems that we call \emph{linearizable} problems. A linearizable problem $\Pi^{\linearizable} = (\Sigma, (F,L,P),B)$ is defined as follows.
\begin{definition}[Linearizable problem]\label{def:Pi_linearizable}
Let $H$ be a hypergraph, and let $G$ be its bipartite incidence graph. Let the nodes of $G$ corresponding to the nodes of $H$ be called \emph{white} nodes, and let the nodes of $G$ corresponding to the hyperedges of $H$ be called \emph{black} nodes. 
\begin{itemize}
    \item The task is to label each edge of $G$ with a label from some finite set $\Sigma$.
    \item There is a list of allowed \emph{black node configurations} $B$, which is a list of multisets of labels from $\Sigma$ that describes valid labelings of edges incident on a black node. We say that a black node \emph{satisfies the black constraint} if its incident edges are labeled in a valid way. It is assumed that the rank of $H$, and hence the maximum degree of black nodes, is a constant. 
    \item Constraints on white nodes are described as a triple $(F,L,P)$, where $F$ (which stands for \emph{first}) and $L$ (which stands for \emph{last}) are finite sets of labels, and $P$ (which stands for \emph{pairs}) is a finite set of ordered pairs of labels. In this formalism, it is assumed that an ordering on the incident edges of a white node is given, and it is required that:
    \begin{itemize}
        \item The first edge is labeled with a label from $F$;
        \item The last edge is labeled with a label from $L$;
        \item Each pair of consecutive edges must be labeled with a pair of labels from $P$.
    \end{itemize}
    We say that a white node \emph{satisfies the white constraint} if its incident edges are labeled in a valid way.
\end{itemize}
When solving a linearizable problem in the distributed setting, it is assumed that each node knows whether it is white or black.
\end{definition}

\begin{figure}
	\centering
	\includegraphics[page=6,width=0.9\textwidth]{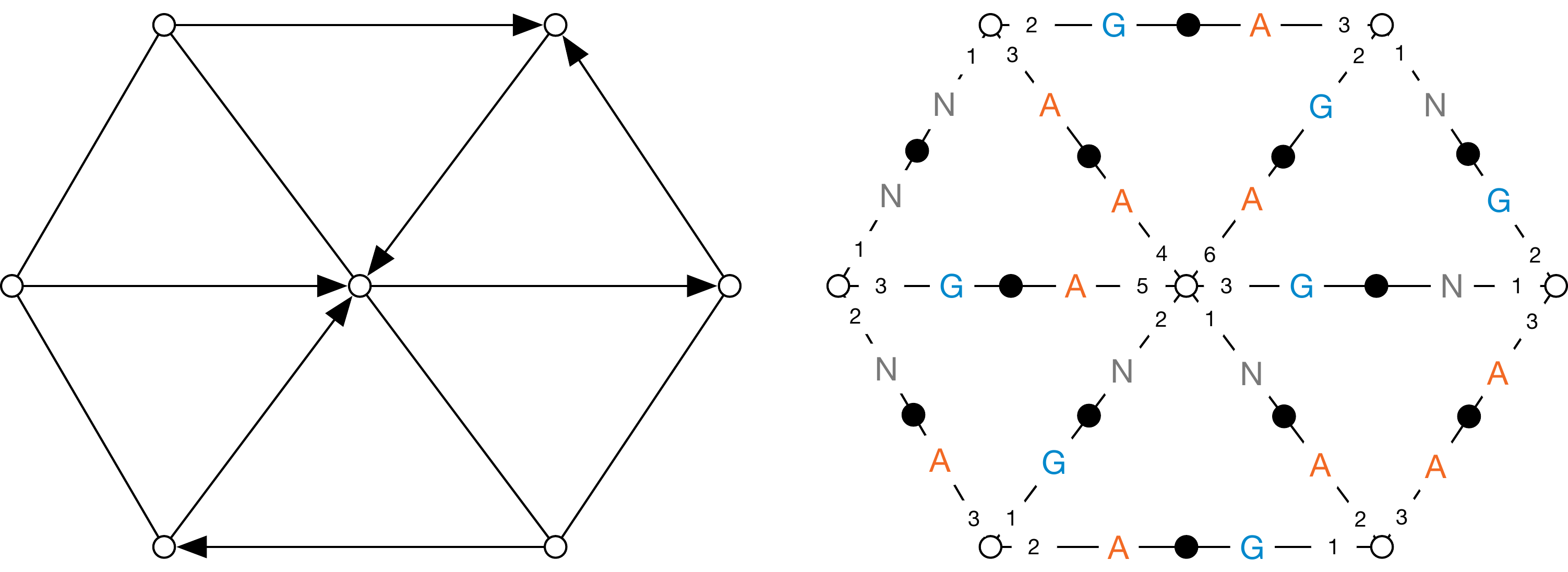}
	\caption{In this example we have hyperedges of rank 2, i.e., black nodes have degree exactly 2. We consider the problem of \emph{edge grabbing}, where we need to orient some edges such that each white node has exactly one outgoing (or grabbed) edge. On the left, it is depicted an example of a solution for this problem. On the right, we show the same solution encoded as a solution for an equivalent linearizable problem. The label $G$ stands for ``grabbed'', the label $N$ stands for ``non grabbed'', the label $A$ stands for ``already grabbed some previous edge''. The constraints on black nodes are $\{\{N, N\}, \{N, A\}, \{N, G\}, \{G, A\}, \{A, A\}\}$, that is, we do not allow white nodes to grab the same edge (i.e., we forbid $\{G,G\}$). The constraints on white nodes are determined by the triple $(F, L, P)$, where $F=\{N, G\}$, $L=\{G, A\}$, $P=\{(N, N), (N, G), (G, A), (A, A)\}$. For example, the configuration of the central white node, ordered according to the ordering of its incident edges, is $(N, N, G, A, A, A)$, where $N\in F$, $A\in L$, and each consecutive pair of labels is in $P$.}
	\label{fig:example-linearizable}
\end{figure}

An example of a definition of a problem using this formalism is shown in \Cref{fig:example-linearizable}.
The useful feature of this family of problems is that, even if the labels of a node are laid down on a (possibly very long) path, it is still possible to locally check for correctness. This property will be crucial for defining $\Pi^\lpromise$. Moreover, as discussed in \Cref{sec:introduction}, the \emph{family} of problems $\iterghz(\Delta)$ can be described as a \emph{single} linearizable problem. 

We are now ready to discuss $\Pi^\lpromise$. Let $\Pi^{\linearizable} = (\Sigma,(F,L,P),B)$ be a linearizable problem. The problem $\Pi^\lpromise$ is defined as a function of the chosen problem $\Pi^{\linearizable}$, as follows. While in \Cref{sec:lcl-problem-pipromise} we will handle the case in which each leaf of the head gadgets is connected to either $1$ or $2$ port gadgets, in the following we consider the simpler setting in which each leaf of the head gadgets is connected to a single port gadget.
\begin{itemize}
    \item Each port gadget of our construction corresponds to an edge of $G$. We require all nodes of each port gadget to output a label of $\Pi^{\linearizable}$, with the requirement that all nodes of a port gadget are labeled with the same label.
    \item In our construction, nodes corresponding to black nodes of $G$ are marked in a special way, and hence they know to be nodes corresponding to black nodes. Let $b$ be one such node, and let $M$ be the multiset containing the output labels of the neighboring nodes (which are nodes belonging to port gadgets) of $b$. We require $M$ to be a configuration allowed by the black constraint of $\Pi^{\linearizable}$.
    \item Thanks to the input of the problem $\Pi$, leaves of a head gadget know that they are leaves of a head gadget, and additionally, the left-most leaf $v$ and the right-most leaf $u$ know that they are, respectively, the left-most and the right-most leaves. Let $f$ be the output label of the port node connected to $v$, let $\ell$ be the output label of the port node connected to $u$, and let $o_i$ be the output label of the port node connected to the $i$-th leaf. We require that $f \in F$, that $\ell \in L$, and that each pair $(o_i,o_{i+1})$ is in $P$. 
\end{itemize}
Thanks to these constraints, we have that, in a valid labeling, each octopus gadget encodes a labeling that satisfies the white constraint of $\Pi^{\linearizable}$.
We will give all the details about $\Pi^{\lpromise}$ in \Cref{sec:lcl-problem-pipromise}.

\subsection{The complexity of \texorpdfstring{\boldmath $\Pi$}{Pi}}
We now provide an overview on how we will prove lower and upper bounds for the problem $\Pi$. Since our goal is to prove a separation between quantum-LOCAL and classical randomized-LOCAL, our upper bound will be a quantum algorithm, while the lower bound will hold for classical randomized algorithms.

\paragraph{Upper bound.}
In the case of upper bounds, we will assume the existence of a quantum algorithm $\mathcal{A}$ for $\Pi^{\linearizable}$ with some complexity $T$ (for example, in the case of $\iterghz(\Delta)$, we know that $T = 1$), and we will show that, with $O(\log n)$ overhead, it can be converted into an algorithm $\mathcal{B}$ that solves $\Pi$. On a high level, algorithm  $\mathcal{B}$ works as follows.
\begin{itemize}
    \item First, nodes spend $O(\log n)$ deterministic classical rounds to solve $\Pi^{\lbadgraph}$. The result will be a labeling satisfying that nodes labeled $\bot$ induce a proper instance.
    \item Consider the graph $G'$ obtained by contracting each octopus gadget into a single node. The graph $G'$ is an instance of $\Pi^{\linearizable}$, and each quantum communication round on $G'$ can be simulated with $O(\log n)$ overhead in the original graph. Hence, it is possible to solve $\Pi^{\linearizable}$ in the original graph in $O(T \log n)$ quantum rounds.
    \item Once $\Pi^{\linearizable}$ is solved, nodes can spend $O(\log n)$ additional rounds to produce a labeling that is valid for $\Pi^{\lpromise}$ on the subgraph induced by nodes labeled $\bot$ for $\Pi^{\lbadgraph}$.
\end{itemize}
Therefore, algorithm $\mathcal{B}$ solves $\Pi$ in $O(T \log n)$ quantum rounds. The details are given in \Cref{sec:ub-quantum}.

\paragraph{Lower bound.}
Suppose we know that the problem $\Pi^{\linearizable}$ requires $\Omega(T)$ classical randomized rounds. We prove that, if $\Pi$ can be solved in $o(T \log n)$ classical randomized rounds, then $\Pi^{\linearizable}$ can be solved in $o(T)$ classical randomized rounds, violating the lower bound. This result will be achieved by showing the following. Suppose for a contradiction that there exists an algorithm $\mathcal{B}$ that solves $\Pi$ in $o(T \log n)$ classical randomized rounds, we construct an algorithm $\mathcal{A}$ for $\Pi^{\linearizable}$  as follows.
\begin{itemize}
    \item Let $G$ be an instance of $\Pi^{\linearizable}$. Nodes construct a virtual graph $G'$, which is an instance of $\Pi$, by imagining themselves to be a valid octopus gadget where each port gadget has height $\Theta(\log n)$, and imagining each (hyper)edge to be a black node connecting different gadgets.
    \item In this virtual graph, the only solution for $\Pi^{\lbadgraph}$ is the one assigning $\bot$ to all nodes.
    \item By assumption, the problem $\Pi^{\lpromise}$ can be solved in $T' = o(T\log n)$ rounds in $G'$, and the $T'$-radius neighborhood of each node in $G'$ is contained in $o(T)$-radius neighborhood of each node in $G$. Hence, nodes can solve $\Pi^{\lpromise}$ by communicating for only $o(T)$ rounds on $G$, and then obtain a solution for $\Pi^{\linearizable}$.
\end{itemize}
There are some details that we need to take care of, like the fact that the number of nodes in $G$ and $G'$ are different, and we will show how to handle all these details in \Cref{sec:lb-local}.
\section{Upper bounding quantum advantage}\label{sec:quantum-advantage}
In this section, we prove \cref{thm:intro:sim}.
We start by defining \emph{$T$-dependent distributions} and the \emph{bounded-dependence model} using \emph{outcomes}.
In \cref{ssec:reduction:proof}, we follow the outline of \cref{ssec:intro-sim-ideas} and provide and prove lemmas for each of the intermediate steps.
We then combine these lemmas to prove \cref{thm:intro:sim}.

\paragraph{Preliminary definitions.}
We start by introducing some key definitions.
Let \(G\) and \(H\) be two graphs.
We define the \emph{union} of \(G\) and \(H\) to be the graph \(G \cup H = (V,E)\) where \(V = V(G)\cup V(H)\) and \(E = E(G) \cup E(H)\).
The \emph{intersection} is similarly defined, with \(G \cap H = (V,E)\) where \(V = V(G) \cap V(H)\) and \(E = E(G) \cap E(H)\).
The \emph{graph difference} is defined as \(G \setminus H = (V,E)\) where \(V = V(G)\) and \(E = E(G) \setminus E(H)\); note that here we only take the difference of the edge sets.
We define the \emph{ring neighborhood} between \(T_1\) and \(T_2\) of a node \(u \in V\) (for \(T_1 \le T_2\)) as \(\neighborhood_{T_2}^{T_1}[u] = \neighborhood_{T_2}[u] \setminus \neighborhood_{T_1}[u]\). 
We use an analogous notation for subsets of nodes.
Consider any node \(u \in V\) (or any subset \(S \subseteq V\)): with an abuse of notation, we define the \emph{open induced subgraph} as the graph \(\mathring{G}[\neighborhood_{T}[u]] = G[\neighborhood_{T}[u]] \setminus G[\neighborhood_{T}^{T-1}[u]]\) (or \(\mathring{G}[\neighborhood_{T}[S]] = G[\neighborhood_{T}[S]] \setminus G[\neighborhood_{T}^{T-1}[S]]\)).
In practice, in this definition we are removing from the classical notion of neighborhood the edges that connect nodes that are at distance \(T\) from \(u\) (or \(S\)) as the graph that \(u\) (or \(S\)) \emph{sees} by moving \(T\) hops away does not include them.

For any function \(f : A \to B\) and any subset \(A' \subseteq A\), we denote by \(f \restriction_{A'}\) the function \(g: A' \to B\) such that \(f(x) = g(x)\) for all \(x \in A'\).
For any input distributed network \((G,\localVar)\)  to any problem, where \(\localVar(v)\) encodes input data for all \(v \in V(G)\), and any subset of nodes \(S\subseteq V(G)\), the \emph{radius-\(T\) view} of \(S\) is \(\view_T(S) = (\mathring{G}[\neighborhood_{T}[S]], \localVar \restriction_{\neighborhood_{T}[S]})\).
Basically, \(\view_T(S)\) includes everything that can be \emph{seen} by the nodes in \(S\) with \(T\) rounds of communication, including input data (degree, ports, identifiers and input labels---if any, etc.).
Suppose \(G\) has \(n\) nodes, and fix any subset of nodes \(S \subseteq V(G)\). 
Given any two graphs \(G,H\) with inputs \(\localVar_G, \localVar_H\) and any two subset of nodes \(S_G \subseteq V(G)\) and \(S_H\subseteq V(H)\), it is natural to define the notion of \emph{isomorphism between views}.
We say that
\(\view_T(S_G)\) is isomorphic to \(\view_T(S_H)\) if there exists a function \(\varphi: V(G) \to V(H)\) such that the following holds: 
\begin{enumerate}[noitemsep]
    \item \(\varphi \restriction _{S_G}\) is an isomorphism between \(G[S_G]\) and \(H[S_H]\);
    \item \(\varphi \restriction _{\NN_T[S_G]}\) is an isomorphism between \(\mathring{G}[\NN_T[S_G]]\) and \(\mathring{H}[\NN_T[S_H]]\);
    \item \(\localVar_G(u) = \localVar_H(\varphi(u))\) for all \(u \in \neighborhood_T(S_G)\).
\end{enumerate}

Consider any \(T\)-round (deterministic or randomized) \local algorithm \(\algo\) run by the nodes of \(G\), and any subset of nodes \(S \subseteq V(G)\).
Run \(\algo\) on \(G\) for \(T\) rounds and look at the output distribution over \(S\).
Clearly, such distribution depends only on \(\view_T(S)\), while everything that is outside \(\view_T(S)\) cannot influence the output distribution, even if some adversary changes the graphs and the inputs outside \(\view_T(S)\) (apart from the number of nodes, as this is an input to the algorithm).
The \emph{cause} of outputs in \(S\) \emph{cannot be influenced} by any action of an adversary outside \(\view_T(S)\).

This property is the so-called \emph{no-signaling from the future} principle in physics, which states that no signals can be sent from the future to the past, and is equivalent to the \emph{causality principle} \cite{d2017quantum}. 
Such principle holds in \emph{every physical distributed network} running any kind of synchronous distributed algorithm, including quantum ones.

To see how this principle formally translates in our setting, let us define the notion of \emph{outcome}, which is some kind of generalization of an algorithm.
Here, we restrict to finite sets of input and output labels since we focus on LCL problems.
For a graph \(G\), let us denote by \(\HH(G) \subseteq V(G) \times V(G) \times E(G)\) the set containing all tuple of the form \((v,(v,e))\) where \((v,e) \in V(G) \times E(G)\) and \((v,e)\) is a half-edge incident to \(v\).
Let us define the notion of \emph{labeling function} for a graph \(G\) as follows:
a function \(\oupt_i: \HH(G) \to \Vout \times \Eout \) is a labeling function if, for two different half-edges \((v,e)\), \((v,e')\) of the same node \(v\), \(\oupt(v,(v,e))[1] = \oupt(v,(v,e'))[1]\), where by \(\oupt(v,(v,e))[1]\) we mean the first component of the tuple \(\oupt(v,(v,e))\).
In words, the labeling function must assign to a node \(v\) only one possible label. 
For any node \(v\), the label assigned to \(v\) is any \(\oupt(v,(v,e))[1]\) for any half-edge \((v,e)\) incident to \(v\).
The label assigned to a half-edge \((v,e)\) is \(\oupt(v,(v,e))[2]\).
With an abuse of notation, if \(v\) has no incident edges we still assume that the writing \(\oupt(v)\) is well-defined and only assigns a label from \(\Vout\) to \(v\).

\begin{definition}[Outcome]
    Let \(\Vin,\Ein,\Vout,\Eout\) be finite sets of labels.
    Let \(\II\) be the family of \emph{all} input distributed networks \((G, \localVar)\), where \(G\) is a \((\Vin,\Ein)\)-labeled graph and \(\localVar\) encodes input data (identifiers, ports, possible random bits, input labels, etc.).
    An \emph{outcome} \(\outcome\) is a function that maps every input distributed network \((G, \localVar) \in \II\) to distribution over output labelings \(\{(\oupt_i, p_i)\}_{i \in J}\), where \(\{(\oupt_i, p_i)\}_{i \in J}\) is defined as follows:
    \(J\) is a finite set of indices. 
    The function \(\oupt_i: \HH(G) \to \Vout \times \Eout \) is a labeling function that we call \emph{output labeling function}.
    The value \(p_i\) is the probability that \(\oupt_i\) occurs.
\end{definition}
This definition is easily generalizable to the case of infinite label sets, but we avoid it for the sake of simplicity.
Notice that all synchronous distributed algorithms yield outcomes: it is just the assignment of an output distribution (the one that is the result of the algorithmic procedure) to the input graph.
We remark we have defined the domain set of outcomes to include \emph{all possible graphs}. 
This is not restrictive: Even classical (or quantum) algorithms can run on every possible graph. 
One can just introduce some garbage output label so that whenever a node running an algorithm needs to do something that is not well-defined (given its local view), it can just output the garbage label.

We say that an outcome \(\outcome\) solves a problem \(\problem\) over a family of graphs \(\FF\) with probability \(q > 0\) if, for every \(G \in \FF\) and every input data \(\localVar\), it holds that
\[
    \sum_{\substack{i \in I : \\ \oupt_i \in \problem(G,\localVar)}} p_i \ge q.
\]
Let \(\outcome: (G, \localVar) \mapsto \{(\oupt_i, p_i)\}_{i \in I}\) be any outcome and fix an input \((G, \localVar)\).
Consider any subset of nodes \(S \subseteq V(G)\).
Let \(\HH(G)[S]\) be the subset of \(\HH(G)\) that contains all elements \((v,(v,e)) \in \HH(G)\) where \(v \in S\).
The restriction of the output distribution \(\outcome(G,\localVar) = \{(\oupt_i: \HH(G) \to \Vout \times \Eout, p_i)\}_{i \in I}\) to \(S\) is the distribution 
\(\{(\oupt_j: \HH(G)[S] \to \Vout \times \Eout, p_j')\}_{j \in J}\) such that 
\[
    p_j' = \sum_{\substack{i \in I : \\ \oupt_j = \oupt_i \restriction_{\HH(G)[S]}}} p_i,
\] 
and is denoted by \(\outcome(G,\localVar)[S]\) or \(\{(\oupt_i, p_i)\}_{i \in I}[S]\).
We also say that two labeling distributions \(\{(\oupt_i: \HH(G) \to \Vout \times \Eout, p_i)\}_{i \in I}, \{(\oupt_j: \HH(G') \to \Vout \times \Eout, p_j)\}_{j \in J}\) over two graphs \(G,G'\) are isomorphic if there is an isomorphism \(\varphi: V(G) \to V(G')\) between \(G\) and \(G'\) such that \(\{(\oupt_i: \HH(G) \to \Vout \times \Eout, p_i)\}_{i \in I} = \{(\oupt_j \circ \varphi: \HH(G) \to \Vout \times \Eout, p_j)\}_{j \in J}\).

We are now ready to define \emph{\nonsign outcomes}.
\begin{definition}[Non-signaling outcome]\label{def:non-signaling:non-signaling-outcome}
    Let \(\outcome\) be any outcome.
    Suppose there exists a non-negative integer \(T \ge 0\) with the following property: 
    For any two subsets \(S_G \subseteq V(G), S_H \subseteq V(H)\) such that \(\varphi : V(G) \to V(H)\) is an isomorphism between \(\view_T(S_G)\) and \(\view_T(S_H)\), then the restrictions \(\outcome(H, \localVar_H)[S_H]\) and \(\outcome(G, \localVar_G)[S_G]\) are isomorphic under \(\varphi\).
    In this case, we say that \(\outcome\) is \emph{\nonsign beyond distance \(T\)} or, alternatively, that \(\outcome\) has locality \(T\).
\end{definition}

Running \(T\)-rounds classical or \qlocal algorithms without shared resources, we obtain the following property on the output labeling distribution: for every two subset of nodes \(A,B\) such that \(\dist(A,B) > 2T\), then the output distributions restricted to \(A\) and \(B\) are independent.
Let us formalize this notion.

\begin{definition}[\(T\)-dependent distribution]
    Let \(\Vin,\Ein,\Vout,\Eout\) be finite sets of labels, and \(J\) a finite set of indices.
    Let \(\II\) be the family of \emph{all} input distributed networks \((G, \localVar)\), where \(G\) is a \((\Vin,\Ein)\)-labeled graph and \(\localVar\) contains input data (identifiers, ports, possible random bits, etc.).
    An output labeling distribution \(\{(\oupt_i : \HH(G) \to \Vout \times \Eout, p_i)\}_{i \in J}\) is  \(T\)-dependent if, for any two subsets of nodes \(A,B \subseteq V(G)\) such that \(\dist(A,B) > T\), we have that \(\{(\oupt_i, p_i)\}_{i \in J}[A]\) is independent of 
    \(\{(\oupt_i, p_i)\}_{i \in J}[B]\)
\end{definition}

We can think now of \nonsign outcomes that produce such distributions.

\begin{definition}[Bounded-dependent outcome]\label{def:bounded-dependence:bounded-dependent-outcome}
    Let \(\outcome\) be any outcome that is \nonsign beyond distance \(T\).
    We say that \(\outcome\) is \boundept with locality \(T\) if for any input \((G, \localVar)\) we have that \(\outcome(G,\localVar)\) is \(2T\)-dependent.
    Furthermore, when \(T = \myO{1}\) (i.e., it does not depend on the input graph), we say that \(\outcome(G,\localVar)\) is a finitely-dependent distribution.
    Moreover, if for all inputs \((G,\localVar)\) it holds that \(\outcome(G,\localVar)\) does not depend on identifiers and port numbers, we say that \(\outcome\) is \emph{invariant under subgraph isomorphism}.
\end{definition}

With the addition of this further property, we can define the \emph{\boundep model}, first formalized in \cite{akbari2024}.

\paragraph{The \boundep model.}
The \emph{\boundep model} is a computational model that produces bounded-dependent outcomes.
The required success probability when solving a problem must be at least \(1 - 1/\poly(n)\).

\subsection{Limits on quantum advantage}
\label{ssec:reduction:proof}

With these definitions, we are ready to prove \cref{thm:intro:sim}.
Formally, in this section, we prove the following theorem:
\begin{theorem}
    \label{thm:findep-sim}
    Let $\Pi$ be an LCL problem with checking radius~$r$.
    Let $\outcome$ be a bounded-dependent outcome with locality~$T(n)$ of labelings that solve problem~$\Pi$ with probability $p(n)$.
    Then there exists a \randlocal algorithm solving~$\Pi$ with locality~$O(\sqrt{n T(n)} \poly \log n)$ with probability at least~$p(n)$.
\end{theorem}

Like in the overview, we start by clustering the graph into small well-separated clusters.
Formally, we mean the following:
\begin{definition}[$(\varepsilon, d)$-clustering]
	Given a graph $G$, an $(\varepsilon, d)$-clustering is a partition $V(G) = D \cup S_1, \cup \ldots, S_k$ meeting the following conditions:
	\begin{itemize}[noitemsep]
		\item The distance between two nodes $u \in S_i, v \in S_j$, from different clusters $i \neq j$, is at least 2.
		\item For every $1 \le i \le k$, the diameter $\max_{u,v \in S_i} \dist_G(u,v)$ of $S_i$ is at most $\myO{d}$.
		\item $D$ contains at most $|D| \le \varepsilon |V(G)|$ vertices.
	\end{itemize}
\end{definition}

To actually compute such clustering, we use the following clustering algorithm by \textcite{coiteuxroy2023}:
\begin{theorem}\label{thm:clustering}
	For any $0 < \varepsilon \le 1$, an $(\varepsilon, \frac{\log n \cdot \log \log \log n}{\varepsilon})$-clustering can be computed with locality
	\[
	\myO{\frac{\log (n)^2 \cdot \log (1/\varepsilon)}{\varepsilon}}
	\]
	rounds in the deterministic \local model.
\end{theorem}

By the definition of clustering, we only get that nodes in different clusters are mutually non-adjacent.
However, to ensure that the labelings in each cluster are mutually independent given a labeling for the unclustered nodes, we require nodes of different clusters to be at distance at least~$r$, where $r$ denotes the checkability radius of our LCL problem~$\Pi$.
We can fix this by running the algorithm from \cref{thm:clustering} on the power graph $G^r$.
By doing so and choosing $\varepsilon = (1/ \sqrt{n T(n)})$ we obtain the following corollary.

\begin{corollary}
    \label{cor:nt-clustering}
	A $(1/ \sqrt{n T(n)}, \sqrt{n T(n)} \log^2 n)$-clustering can be computed with locality
	\[
		\myO{\sqrt{n T(n)} \log^3(n)}
	\]
	in the deterministic \local model.
	Additionally, any two nodes $u,v$ from different clusters have distance at least $r$ in $G$ for any pre-chosen constant~$r$.
\end{corollary}

This clustering ensures that the clusters are well-separated enough such that their labelings can be completed independently given a labeling in the unclustered nodes~$D$.
Our next step is to produce a labeling on~$D$ by sampling for the bounded-dependence original distribution.
To be able to do this, we need to partition~$D$ into well-separated partitions whose outputs are independent.
Formally, we do the following:
\begin{lemma}
    \label{lem:d_partition}
    For any subset of nodes $D \subseteq V$ of size at most $O\bigl(\sqrt{n/ T(n)}\bigr)$, we can partition $D$ into $D_1, \ldots, D_h$ such that
    \begin{itemize}[noitemsep]
        \item for every $1\leq i\leq h$, the diameter of $D_i$ is $O\bigl(\sqrt{n T(n)}\bigr)$ in $G$, and
        \item for every $1\leq i<j\leq h$, if $u \in D_i$ and $v \in D_j$ then $\dist_G(u,v)\geq 2T(n)+1$.
    \end{itemize}
    Moreover, this partitioning can be computed in the \local model with locality~$O\bigl(\sqrt{n T(n)}\bigr)$.
\end{lemma}
The latter condition allows us to sample the outputs for each partition independently as, by the definition of bounded-dependence distributions, their outputs are independent.
The former condition makes sure that each partition has a low diameter and hence the sampling can be done with low locality.
Intuitively, the partitioning is done by finding the connected components in the power graph~$G^{2T(n)}$:
\begin{proof}
    Each node in~$D$ iteratively runs the following exploration procedure:
    The node~$v$ looks at its radius-$2T(n)$ neighborhood in the original graph~$G$ to see if there are other nodes that belong to~$D$.
    If not, the node stops as it has found all nodes of~$D$ that are within distance~$2T(n)$ of it, and hence all other nodes must belong to different partitions.
    Otherwise the node adds newly-found nodes to its own partition and repeats this exploration procedure from these newly-found nodes.

    It is clear that this algorithm produces a partitioning that satisfies the latter condition.
    It remains to prove that the algorithm stops within~$O\bigl(\sqrt{n T(n)}\bigr)$ steps.
    Since nodes only start another iteration if they find a new node from the set $D$, this procedure stops after at most $|D|$ iterations.
    As each iteration has locality~$2T(n)$, the total running time of the algorithm is bounded by
    \[
        |D| \cdot 2T(n) = O\left(\sqrt{n / T(n)}\right) \cdot 2T(n) = O\left(\sqrt{n T(n)}\right) .
    \]
    This also implies that the diameter of each cluster is~$O\bigl(\sqrt{n T(n)}\bigr)$.
\end{proof}

We now have the nodes of graph \(G\) clustered into \(D,S_1,\dots,S_k\), and nodes \(D\) partitioned into \(D_1,\dots,D_h\). Let \(\localVar\) be an input for \(G\) and consider a \(2T(n)\)-dependent labeling distribution \(\outcome(G,\localVar)\) over \((G,\localVar)\) that is given by a bounded-dependent outcome \(\outcome\) with locality \(T(n)\). We now show that a randomized \local algorithm can sample a labeling for \(D\) from \(\outcome(G,\localVar)[D]\) if it is given an oracle to \(\outcome\).

\begin{lemma}
    \label{lem:sampling_D}
    Let \(\outcome\) be a bounded-dependent outcome with locality \(T(n)\). There exists a \randlocal algorithm that runs on \((G,\localVar)\) with \(V(G)\) clustered as \(D,S_1,\dots,S_k\), \(D\) partitioned as \(D_1,\dots,D_h\) and an oracle to \(\outcome\), and outputs labelings for \(D\) that follow the distribution \(\outcome(G,\localVar)[D]\). This algorithm has locality \(O(\max_i \diam_G(D_i))\).
\end{lemma}
\begin{proof}
    Since \(\dist_G(D_i,D_j) > 2T(n)\) if \(i \neq j\), for all \(1 \le \ell \le h\), for any sequence \(1 \le i_1 < \ldots < 1_\ell \le h\), the restrictions 
    \(\outcome(G,\localVar)[D_{i_1}], \ldots, \outcome(G,\localVar)[D_{i_\ell}]\) of \(\outcome(G,\localVar)\) to \(D_{i_1}, \ldots, D_{i_\ell}\) are mutually independent by definition of \(2T(n)\)-dependent distribution. 
    Let \(\oupt: \HH(G) \to \Vout \times \Eout\) be any output labeling function. 
    For a given subset of nodes \(S \subseteq V(G)\), 
    let us denote the probability that \(\outcome(G,\localVar)\) actually yields the output \(\oupt\) on \(\HH(G)[S]\) by \(P_{\oupt}(\HH(G)[S])\). 
    For all \(1 \le \ell \le h\), for any sequence \(1 \le i_1 < \ldots < i_\ell \le h\), we have
    \[
        P_{\oupt}(\HH(G)[D_{i_1}]) \cdot \ldots \cdot P_{\oupt}(\HH(G)[D_{i_\ell}]) = P_{\oupt}(\HH(G)[D_{i_1} \cup \ldots \cup D_{i_\ell}]).
    \]
    It follows that, for \(1\le i\le h\), a \randlocal algorithm can gather the subgraph induced by \(D_i\) and sample a labeling from distribution \(\outcome(G,\localVar)[D_i]\) independently from the other \(D_j\), \(i\neq j\). The equality above ensures that the final labeling obtained in this way for \(D\) is sampled with the same probability as if we sampled it from \(\outcome(G,\localVar)[D]\).
    As the algorithm samples the labeling for each~$D_i$ independently, it is obvious that the locality of the algorithm is bounded by $O(\max_i \diam_G(D_i))$.
\end{proof}

In other words, the output labels of nodes belonging to different partitions can be sampled independently while respecting the locality of the bounded-dependent distribution. 
This implies that, by sampling a function \(\oupt_D: \HH(G)[D] \to \Vout \times \Eout\) from the bounded-dependent outcome \(\outcome(G,\localVar)[D]\), we can first fix the output labels \(\oupt_D(D)\), and then label the remaining nodes in \(G\) by brute force. We first show that, if the original outcome \(O\) has success probability \(p\), then there exists an extension of \(\oupt_D\) to a labeling for \(\HH(G)\) that is a solution for \(\problem\) with probability at least \(p\). Actually finding such a labeling for \(\HH(G)\) by brute force will follow as a corollary.

\begin{lemma}
    \label{lem:cluster_completability}
    Let \(\outcome\) be a bounded-dependent outcome that solves the LCL problem \(\problem\) on \((G,\localVar)\) with success probability \(p\), and let \(\oupt_D\) be a labeling for \(\HH(G)[D]\) sampled from \(\outcome(G,\localVar)[D]\). Then, with probability at least \(p\), there exists at least one labeling \(\oupt_G\) for \(\HH(G)\) such that \(\oupt_D\) is a restriction of \(\oupt_G\) to \(\HH(G)[D]\), namely \(\oupt_D=\oupt_G\restriction_{\HH(G)[D]}\), and \(\oupt_G\) is a solution for \(\problem\) on \((G,\localVar)\).
\end{lemma}
\begin{proof}    
    If we sample labeling \(\oupt_D\) from \(O(G,\localVar)[D]\), we can be in three possible cases.
    \begin{enumerate}
        \item Labeling \(\oupt_D\) for \(\HH(G)\) is the restriction of some labeling \(\oupt\) for \(\HH(G)\), and  \(\oupt\) solves \(\problem\) in \((G,\localVar)\). In this case we are guaranteed that at least the extension \(\oupt_G=\oupt\) exists.
        \item  Labeling \(\oupt_D\) for \(\HH(G)\) is the restriction of some labeling \(\oupt\) for \(\HH(G)\), \(\oupt_D\) is a solution for \(\problem\) on \((\mathring{G}[D],\localVar\restriction_{D})\), but \(\oupt\) is not a solution for \(\problem\) in \((G,\localVar)\). Even if we cannot take the extension \(\oupt_G=\oupt\), it might still be possible to find a different \(\oupt'\) that does not belong to \(\outcome(G,\localVar)\) but that is a solution \(\problem\) on \((G,\localVar)\) nonetheless.
        \item Labeling \(\oupt_D\) for \(\HH(G)\) is the restriction of some labeling \(\oupt\) for \(\HH(G)\), and \(\oupt_D\) is not a solution for \(\problem\) on \((\mathring{G}[D],\localVar\restriction_{D})\). Then, no correct extension is possible.
    \end{enumerate}
    The first case occurs with probability \(p\), while with probability \(1-p\) either the second or the third occur. Since a correct extension of \(\oupt_D\) could be possible even in the second case, it follows that the overall probability of finding a correct extension of \(\oupt_D\) is at least \(p\).
\end{proof}

\begin{corollary}
    \label{cor:bruteforce_completion}
    Let \(\oupt_D\) be a labeling for \(\HH(G)[D]\) such that there exists a non-empty set of labelings \(L=\{\oupt_j\}_{j \in J}\) for \(\HH(G)\) where each \(\oupt_j\) is a solution for \(\problem\) on \((G,\localVar)\) and \(\oupt_D\) is a restriction of \(\oupt_j\) to \(\HH(G)[D]\), namely \(\oupt_D=\oupt_j\restriction_{\HH(G)[D]}\). There exists a \local algorithm that runs on \((G,\localVar)\) with \(V(G)\) clustered as \(D,S_1,\dots,S_k\) and outputs a labeling \(\oupt_G \in L\). This algorithm has locality \(\myO{\max_{1\le i \le k}\{\diam_G(S_i)\}}\)
\end{corollary}
\begin{proof}
    Every node in a given \(S_i\) gathers its radius-\(\myO{\diam_G(S_i)}\) neighborhood, making sure that this radius is not smaller than the diameter of \(S_i\) plus \(r\). Then we know that every node in \(S_i\) sees every other node in the same \(S_i\), plus a frontier of depth \(r\) in \(D\). We can then compute by brute force a valid solution to \(\problem\) restricted to \(S_i\) by testing every possible labeling \(\oupt_j:\HH(G)[S_i] \to \Vout \times \Eout\).
    
    The procedure just described clearly has locality \(\myO{\diam_G(S_i)}\), and it is always correct because \(\problem\) has checkability radius \(r\), and \cref{cor:nt-clustering} guarantees that nodes in different clusters have distance at least \(r\) in \(G\). Thus, the radius-\(r\) neighborhood of any node in \(S_i\) can intersect only the same \(S_i\) and \(D\). Since the labeling for \(D\) is already fixed by \(\oupt_D\), we can safely check whether an output labeling solves \(\problem\) in \(S_i\).
\end{proof}



We can now prove \cref{thm:findep-sim}:
\begin{proof}[Proof of \Cref{thm:findep-sim}:]
    Let $\Pi$ be an LCL problem with checking radius~$r$ and that admits a bounded-dependence distribution with locality~$T(n)$ and success probability~$p(n)$.
    We construct a \randlocal algorithm that solves problem~$\Pi$ with locality~$O\bigl(\sqrt{n T(n)} \poly \log n\bigr)$.
    The algorithm works in four steps:
    \begin{enumerate}
        \item It first computes a $\bigl(1/ \sqrt{n T(n)}, \sqrt{n T(n)} \log^2 n\bigr)$-clustering~$D \cup S_1, \cup \ldots, S_k$.
        By \cref{cor:nt-clustering}, this can be done with locality~$O\bigl(\sqrt{n T(n)} \log^3(n)\bigr)$.
        
        \item The algorithm then partitions the unclustered nodes~$D$ into partitions~$D_1, \ldots, D_p$ such that each partition has distance at least $2T(n)$ to every other partition.
        By \cref{lem:d_partition}, this can be done with locality~$O\bigl(\sqrt{n T(n)}\bigr)$.

        \item\label[step]{step:sim:sample} The algorithm then samples a labeling for each partition~$D_i$ independently.
        This too can be done with locality~$O\bigl(\sqrt{n T(n)}\bigr)$ thanks to \cref{lem:sampling_D} as each partition has diameter~$O\bigl(\sqrt{n T(n)}\bigr)$ and partitions are at distance at least $2T(n)$ from each other, which means that their distributions are independent.
        
        \item Finally, the algorithm completes the labeling for each cluster independently.
        This succeeds with probability at least \(p(n)\) by \cref{lem:cluster_completability}, and the brute force completion of \cref{cor:bruteforce_completion} takes locality~$O\bigl(\sqrt{n T(n)}\log^2 n\bigr)$ as this is the bound on the diameter of each cluster.
    \end{enumerate}

    The locality of the algorithm is dominated by the first step.
    In total, the algorithm has locality~$O\bigl(\sqrt{n T(n)} \log^3n\bigr)$.

    The only part in this algorithm where we use randomness is in \cref{step:sim:sample} when sampling from the bounded-dependence distribution.
    By \cref{lem:cluster_completability}, we sample a labeling that is extendable to a labeling for the whole graph with probability at least~$p(n)$.
    As the rest of the algorithm is deterministic, the algorithm succeeds with probability at least~$p(n)$.
\end{proof}

Finally, we note that \cref{thm:findep-sim} can also be derandomized e.g.\ by using techniques by \textcite{ghaffari2018derandomizing} combined with a modern network decomposition~\cite{GGHIR23}, assuming that the original bounded-dependence distribution succeeds with high enough probability.
This gives us the following result:
\begin{corollary}
    Let $\Pi$ be an LCL problem with checking radius~$r$.
    Let $\outcome$ be a bounded-dependent outcome with locality~$T(n)$ of labelings that solve problem~$\Pi$ with probability strictly larger than $1 - \frac{1}{n}$.
    Then there exists a \detlocal algorithm solving~$\Pi$ with locality~$O\bigl(\sqrt{n T(n)} \poly \log n\bigr)$.
\end{corollary}

\section{Quantum advantage: technical details}\label{sec:lcl-details}

\paragraph{Preliminary definitions.}
For any \((\VV,\EE)\)-labeled graph \(G = (V,E)\), and for any node \(v \in V\) and any edge \(e \in E\), we denote by \(L_v(e)\) the label of the half-edge \((v,e)\).
We also want to be able to refer to the endpoint of a path that starts from some node \(v\) and follows a given sequence \(L_1, \dots, L_k\) of labels over half-edges. 
We define the function \(f\) as follows:
Consider any node \(v\) and any sequence of edge labels \(L_1, \dots, L_k \in \EE\).
Let \(P = (v_1, \dots, v_{k+1})\) be a path in \(G\) such that \(v_1 = v\) and, for any edge \(e = \{v_i,v_{i+1}\}\), the half-edge \((v_i,e)\) is labeled with \(L_i\).
Then, we set \(f(v,L_1, \dots, L_k) = v_{k+1}\) if the path \(P\) exists and is unique, and \(\bot\) otherwise.

\subsection{Tree-like gadget}\label{sec:tree-like gadget}
In order to define our LCL problem, we will use, as a basic building block, an object called \emph{tree-like} gadget, which has already been used in different works related to LCLs \cite{balliu20lcl-randomness,balliu2024shared-randomness,balliu21lcl-congest}. An example of tree-like gadget is shown in \Cref{fig:treelike}.
We report here the definition provided in \cite{balliu20lcl-randomness}.
\begin{definition}[Tree-like gadget \cite{balliu20lcl-randomness}]\label{def:tree-like-gadget}
	A graph $G$ is a \emph{tree-like} gadget of height $\ell$ if it is possible to assign coordinates $(l_u, k_u)$ to each node $u\in G$, where
	\begin{itemize}[noitemsep]
		\item $0\le l_u < \ell$ denotes the depth of $u$ in the tree, and
		\item $0\le k_u < 2^{l_u}$ denotes the position of $u$ (according to some order) in layer $l_u$,
	\end{itemize}
	such that there is an edge connecting two nodes $u,v\in G$ with coordinates $(l_u, k_u)$ and $(l_v, k_v)$ if and only if:
	\begin{itemize}[noitemsep]
		\item $l_u = l_v$ and $|k_u - k_v | = 1$, or
		\item $l_v = l_u-1$ and $k_v = \lfloor \frac{k_u}{2} \rfloor$, or
		\item $l_u = l_v-1$ and $k_u = \lfloor \frac{k_v}{2} \rfloor$.
	\end{itemize}
\end{definition}

A useful property of this gadget is that it can be made \emph{locally checkable}, in the sense that there exists a set of labels $\EE^{\ltreelike}$ that can be assigned to each node-edge pair, and a constraint $C^{\ltreelike}$ over the labels $\EE^{\ltreelike}$, satisfying the following.
\begin{lemma}[\cite{balliu20lcl-randomness}]\label{lemma:constraints-to-treelike}
    Let $G$ be a graph that is labeled with labels in $\EE^{\ltreelike}$ such that $C^{\ltreelike}$ is satisfied for all nodes in $G$. Then, each connected component of $G$ is a tree-like gadget.
\end{lemma}
\begin{lemma}[\cite{balliu20lcl-randomness}]\label{lem:treelike-are-labellable}
    Each tree-like gadget graph $G$ can be labeled with labels in  $\EE^{\ltreelike}$ such that $C^{\ltreelike}$ is satisfied for all nodes in $G$.
\end{lemma}
The set of labels $\EE^{\ltreelike}$ is defined in \cite{balliu21lcl-congest} as $\EE^{\ltreelike} = \{\lleft,\lright, \lparent,\llch,\lrch\}$ (the labels stand for ``left'', ``right'', ``parent'',  ``left child'', and ``right child'', respectively). 
The set of local constraints $\CC^{\ltreelike}$ is defined in \cite{balliu21lcl-congest} as follows.
\begin{myframe}{The constraints $\CC^{\ltreelike}$ of \cite{balliu21lcl-congest}}
	\begin{enumerate}
		\item For any two edges $e,e'$ incident to a node $u$, it must hold that $L_u(e)\neq L_u(e')$;\label{cons-tree:differentEdgeLabels}
		
		\item For each edge $e=\{u,v\}$, if $L_u(e)=\lleft$, then $L_v(e)=\lright$, and vice versa; \label{cons-tree:left-right}
		
		\item For each edge $e=\{u,v\}$, if $L_u(e)=\lparent$, then $L_v(e)\in\{\llch,\lrch\}$, and vice versa;\label{cons-tree:parent-child}
		
		\item If a node $u$ has an incident edge $e=\{u,v\}$ with label $L_u(e)=\lparent$ such that $L_v(e)=\llch$, then $f(u, \lparent,\lrch,\lleft)=u$;\label{cons-tree:triangle}
		
		\item If a node $u$ has an incident edge $e=\{u,v\}$ with label $L_u(e)=\lparent$ such that $L_v(e)=\lrch$, if $u$ has an incident edge labeled $\lright$, then $f(u, \lparent,\lright,\llch,\lleft)=u$.\label{cons-tree:square}
		
		\item If a node has an incident half-edge labeled $\llch$, then it must also have an incident half-edge labeled $\lrch$, and vice versa;\label{cons-tree:2children}
		
		\item A node does not have an incident half-edge labeled $\lparent$ if and only if it has no incident half-edges labeled $\lleft$ or $\lright$; \label{cons-tree:root}
		
		\item If a node $u$ does not have an incident edge $e$ with label $L_u(e)\in\{\llch, \lrch\}$, then neither do nodes $f(u,\lleft)$ and $f(u,\lright)$ (if they exist);\label{cons-tree:boundarychildren}
		
		\item If a node $u$ has an incident edge $e=\{u,v\}$ with label $L_u(e)=\lparent$ such that $L_v(e)=\lrch$ (resp.\ $L_v(e)=\llch$), then $u$ has an incident edge labeled $\lright$ (resp.\ $\lleft$) if and only if $f(u,\lparent)$ has an incident edge labeled $\lright$ (resp.\ $\lleft$).\label{cons-tree:boundarylr}
		
	\end{enumerate}
\end{myframe}
An example of tree-like gadget labeled from labels in $\EE^{\ltreelike}$ such that $C^{\ltreelike}$ is satisfied for all nodes is shown in \Cref{fig:treelike}.

Moreover, in \cite{balliu2024shared-randomness}, it is shown that it is easy for the nodes to \emph{prove} that the graph in which they are is an invalid tree-like gadget. More formally, in \cite{balliu2024shared-randomness}, it is defined an LCL problem $\problem^{\mathrm{badtree}}$ satisfying the following.
Nodes have input $0$ or $1$, and nodes with input $1$ are called \emph{marked}. Node-edge pairs have an input label from $\EE^{\ltreelike}$. Recall that a graph with such an input is called $(\{0,1\},\EE^{\ltreelike})$-labeled graph. The set $\VV^{\lbadtree}$ of outputs labels for $\problem^{\mathrm{badtree}}$ contains a special label called $\bot$, and in this problem there are only node output labels (hence, there are no half-edge output labels). 
\begin{lemma}[\cite{balliu2024shared-randomness}]\label{lem:badtree}
    There exists an LCL problem $\problem^{\mathrm{badtree}}$ satisfying the following properties.
    \begin{itemize}
        \item Let $G$ be a $(\{0,1\},\EE^{\ltreelike})$-labeled graph where $C^{\ltreelike}$ is satisfied on all nodes and all nodes are not marked. Then, the only valid solution for $\problem^{\mathrm{badtree}}$ is the one assigning $\bot$ to all nodes.
        \item Let $G$ be a connected $(\{0,1\},\EE^{\ltreelike})$-labeled graph where either $C^{\ltreelike}$ is not satisfied on at least one node, or there is at least one marked node. Then, there exists a solution for $\problem^{\mathrm{badtree}}$ where all nodes produce an output different from $\bot$. Moreover, such a solution can be computed in $O(\log n)$ deterministic rounds in the LOCAL model.
    \end{itemize}
\end{lemma}
On a high level, we will later use $\problem^{\mathrm{badtree}}$ as a black box, and we will mark nodes that witness some local errors that are not related to the tree-like gadget itself. In this way, we will be able to use $\problem^{\mathrm{badtree}}$ to also prove errors that are unrelated to the tree-like gadget itself. An example of valid solution of $\problem^{\mathrm{badtree}}$ is depicted in \Cref{fig:badtree-solution}.

\subsection{Octopus gadget}\label{sec:octopus gadget}
In this section, we formally define the notion of \emph{octopus gadget}, then we prove that an octopus gadget is locally checkable, and finally we define the LCL problem $\problem^{\lbadgadget}$ and prove some properties about it.

\begin{definition}[Octopus gadget]\label{def:quantum-advantage:octopus-gadget}
	Let $x \ge 1$ be a natural number, and \(\eta = (\eta_0, \dots, \eta_{2^{x-1}-1}) \) a vector of \(2^{x-1}\) entries in \(\{1,2\}\).
	Let \(W = \{w_{(i,j)}\}_{(i,j) \in I}\) be a family of positive integer weights, where $I$ is the set containing all pairs $(i,j)$ satisfying \((i,j) \in \{0,1,\dots,2^{x-1}-1\}\times\{1,2\}\)  and \(j \le \eta_{i}\).

	A graph \(G = (V,E)\) is an \((x,\eta,W)\)-\emph{octopus gadget} if there exists a labeling \(\lambda: V \to \LL = I \cup \{\text{root}\}\) of the nodes of \(G\) such that the following holds.
	\begin{enumerate}
		\item For each element $y \in \LL$, let \(G_y\) be the subgraph of \(G\) induced by nodes labeled with \(y\). Then, for all \(y \in \LL\), \(G_y\) must be a tree-like gadget according to \cref{def:tree-like-gadget}.
		\item For all \(y,z \in \LL\) such that \(y \neq z\), \(G_y\) and \(G_z\) must be disjoint.
		\item \(G_{\text{root}}\) has height \(x\) and, for all \((i,j) \in I\), \(G_{(i,j)}\) has height \(w_{(i,j)} \in W\).
		\item For all \((i,j) \in I\), there is an edge connecting the node of \(G_{(i,j)}\) that has coordinates \((0,0)\) with the node of \(G_{\text{root}}\) that has coordinates \((x-1,i)\).
	\end{enumerate}
	\(G_{\text{root}}\) is called the \emph{head-gadget} and, for all \((i,j) \in I\), \(G_{(i,j)}\) is called a \emph{port-gadget}.

\end{definition}

\paragraph{Local checkability.}
While \Cref{def:quantum-advantage:octopus-gadget} gives the definition of octopus gadget from a \emph{global} perspective, we now define a finite set of labels and a set of local constraints satisfying that a connected graph is a properly labeled octopus gadget if and only if these \emph{local} constraints are satisfied by all nodes.

We first define the sets of labels $\VV^{\lgadget}$ and $\EE^{\lgadget}$,  and then we define a set of local constraints $\CC^{\lgadget}$ over these labels. We will prove that a connected graph $G$ can be ($\VV^{\lgadget}$, $\EE^{\lgadget}$)-labeled such that the constraints $\CC^{\lgadget}$ are satisfied on all nodes if and only if $G$ is an ($x$, $\eta$, $W$)-octopus gadget, for some $x$, $\eta$, and $W$.

The sets of labels are defined as follows:
\begin{align*}
\VV^{\lgadget} &= \{\gadgetN, \gadgetP\}, \\
\EE^{\lgadget} &= \{\nplink_1, \nplink_2, \pnlink\} \cup \EE^{\ltreelike},
\end{align*}
where $\nplink$ stands for ``head-port link'' and $\pnlink$ stands for ``port-head link'', respectively. Recall that the labels in $\VV^{\lgadget}$ are \emph{node} labels, while $\EE^{\lgadget}$ are labels for \emph{node-edge pairs}. 

We now define the set of constraints $\CC^{\lgadget}$. An edge $e = \{u,v\}$ is called \emph{internal} if and only if $\{ L_u(e), L_v(e) \} \cap \{ \nplink_1, \nplink_2, \pnlink \} = \emptyset$, and \emph{external} otherwise. 
\begin{myframe}{The constraints $\CC^{\lgadget}$}
	\begin{itemize}
		\item[0.] In the subgraph induced by internal edges, all the constraints in $\CC^{\ltreelike}$ must be satisfied.\label{cons-gadget:tree}
		
		\item[1.] For each edge $e = \{u,v\}$, it must hold that nodes $u$ and $v$ are both labeled \(\gadgetN\), or both labeled \(\gadgetP\), if and only if $e$ is an internal edge.\label{cons-gadget:fixed-point}
		
		\item[2.] For each edge $e = \{u,v\}$, if $L_u(e) = \nplink_j$ for some \(j \in \{1,2\}\), then $L_v(e) = \pnlink$, and vice versa.\label{cons-gadget:edge}
		
		\item[3.] If a node $v$ is labeled with $\gadgetN$ and $e$ is an edge incident to $v$, then $L_v(e) \neq \pnlink$.\label{cons-gadget:node-edge}
		
		\item[4.] If a node $v$ is labeled with $\gadgetP$ and $e$ is an edge incident to $v$, then $L_v(e) \neq \nplink_j$ for all \(j \in \{1,2\}\).\label{cons-gadget:port-edge}
		
		\item[5.] A node $v$ has an incident edge labeled $\nplink_j$ for some \(j \in \{1,2\}\) if and only if $v$ is labeled with $\gadgetN$ and it has no incident edge $e$ satisfying $L_v(e) \in \{\llch, \lrch\}$.\label{cons-gadget:child-link}
		
		\item[6.] A node $v$ has an incident edge labeled $\pnlink$ if and only if $v$ is labeled with $\gadgetP$ and it has no incident edge $e$ satisfying $L_v(e) = \lparent$.\label{cons-gadget:parent-link}
		
		\item[7.] A node cannot have more than one incident edge labeled with $\pnlink$, \(\nplink_1\), and \(\nplink_2\).\label{cons-gadget:no-multiple}
		
		\item[8.] If a node \(v\) has an incident edge \(e\) such that \(L_v(e) = \nplink_2\), then it has another incident edge \(e'\) such that \(L_v(e') = \nplink_1\).\label{cons-gadget:at_least}
		\end{itemize}

\end{myframe}

\begin{lemma}\label{lem:octopuslabeling}
    Let $G$ be a connected graph that is ($\VV^{\lgadget}$, $\EE^{\lgadget}$)-labeled such that $\CC^{\lgadget}$ is satisfied at all nodes. Then, $G$ is an octopus gadget.
\end{lemma}

\begin{proof}
	In the following, by labels of an edge $e = \{u,v\}$ we denote the labels $L_v(e)$ and $L_u(e)$.
	By constraint 2, for each edge $e$ it holds that the labels of $e$ are either both in $\{\nplink_1, \nplink_2, \pnlink\}$, or both in $\EE^{\ltreelike}$. Hence, internal edges have only labels from  $\EE^{\ltreelike}$ and external edges have only labels from $\{\nplink_1, \nplink_2, \pnlink\}$. Moreover, by constraint 2, an external edge has exactly two labels from $\{\nplink_1, \nplink_2, \pnlink\}$, and these two labels cannot be $\{\nplink_1, \nplink_2\}$. 
	
	By constraint 0 and \Cref{lemma:constraints-to-treelike}, each connected component in the subgraph induced by internal edges is a tree-like gadget (possibly, a graph without edges). Moreover, for each connected component, by constraint 1, it must hold that all nodes in the component are either all labeled \(\gadgetN\), or all labeled \(\gadgetP\).

	Summarizing what we have showed so far, each connected component in the subgraph induced by internal edges is either a tree-like gadget where all nodes are labeled \(\gadgetN\) or a tree-like gadget where all nodes are labeled \(\gadgetP\), and external edges have labels $\{\nplink_1, \pnlink\}$ or labels $\{\nplink_2, \pnlink\}$. Hence, we only need to prove that external edges are properly connected.
	
	Suppose $G$ is non-empty. Let $v$ be a node in $G$. We prove that the connected component $G'$ (in $G$, including external edges) containing $v$ is an octopus gadget. 
	
	Let $H$ be the connected component containing $v$ in the graph induced by internal edges. As discussed, $H$ is a tree-like gadget. 
	We start by proving that, in $G'$, there must be a tree-like gadget where all nodes are labeled \(\gadgetN\). If $v$ itself is labeled \(\gadgetN\), we are done. If $v$ is marked \(\gadgetP\), then by constraint 6 the node $u$ in the tree-like gadget $H$ that has no incident half-edge labeled $\lparent$ (i.e., the node with coordinates $(0,0)$) must have an edge $e = \{u,z\}$ satisfying $L_u(e) = \pnlink$. By constraint 2, $L_z(e) = \nplink_j$ for some \(j \in \{1,2\}\), and by constraint 4 no node labeled \(\gadgetP\) can have incident edges with label $\nplink_j$ for any \(j \in \{1,2\}\). Hence, $z$ must be a node labeled \(\gadgetN\).

	Let $H_h$ be the connected component containing $z$ in the graph induced by internal edges. $H_h$ must be a tree-like gadget.
	By constraint 5, only nodes in the last layer of $H_h$ can have an incident edge labeled $\nplink_j$ for some \(j \in \{1,2\}\), and such nodes must have at least one such edge. 
	By constraint 7, each of these nodes cannot have other edges labeled with \(\nplink_j\).
	Let $e = \{w_h,w_p\}$ be an arbitrary such edge, where $L_{w_h}(e) = \nplink_j$ for some \(j \in \{1,2\}\). 
	By constraint 2, it must hold that  $L_{w_p}(e) = \pnlink$. Moreover, by constraint 7, node $w_p$ cannot have additional incident edges labeled $\pnlink$. 
	By constraint 3, node $w_p$ must be labeled \(\gadgetP\) and by constraints 6, 0, and 1, node $w_p$ must be the root of a tree-like gadget $H_p$ where all nodes are labeled \(\gadgetP\). 
	By constraint 6, no other node of $H_p$ can have edges labeled $\pnlink$. 
	Furthermore, by constraint 8, each node in the last layer of $H_p$ must have an incident edge labeled $\nplink_1$ if it has an incident edge labeled $\nplink_2$.

	We thus get that, for each head-gadget $H_h$, there is at least one port-gadget $H_p$ connected to each leaf of $H_h$ and at most two of them, and that $H_p$ is connected to exactly one head-gadget via its root. Hence, $G'$ is an octopus gadget. 
	Let \(x\) be the height of the tree-like gadget \(H_h\).
	Now, we assign the label \(\texttt{root}\) to all nodes in \(H_h\), and the label \((i,j)\) to all nodes in the port-gadget \(H_p\) if \(H_p\) is connected to \(H_h\) via the leaf \(v\) of \(H_h\) having coordinates \((x-1,i)\) through the edge \(e\) such that \(L_v(e) = \nplink_j\).
	By this assignment, we have shown that \(G'\) satisfies \cref{def:quantum-advantage:octopus-gadget}.

	Finally, since $G$ is connected, then $G' = G$.
\end{proof}

\begin{lemma}\label{lemma:label:octopus}
	Let $G$ be an octopus gadget. Then, $G$ can be ($\VV^{\lgadget}$, $\EE^{\lgadget}$)-labeled such that $\CC^{\lgadget}$ is satisfied at all nodes.
\end{lemma}

\begin{proof}
	Let $G=(V,E)$ be an octopus gadget. We can assign the labels in this way.
	\begin{itemize}
		\item We assign the label $\gadgetN$ to all the nodes of the head-gadget.
		\item We assign the label $\gadgetP$ to all the nodes of the port-gadgets.
		\item By \Cref{lem:treelike-are-labellable}, we can assign labels to each tree-like gadget such that the constraints of $\CC^\ltreelike$ are satisfied. 
		\item Let $e = \{u,v\}$ be an edge connecting a head-gadget to a port-gadget (i.e., the edges of point 4 of \cref{def:quantum-advantage:octopus-gadget}), where $u$ is the node in the head-gadget.
		If \(v\) has the label \((i,j) \in I\) in \cref{def:quantum-advantage:octopus-gadget}, we set $L_u(e) = \nplink_j$ and $L_v(e) = \pnlink$.
	\end{itemize}
	By construction, all the constraints of $\CC^\lgadget$ are satisfied.
\end{proof}

\paragraph{The LCL problem \boldmath $\problem^{\lbadgadget}$.}

We now define an LCL problem $\problem^{\lbadgadget}$ that, on a high level, allows the nodes to prove that the graph in which they are is not an octopus gadget. Similarly as in the case of $\problem^{\lbadtree}$, nodes also receive a binary input, where nodes that receive $1$ are called \emph{marked}. Again, we will use this input later, to mark nodes that witness errors that may be unrelated to the octopus gadget itself.
On a high level, the problem will satisfy the following properties.
\begin{itemize}
	\item Nodes receive as input whether they are marked or not.
	
	\item There are two possible types of output: a node can either produce an empty output ($\bot$), or it can output an error. In the latter case, a node needs to also output a locally checkable proof of the fact that the graph is not an octopus gadget or that it contains at least one marked node.
	
	\item If the octopus gadget is valid and it does not contain any marked node, the constraint of the problem are defined such that the only valid solution to the problem $\problem^{\lbadgadget}$ is the one where all nodes output $\bot$.
	
	\item Otherwise, if the graph is not an octopus gadget or it contains at least one marked node, \emph{all} nodes are able to spend $O(\log n)$ round in the \local model to produce a proof of this fact.
\end{itemize}
We now give a formal definition of the LCL problem $\problem^\lbadgadget$.
The inputs are defined as follows.
\begin{itemize}
	\item Each node $v$ receives an input pair $(m_v, g_v)$ from the set $\{0,1\} \times \VV^{\lgadget}$, that is, each node receives a pair, where the first element $m_v$ denotes whether the node is marked or not, while the second element $g_v$ of the pair is an element from the previously described set $\VV^{\lgadget}$.
	\item  Each half-edge receives an input from the previously described set $\EE^{\lgadget}$. 
\end{itemize}
Let $\VV^{\lbadtree}$ be the node output labels of the LCL problem $\problem^{\lbadtree}$. The possible output labels $\VV^{\lbadtree}$ are the following. 
\begin{itemize}
	\item The label $\bot$, which represent an empty output.
	\item A pair $(\lerror,\mathcal{M})$, where $\mathcal{M}$ is a triple, and each element of the triple is a label from $\VV^{\lbadtree}$.
\end{itemize}
We denote the pairs $(\lerror,\mathcal{M})$ as \emph{error outputs}. On a high level, the constraints $\CC^{\lbadgadget}$ are defined such that the output $\bot$ is always allowed, while the constraints on the error outputs are defined such that these outputs encode a proof of the fact that there is an error in the graph.
Informally, the error output triples will be used as follows.
\begin{itemize}
	\item If the \emph{first} element of a triple of some node $v$ is not $\bot$, it means that, in the connected component of the subgraph induced by internal edges that contains $v$ there is some error.
	\item If the \emph{second} element of a triple of some node $v$ is not $\bot$, it means that, even if the connected component of $v$ may be correct, node $v$ is in a head gadget connected to at least one broken port gadget.
	\item If the \emph{third} element of a triple of some node $v$ is not $\bot$, it means that, even if the connected component of $v$ may be correct, and even if $v$ is in a correct port gadget connected to a correct head gadget, node $v$ is in a port gadget connected to a head gadget that is connected to at least one broken port gadget.
\end{itemize}
In order to allow this labeling, we allow more and more nodes to be marked, as follows.
\begin{itemize}
	\item For the first instance of $\problem^{\lbadtree}$, only marked nodes and nodes that have already witnessed some error are marked.
	\item For the second instance, additionally, nodes in head gadgets that are neighbors of nodes of port gadgets that output an error in the previous instance are marked.
	\item For the third instance, additionally, nodes in port gadgets that are neighbors of nodes of head gadgets that output an error in the previous instance are marked.
\end{itemize}
More formally, the constraints $\CC^{\lbadgadget}$ are defined as follows.
\begin{myframe}{The constraints $\CC^{\lbadgadget}$}
	\begin{itemize}
		\item A node can always output $\bot$.\label{cons-badoct:bot}
		\item If a node $v$ outputs $(\lerror,(x_{v,1},x_{v,2},x_{v,3}))$, then the following must hold. \label{cons-badoct:error}
		\begin{enumerate}
			\item The triple $(x_{v,1},x_{v,2},x_{v,3})$ must be different from $(\bot,\bot,\bot)$.
			\item Consider the labeling of the graph obtained by labeling each node $v$ with the label $g_v \in \VV^{\lgadget}$, and each half-edge $(v,e)$ with the label $L_v(e) \in \EE^{\lgadget}$ (i.e., on the nodes we take the second element of their input pair, while on the edges we take the original input half-edge label). For each node $v$, let $i_{v,1}:= 1$ if, according to this labeling, node $v$ does not satisfy the constraints $\CC^{\lgadget}$ or if $m_v = 1$, and let $i_{v,1} := 0$ otherwise. 
			\item Consider the labeling of the subgraph induced by internal edges obtained by input labeling each node $v$ with $i_{v,1}$ and output labeling each node $v$ with $x_{v,1}$. This labeling must satisfy the constraints of $\problem^{\lbadtree}$.
			\item Let $i_{v,2} := 1$ if $i_{v,1} = 1$ or if $v$ is a node such that $g_v = \gadgetN$ and such that there exists an edge $e = \{u,v\}$ satisfying $L_v(e) = \nplink_j$ for some \(j \in \{1,2\}\) and $x_{u,1} \neq \bot$. Let $i_{v,2} := 0$ otherwise.
			\item Consider the labeling of the subgraph induced by internal edges obtained by input labeling each node $v$ with $i_{v,2}$ and output labeling each node $v$ with $x_{v,2}$. This labeling must satisfy the constraints of $\problem^{\lbadtree}$.
			\item Let $i_{v,3} := 1$ if $i_{v,2} = 1$ or if $v$ is a node such that $g_v = \gadgetP$ and such that there exists an edge $e = \{u,v\}$ satisfying $L_v(e) = \pnlink$ and $x_{u,2} \neq \bot$. Let $i_{v,3} := 0$ otherwise.
			\item Consider the labeling of the subgraph induced by internal edges obtained by input labeling each node $v$ with $i_{v,3}$ and output labeling each node $v$ with $x_{v,3}$. This labeling must satisfy the constraints of $\problem^{\lbadtree}$.
		\end{enumerate}
	\end{itemize}
\end{myframe}

\begin{lemma}\label{lem:validoctopus}
	Let $G$ be a $(\{0,1\} \times \VV^{\lgadget}$, $\EE^{\lgadget})$-labeled graph where $\CC^{\lgadget}$ is satisfied for all nodes and all nodes are not marked. Then, the only valid solution for $\problem^{\lbadgadget}$ is the one assigning $\bot$ to all nodes.
\end{lemma}
\begin{proof}
	For a contradiction, assume that some node $v$ in $G$ outputs $(\lerror,(x_{v,1},x_{v,2},x_{v,3}))$ for some labels $x_{v,1}$, $x_{v,2}$, $x_{v,3}$. By the definition of $i_{v,1}$ in constraint 2, for all nodes $v$ it must hold that $i_{v,1} = 0$, and by constraint 3 and \Cref{lem:badtree} we thus get that $x_{v,1} = \bot$.
	Similarly, by the definition of $i_{v,2}$ in constraint 4, for all nodes $v$ it must hold that $i_{v,2} = 0$, and by constraint 5 and \Cref{lem:badtree} we thus get that $x_{v,2} = \bot$.
	Again, by the definition of $i_{v,3}$ in constraint 6, for all nodes $v$ it must hold that $i_{v,3} = 0$, and by constraint 7 and \Cref{lem:badtree} we thus get that $x_{v,3} = \bot$.
	The claim then follows by the fact that constraint 1 forbids the triple $(\bot,\bot,\bot)$.
\end{proof}

\begin{lemma}\label{lem:badoctopus}
	Let $G$ be a connected $(\{0,1\} \times \VV^{\lgadget}$, $\EE^{\lgadget})$-labeled graph where either $\CC^{\lgadget}$ is not satisfied on at least one node, or there is at least one marked node. Then, there exists a solution for $\problem^{\lbadgadget}$ where all nodes produce an output different from $\bot$. Moreover, such a solution can be computed in $O(\log n)$ deterministic rounds in the \local model.
\end{lemma}

\begin{proof}
	In order to produce the claimed solution, we will use, as a black box, the algorithm for $\problem^{\lbadtree}$ reported in \cref{lem:badtree}.
	The algorithm for solving $\problem^{\lbadgadget}$ works as follows. Each node $v$ spends $O(1)$ rounds to compute $i_{v,1}$. 
	By \cref{lem:badtree}, nodes can spend $O(\log n)$ deterministic classical rounds to compute a solution for $\problem^{\lbadtree}$, where $i_{v,1}$ is the input of $v$ for the problem, and, if in the subgraph induced by internal edges, in the connected component containing $v$, there is at least one marked node or a node not satisfying $\CC^{\ltreelike}$, then $v$ outputs a label different from $\bot$. For each node $v$, let $x_{v,1}$ be this output.

	Then, each node $v$ spends $O(1)$ rounds to compute $i_{v,2}$, and again, as before, nodes spend $O(\log n)$ rounds to compute a solution for $\problem^{\lbadtree}$, where this time the input of node $v$ is $i_{v,2}$.

	Finally, each node $v$ spends $O(1)$ rounds to compute $i_{v,3}$ and spends $O(\log n)$ rounds to compute a solution for $\problem^{\lbadtree}$, where the input of node $v$ is now $i_{v,3}$. If a node $v$ satisfies $x_{v,1} = x_{v,2} = x_{v,3}$, then it outputs $\bot$. Otherwise, it outputs $(\lerror,(x_{v,1}, x_{v,2}, x_{v,3}))$.

	The correctness of the output directly follows from the definition of $\CC^{\lbadgadget}$, and the runtime is clearly upper bounded by $O(\log n)$ deterministic classical rounds. We now prove that, if $G$ contains at least one marked node, or at least one node not satisfying $\CC^{\lgadget}$, then all nodes output a label different from $\bot$.

	Let $\hat{G}$ be the graph obtained by contracting each connected component induced by internal edges into a single node. If there are multiple edges between different connected components, we have parallel edges in $\hat{G}$. By the definition of the algorithm, it holds that all nodes $v$ of $G$ corresponding to the same node of $\hat{G}$, in the first output of $\problem^{\lbadtree}$ (i.e., the values $x_{v,1}$), either all output $\bot$ or all output something different. Let us label $E$ (which stands for \emph{error}) each node of $\hat{G}$ satisfying that all its nodes output something different from $\bot$. We label all the other nodes either $H$ (\emph{head}) or $P$ (\emph{port}), as follows. If a node of $\hat{G}$ is not labeled $E$, then either all its nodes of $G$ have input $\gadgetN$, or all its nodes of $G$ have input $\gadgetP$. In the former case, we label the node $H$, while in the latter case we label the node $P$.
	
	By \Cref{lem:octopuslabeling}, at least one node of $\hat{G}$ is labeled $E$. Moreover, by the definition of $\CC^{\lgadget}$, the graph $\hat{G}$ must satisfy the following properties.
	\begin{itemize}[noitemsep]
		\item Each node labeled $P$ must have degree $1$.
		\item Nodes labeled $P$ must form an independent set.
		\item Nodes labeled $H$ must form an independent set.
	\end{itemize}
	We thus get that each connected component of $\hat{G}$ induced by nodes not labeled $E$ must form stars centered at nodes labeled $H$.
	Moreover, we observe that each node labeled $H$ must have at least one neighbor labeled $E$, since otherwise we would get that, in $\hat{G}$, there is a connected component not containing $E$, which in $G$ corresponds to a valid octopus gadget.
	By the definition of the algorithm, we thus get that, in each connected component corresponding to a node labeled $H$, at least one node is marked when solving $\problem^{\lbadtree}$ for the second time. Let us update the labeling of $\hat{G}$ by changing nodes labeled $H$ into nodes labeled $E$ if their connected component contains nodes not outputting $\bot$.
	We get that, after this update, no node of $\hat{G}$ is labeled $H$.
	Hence, all nodes of $\hat{G}$ labeled $P$ have a neighbor labeled $E$.
	By the definition of the algorithm, we thus get that, in each connected component corresponding to a node labeled $P$, at least one node is marked when solving $\problem^{\lbadtree}$ for the third time. Let us update the labeling of $\hat{G}$ by changing nodes labeled $P$ into nodes labeled $E$ if their connected component contains nodes not outputting $\bot$.
	We get that, after this update, no node of $\hat{G}$ is labeled $P$.
	Hence, all nodes of $\hat{G}$ are labeled $E$, and hence no node of $G$ outputs $\bot$.
	
\end{proof}

\subsection{The family of proper instances}\label{sec:graph family}
In \cref{sec:intro_proper_instances}, we informally defined the family of graphs called \emph{proper instances} that we will use to prove our main result. 
We now provide a formal definition of such graphs.
The main ingredient is the octopus gadget defined in \cref{def:quantum-advantage:octopus-gadget}. 

\begin{definition}[Proper instance]\label{def:proper-instance}
	Let $G=(V,E)$ be a graph.
	We say that \(G\) is a \emph{proper instance} if there exists a node labeling function \(\lambda: V \to \{\texttt{intra-octopus}, \texttt{inter-octopus}\}\) with the following properties.
	\begin{enumerate}
		\item Every connected component in the subgraph induced by nodes labeled \(\texttt{intra-octopus}\) is an octopus gadget (according to \cref{def:quantum-advantage:octopus-gadget}).
		\item The subgraph induced by nodes labeled  \(\texttt{inter-octopus}\) does not contain any edge.
		\item A node $v$ labeled \(\texttt{intra-octopus}\) is connected to a node labeled \(\texttt{inter-octopus}\) if and only if $v$ has coordinates \((w-1,0)\) in the port-gadget $P$ containing $v$, where $w$ is the height of $P$ (that is, $v$ is the left-most leaf of the port-gadget containing $v$).
	\end{enumerate}

		
 
\end{definition}

\paragraph{Local checkability.}
While \Cref{def:proper-instance} gives the definition of proper instances from a \emph{global} perspective, we now define a finite set of labels and a set of local constraints satisfying that a connected graph is a properly labeled valid instance if and only if these \emph{local} constraints are satisfied by all nodes.

We first define the sets of labels $\VV^{\proper}$ and $\EE^{\proper}$,  and then we define a set of local constraints $\CC^{\proper}$ over these labels. We will prove that a connected graph $G$ can be ($\VV^{\proper}$, $\EE^{\proper}$)-labeled such that the constraints $\CC^{\proper}$ are satisfied on all nodes if and only if $G$ is a proper instance.

The sets of labels are defined as follows. The node labels are defined as $\VV^{\proper} = \VV^\lgadget \cup \{\inter\}$, where $\inter$ stands \emph{inter-octopus node}. The half-edge labels are defined as $\EE^{\proper} = \EE^\lgadget \cup \{\pilink, \iplink\}$, where $\pilink$ stands for ``port$-$inter-octopus link'' and  $\iplink$ stands for ``inter-octopus$-$port link''.

We now define the set of constraints $\CC^{\proper}$. An edge $e = \{u,v\}$ is called \emph{intra-octopus} if and only if $\{ L_u(e), L_v(e) \} \cap \{ \pilink, \iplink \} = \emptyset$, and \emph{inter-octopus} otherwise.

\begin{myframe}{The constraints $\CC^{\proper}$}
	\begin{itemize}
		\item[0.] In the subgraph induced by intra-octopus edges, all the constraints in $\CC^{\lgadget}$ must be satisfied.\label{cons-proper:octopus}
		
		\item[1.] For each edge $e = \{u,v\}$, if $L_u(e) = \pilink$, then $L_v(e) = \iplink$, and vice versa.\label{cons-proper:edge}
		
		\item[2.] Each node $v$ has an incident edge labeled $\pilink$ if and only if $v$ is labeled with $\gadgetP$ and it has no incident edge $e$ satisfying $L_v(e) \in \{\llch, \lrch, \lleft\}$.\label{cons-proper:pilink}
		
		\item[3.] Each node $v$ labeled $\inter$ has degree at least $1$ and only incident half-edges labeled $\iplink$. Moreover, nodes not labeled $\inter$ have no incident half-edges labeled $\iplink$.\label{cons-proper:iplink:iff}
	
		\end{itemize}

\end{myframe}

\begin{lemma}\label{lem:constraints-to-proper}
	Let \(G\) be any non-empty connected graph that is \((\VV^\proper, \EE^\proper)\)-labeled such that \(\CC^\proper\) is satisfied at all nodes. 
	Then, \(G\) is a proper instance according to \cref{def:proper-instance}.
\end{lemma}
\begin{proof}
	We start by proving that there must be at least one node not labeled $\inter$. Since the graph $G$ is non-empty, it contains at least some node $v$, which is either labeled $\inter$ or not. In the latter case we are done. In the former case, by constraint \(3\), \(v\) has an incident edge \(e = \{u,v\}\) such that \(L_v(e) = \iplink\).
	By constraints \(1\) and \(2\), \(u\) is labeled with \(\gadgetP\) and \(L_u(e) = \pilink\). Hence, a node not labeled $\inter$ exists.

	We thus know that the subgraph \(H_{\text{Oct}}\) induced by nodes not labeled $\inter$ is non-empty, and, by constraint 0, \(H_{\text{Oct}}\) is such that each connected component is an octopus gadget.

	By constraint 2, each \(\gadgetP\) node \(u\) that is the left-most leaf of its gadget must have an incident edge \(\{u,u'\}\) such that \(L_u(e) = \pilink\), and no other node is allowed to have such an incident edge.
	By constraint 1 and 3, \(u'\) must be labeled \(\inter\) and \(L_{u'}(e) = \iplink\). This in particular implies that the set \(V(H_{\text{Ext}}) = V(G) \setminus V(H_{\text{Oct}})\) is non-empty.

	By constraint 3, all edges incident to nodes in \(V(H_{\text{Ext}})\) are labeled \(\iplink\).
	Moreover, by constraint 3, no nodes in $H_{\text{Oct}}$ have incident half-edges labeled $\iplink$. By constraint 1, each edge $e = \{u,v\}$ incident to an $\inter$ node $u$, must satisfy \(L_v(e) = \pilink\), and, by constraint 2, $v$ must be a left-most leaf of a port gadget.
	Summarizing, we obtained the following.
	\begin{itemize}
		\item Each connected component in the subgraph induced by intra-octopus edges is an octopus gadget.
		\item Only left-most leaves of port gadgets have incident inter-octopus edges.
		\item Inter-octopus edges must have an endpoint that is a left-most leaf of a port gadget and an endpoint that is an inter-octopus node.
	\end{itemize}
	Hence, $G$ is a proper instance.
\end{proof}

\begin{lemma}
	Let \(G\) be a proper instance as defined in \cref{def:proper-instance}.
	Then, there exists a \((\VV^\proper, \EE^\proper)\)-labeling of \(G\) that satisfies the constraints in \(\CC^\proper\) at all nodes.
\end{lemma}
\begin{proof}
	We properly label all nodes and edges in octopus gadgets respecting \(\CC^{\lgadget}\) as in \cref{lemma:label:octopus}.
	To all nodes that are outside octopus gadgets, we assign the label \(\inter\).
	To all edges \(e\) that connect a node \(u\) labeled with \(\inter\) to a node \(v\) labeled with \(\gadgetP\), we assign labels so that \(L_u(e) = \iplink\) and \(L_v(e) = \pilink\).
	All constraints are satisfied by construction.
\end{proof}

\paragraph{The LCL problem $\problem^{\badgraph}$.} 
We now define an LCL problem $\problem^{\badgraph}$ that, on a high level, allows the nodes to prove that the graph in which they are is not a proper instance labeled in a valid way.
Differently from the cases of $\problem^{\lbadtree}$ and $\problem^{\lbadgadget}$, we cannot guarantee the property that, if the graph is invalid, then nodes can spend $O(\log n)$ rounds to produce a solution in which no node outputs $\bot$. This time, the problem $\problem^{\badgraph}$ will satisfy the following weaker guarantee, which will be sufficient for our purposes: nodes can spend $O(\log n)$ rounds to produce a solution in which the graph induced by nodes outputting $\bot$ is a proper instance.
On a high level, the problem will satisfy the following properties.
\begin{itemize}
	\item There are three possible types of output: a node can either produce an empty output ($\bot$), or it can output an error, and there are two types of errors. In the latter case, a node needs to also output a locally checkable proof of the fact that the graph is not a proper instance.
	
	\item If the graph is a proper instance, the constraints of the problem are defined such that the only valid solution to the problem $\problem^{\lbadgraph}$ is the one where all nodes output $\bot$.
	
	\item Otherwise, if the graph is not a proper instance, nodes are able to spend $O(\log n)$ deterministic classical rounds to label invalid parts of the graph, such that the graph induced by nodes outputting $\bot$ is a proper instance, and such that the labeling on the invalid parts of the graph encodes a proof that these parts are indeed invalid.
	
\end{itemize}
We now give a formal definition of the LCL problem $\problem^\lbadgraph$.
The inputs are defined as follows.
\begin{itemize}[noitemsep]
	\item Each node $v$ receives as input a label $g_v \in \VV^{\proper}$.
	
	\item Each half-edge receives as input a label from $\EE^{\proper}$.
	
\end{itemize}
The possible node output labels $\VV^{\lbadgraph}$ for the problem $\problem^\lbadgraph$ are the following.
\begin{itemize}[noitemsep]
	\item The empty output $\bot$.
	\item A pair $(\lerror_{\lintra},x_v)$, where $x_v \in \VV^{\lbadgadget}$. This label will be allowed only on intra-octopus nodes.
	\item The labels $\lerror_{\linter,1}$ and $\lerror_{\linter,2}$. These labels will be allowed only on inter-octopus nodes.
\end{itemize}
We call \emph{error outputs} all labels that are different from $\bot$. Informally, the label $\lerror_{\linter,1}$ will be used by inter-octopus nodes that do not satisfy the constraints $\CC^{\proper}$, while $\lerror_{\linter,2}$ will be used by inter-octopus nodes satisfying that \emph{all} their intra-octopus neighbors output errors.
On a high level, these labels will be used as follows. Inter-octopus nodes that are not properly connected to intra-octopus nodes output $\lerror_{\linter,1}$. Then, intra-octopus nodes connected to inter-octopus nodes that are outputting an error are marked. Then, nodes solve $\problem^{\lbadgadget}$ in the subgraph induced by intra-octopus nodes. Finally, inter-octopus nodes output $\lerror_{\linter,2}$ if they satisfy the following two properties: they did not output $\lerror_{\linter,1}$ previously, and all their neighbors (which are intra-octopus nodes) output an error.
We now define the constraints $\CC^{\badgraph}$.
\begin{myframe}{The constraints $\CC^{\badgraph}$}
	\begin{itemize}
		\item A node can always output $\bot$.\label{cons-badgraph:bot}
		
		\item A node $v$ can output $\lerror_{\linter,1}$ only if $g_v = \inter$ and it does not satisfy the constraints $\CC^{\proper}$.\label{cons-badgraph:error1}

		\item If a node $v$ outputs $(\lerror_{\lintra},x_v)$, then the following must hold. \label{cons-badgraph:error}
		\begin{enumerate}

			\item $x_v$ must be different from $\bot$.
			
			\item Consider the labeling of the graph obtained by input labeling each node $v$ with the pair $(m_v,g_v)$, where $m_v := 1$ if $v$ has a neighbor outputting $\lerror_{\linter,1}$ and $m_v := 0$ otherwise, and output labeling each node $v$ with $x_v$. On node $v$ the constraints $\CC^{\lbadgadget}$ must be satisfied. That is, $x_v$ is a valid solution for $\problem^{\lbadgadget}$ when we mark nodes that are incident to inter-octopus nodes that output $\lerror_{\linter,1}$.
						
		\end{enumerate}
		\item A node $v$ can output $\lerror_{\linter,2}$ only if $g_v = \inter$ and all its neighbors do not output $\bot$.\label{cons-badgraph:error2}
	\end{itemize}
\end{myframe}

\begin{lemma}\label{lem:validgraph}
	Let $G$ be a $(\VV^{\proper}$, $\EE^{\proper})$-labeled graph where $\CC^{\proper}$ is satisfied for all nodes. Then, the only valid solution for $\problem^{\lbadgraph}$ is the one assigning $\bot$ to all nodes.
\end{lemma}
\begin{proof}
	Since $\CC^{\proper}$ is satisfied at all nodes, an inter-octopus node cannot output $\lerror_{\linter,1}$.
	This implies that the input for $\problem^{\lbadgadget}$ satisfies $m_v = 0$ for all intra-octopus nodes $v$.
	Since each connected component of the subgraph induced by intra-octopus nodes is a properly labeled octopus gadget, intra-octopus nodes must output $\bot$.
	This implies that inter-octopus nodes cannot output $\lerror_{\linter,2}$. We thus get that all nodes must output $\bot$.
\end{proof}

\begin{lemma}\label{lem:badgraph}
	Let $G$ be a $(\VV^{\proper}$, $\EE^{\proper})$-labeled graph. There exists a solution for $\problem^{\lbadgraph}$ where each connected component induced by nodes outputting $\bot$ is a proper instance. Moreover, such a solution can be computed in $O(\log n)$ classical deterministic rounds.
\end{lemma}

\begin{proof}
	In order to produce the claimed solution, we will use the algorithm for $\problem^{\lbadgadget}$ reported in the proof of the \cref{lem:badoctopus}.
	The algorithm for solving $\problem^{\badgraph}$ works as follows. 
	
	Each node $v$ satisfying $g_v = \inter$ spends $O(1)$ rounds to check whether it satisfies $\CC^{\proper}$. If these constraints are not satisfied, $v$ outputs $\lerror_{\linter,1}$. Then, each intra-octopus node $v$ spends $O(1)$ rounds to compute $m_v$.
	By \cref{lem:badoctopus}, intra-octopus nodes can spend $O(\log n)$ deterministic classical rounds to compute a solution for $\problem^{\lbadgadget}$, where $(m_v,g_v)$ is the input of node $v$ for the problem. Such a solution satisfies that the following property: if in the subgraph induced by intra-octopus edges, in the connected component containing $v$, there is at least one marked node or a node not satisfying $\CC^{\lgadget}$, then $v$ outputs a label different from $\bot$. For each node $v$, let $x_v$ be this output.
	If $x_v \neq \bot$, then $v$ outputs $(\lerror_{\lintra},x_v)$, otherwise $v$ outputs $\bot$.
	Finally, each inter-octopus node that did not output $\lerror_{\linter,1}$ spends $O(1)$ rounds to check if at least one neighbor output some error in the previous phase, and in that case it outputs $\lerror_{\linter,1}$.

	The solution clearly satisfies the constraints of $\problem^{\lbadgraph}$, and the runtime is clearly $O(\log n)$. We now prove that the subgraph induced by nodes outputting $\bot$ is a proper instance.
	
	Let $G'$ be the subgraph induced by nodes outputting $\bot$. By the definition of the algorithm, each connected component of the subgraph of $G'$ induced by intra-octopus nodes must be a valid octopus gadget. Moreover, the inter-octopus nodes incident to intra-octopus gadgets of $G'$ must not be labeled $\lerror_{\linter,1}$, since otherwise their neighbors $v$ in octopus gadgets would have satisfied $m_v=1$ and hence output something different from $\bot$.
	Also, inter-octopus nodes incident to intra-octopus gadgets of $G'$ cannot be labeled $\lerror_{\linter,2}$, since they have at least one intra-octopus neighbor outputting $\bot$.
	We thus get that inter-octopus nodes that are neighbors of intra-octopus gadgets of $G'$ are all labeled $\bot$.
	Additionally, all left-most leaves of octopus ports must have an inter-octopus neighbor, since otherwise they would have output an error. We thus get that, in $G'$, $\CC^{\proper}$ is satisfied on all nodes, and by \Cref{lem:constraints-to-proper} this implies that $G'$ is a proper instance.
\end{proof}

\subsection{The problem \texorpdfstring{\boldmath $\problem^{\lpromise}$}{Pi promise}}\label{sec:lcl-problem-pipromise}
In this section, we define the problem $\problem^{\lpromise}$ as a function of a given problem $\problem^{\linearizable} = (\Sigma, (F,L,P), B)$ (see \cref{def:Pi_linearizable} for the definition of linearizable problems).
As mentioned in \cref{sec:higl_level-pi_promise}, the problem $\Pi$ will be defined such that it is required to solve $\Pi^{\lbadgraph}$ on the whole graph, and it is required to solve $\Pi^{\lpromise}$ on the subgraph induced by nodes that output $\bot$ for $\Pi^{\lbadgraph}$. 

Recall that the problem $\Pi^{\lbadgraph}$ allows to output $\bot$ on all nodes, even on graphs that are not proper instances. If an algorithm produces such an output, then it is required to solve $\Pi^{\lpromise}$ on a graph that is not a proper instance.
However, when proving a lower bound for $\Pi$, we will use a proper instance, and when proving an upper bound for $\Pi$, we will define an algorithm satisfying that the subgraph induced by nodes outputting $\bot$ for $\Pi^{\lbadgraph}$ is always a proper instance. 
In other words, it does not matter how $\Pi^{\lpromise}$ is defined on graphs that are not proper instances, as it will not affect the lower and the upper bounds that we prove.
For this reason, we will define the constraints $\CC^{\lpromise}$ of $\Pi^{\lpromise}$ under the assumption that the given graph $G$ is a proper instance labeled such that all nodes satisfy $\CC^{\lproper}$. This will make the definition easier to read.
Then, the definition of $\Pi^{\lpromise}$ is lifted to any arbitrary graph as follows.
\begin{itemize}
	\item Let $r = O(1)$ be the distance that each node $v$ needs to inspect to check whether it satisfies $\CC^{\lpromise}$ (on a proper instance).
	\item If all nodes that are within distance $r$ from $v$ (including $v$) satisfy $\CC^{\lproper}$, then $v$ needs to satisfy $\CC^{\lpromise}$.
	\item Otherwise, any output is considered incorrect.
\end{itemize}
In this way, whenever a node $v$ needs to satisfy $\CC^{\lpromise}$, the radius-$r$ neighborhood of $v$ is a subgraph of a proper instance, and hence the constraints of $\CC^{\lpromise}$ are well-defined.


Let us proceed with some nomenclature.
If $v$ is labeled $\gadgetN$ and $L_v(e) \notin \{\llch, \lrch\}$ for all $e$ incident to $v$, then we call $v$ \emph{head-leaf}. If $v$ is labeled $\gadgetP$ and $L_v(e) \neq \lparent$ for all $e$ incident to $v$, then we call $v$ \emph{port-root}. 
For each head-leaf $v$,  let $\eta(v) = \abs{\{ e: L_v(e) \in \{\nplink_1, \nplink_2\} \}} $ count the number of port gadgets that are connected to \(v\) (which are either 1 or 2).

We are now ready to define \(\problem^{\lpromise}\).
The node input set for \(\problem^{\lpromise}\) is \(\Vin^{\lpromise} = \VV^{\lproper}\), while the node-edge input set for \(\problem^{\lpromise}\) is \(\Ein^{\lpromise} = \EE^{\proper}\).
The node output set is \(\Vout^{\lpromise} = \Sigma \cup \{\bot\}\).
Let us denote with \(\oupt(v)\) the output of a node \(v\) for \(\problem^{\lpromise}\). 
The constraints $\CC^{\lpromise}$ of $\problem^\lpromise$ are defined as follows.



\begin{myframe}{The constraints $\CC^{\lpromise}$}
	\begin{enumerate}
		\raggedright
		\item If $v$ is labeled $\gadgetP$, then it must hold that \(\oupt(v) \in \Sigma\).
		\item If $v$ is not labeled $\gadgetP$, then it must hold that \(\oupt(v) = \bot\).
		
		\item If a node $v$ is labeled $\gadgetP$, then, for each neighbor $u$ of $v$ that is also labeled $\gadgetP$, it must hold that $\oupt(v) = \oupt(u)$.
		
		\item If $v$ is a head-leaf node, and $L_v(e) \neq \lleft$ for all $e$ incident to $v$, then it must hold that $\oupt(f(v, \nplink_1)) \in F$.
		
		\item If $v$ is a head-leaf node, and $L_v(e) \neq \lright$ for all $e$ incident to $v$, then:
		\begin{itemize}
			\item if $\eta(v) = 1$, then it must hold that $\oupt(f(v, \nplink_1)) \in L$;
			\item if $\eta(v) = 2$, then it must hold that $\oupt(f(v, \nplink_2)) \in L$.
		\end{itemize}
		
		\item If $v$ is a head-leaf node satisfying that $\eta(v) = 2$, then it must hold that $(\oupt(f(v, \nplink_1)), \oupt(f(v, \nplink_2)) ) \in P$.
		
		\item If $v$ is a head-leaf node satisfying that $\eta(v) = j$ for some $j \in \{1,2\}$, and $f(v, \lright) = u \neq \bot$, then it must hold that $(\oupt(f(v, \nplink_j)), \oupt(f(u, \nplink_1)) ) \in P.$
		

		
		\item If $v$ is labeled $\inter$ and $\{u_1, u_2, \dots, u_k \}$ is the set of all nodes adjacent to $v$, then it must hold that $\{\oupt(u_1), \oupt(u_2), \dots, \oupt(u_k)\} \in B$.
		
	\end{enumerate}
\end{myframe}

We now prove some facts about the relation between \(\problem^{\linearizable}\) and \(\problem^{\lpromise}\). As informally explained in \Cref{sec:intro_proper_instances}, a proper instance can be obtained by replacing the nodes belonging to one side of a bipartite graph with octopus gadgets. We now provide a formal definition of this construction.

\begin{definition}[Lift of a bipartite graph]\label{def:lift:bipartite-graph}
	Let \(G = (W \cup B, E)\) be a bipartite graph that is not necessarily connected and that may contain parallel edges.
	A \emph{lift} of \(G\) is a pair $(L, \compression)$, where \(L = (V', E')\) is a graph that is a proper instance (according to \cref{def:proper-instance}), and $\compression : V' \to W \cup B \cup E$ is a function, with the following properties.
	\begin{enumerate}
		\item For each node \(v\) of any head gadget, \(\compression(v) \in W\).
		\item If \(u,v \in V'\) are nodes of the same head gadget, then \(\compression(u) = \compression(v)\).
		\item If \(u,v \in V'\) are nodes of different head gadgets, then \(\compression(u) \neq \compression(v)\).
		\item For any node \(v\) of any port gadget, \(\compression(v) \in E\).
		\item If \(u,v \in V'\) are nodes of the same port gadget, then \(\compression(u) = \compression(v)\).
		\item If \(u,v \in V'\) are nodes of different port gadgets, then \(\compression(u) \neq \compression(v)\).
		\item Let \(u \in V'\) be a node of a port gadget $P$, and let \(v \in V'\) be a node of a head gadget $H$. The edge \(\compression(u)\) is incident to the node \(\compression(v)\) if and only if $P$ is connected to $H$.
		\item If \(v \in V'\) is an inter-octopus node, then \(\compression(v) \in B\).
		\item Let \(u \in V'\) be a node of a port gadget $P$, and let \(v \in V'\) be an inter-octopus node. The edge \(\compression(u)\) is incident to the node \(\compression(v)\) if and only if $P$ is connected to $v$.
	\end{enumerate}
	Furthermore, for a proper instance \(L\) and a bipartite graph \(G\), if there exists a function $\compression$ for which $(L, \compression)$ is a lift of \(G\), then we say that \emph{\(L\) can be compressed} to \(G\).
\end{definition}

In the following, for a bipartite graph \(G = (W \cup B, E)\), by \emph{ordering assigned to the edges of $G$} we denote a function $\tau : w \mapsto \sigma_w$ that takes as input a node $w \in W$ and returns a function $\sigma_w$, where the function $\sigma_w$ gives an ordering of the incident half-edges of $w$. That is, $\tau$ assigns, to each node $w$, an injective function $\sigma_w$ mapping each edge incident to $w$ to an integer in $\{1,\ldots,\deg(w)\}$, where \(\deg(w)\) is the degree of \(w\).

\begin{observation}\label{obs:compression-implies-local-ordering}
	Let \(G = (W \cup B, E)\) be a bipartite graph, and let \((L = (V', E'), \compression)\) be a lift of $G$ that is \((\VV^\proper, \EE^\proper)\)-labeled such that $\CC^{\lproper}$ is satisfied on all nodes.
	Then, the function \(\compression : V' \to W \cup B \cup E\) uniquely defines, for each node $w \in W$, the ordering $\sigma_w$ of the edges incident to $w$.
\end{observation}
\begin{proof}
	For each node $w \in W$, let $G_w$ be the octopus gadget whose nodes map either to \(w\) or to its incident edges according to the function $\compression$.
	We now define an ordering of the port gadgets of \(G_w\) which will automatically induce an ordering of the edges incident to \(w\), and hence an ordering assigned to the edges of $G$.
	Assume that \(G_w\) has height \(h\).
	Consider two port gadgets \(P\) and \(P'\) of \(G_w\), and let \(r_P\) and \(r_{P'}\) be their root nodes, respectively.
	Furthermore, let \(v_{P}\) and \(v_{P'}\) be the head-leaf nodes of \(G_w\) that are connected to \(r_P\) and \(r_{P'}\), respectively.
	Assume that \(v_P = (h-1,i)\) and \(v_{P'} = (h-1,j)\) for some \(i,j \in \{0, \ldots, 2^{h-1}\}\).
	We say that \(P < P'\) if and only if \(i < j\) or, if \(i = j\), \(L_{v_P}(\{v_P, r_P\}) = \nplink_1\) and \(L_{v_{P}}(\{v_{P}, r_{P'}\}) = \nplink_2\).
	Let $P_i$ be the $i$-th port of $G_w$ according to this ordering, and let $u$ be an arbitrary node of $P_i$. The $i$-th edge of $w$ is the edge $\compression(u)$.
\end{proof}

In the definition of \(\problem^{\linearizable}\), the given bipartite graph \(G = (W \cup B, E)\) comes with a local ordering \(\sigma_w\) of the edges of each node \(w \in W\).
We can then think of lifts that ``respect'' such an ordering.

\begin{definition}[Edge-order preserving lift]
	Let \(G = (W \cup B, E)\) be a bipartite graph that is not necessarily connected and that may contain parallel edges, and let $\tau : w \mapsto \sigma_w$ be an ordering assigned to the edges of $G$.
	Let \((L,\compression)\) be a lift of \(G\), where \(L\) is \((\VV^\proper,\EE^\proper)\)-labeled.
	Let \(\tau' : w \mapsto \sigma'_w \) be a function assigning to each \(w \in W\) the edge ordering defined via \(\compression\) according to \cref{obs:compression-implies-local-ordering}.
	We say that \((L,\compression)\) is edge-order preserving with respect to \(\tau \) if \(\sigma'_w = \sigma_w\) for all \(w \in W\).
\end{definition}

\begin{lemma}\label{lemma:lift-well-defined}
	For each bipartite graph \(G = (W \cup B, E)\) which comes with an ordering \(\tau: w \in W \mapsto \sigma_w\) assigned to its edges, and for each integer $h$, there exists a lift \((L, \compression)\) that is \((\VV^\proper, \EE^\proper)\)-labeled such that all nodes satisfy $\CC^\lproper$ where all port gadgets of $L$ have height $h$, and such that \((L, \compression)\) is edge-order preserving with respect to \(\tau\).
	Moreover, for each proper instance \(L\), there exists a bipartite graph \(G\) such that \(L\) can be compressed to $G$.
\end{lemma}
\begin{proof}
	We first prove that for each bipartite graph \(G = (W\cup B, E)\) with an edge-ordering \(\tau\) and for each integer $h$, there exists a lift \((L, \compression)\) where all port gadgets of $L$ have height $h$ and that is edge-order preserving w.r.t.\ \(\tau\). Recall that to each node $w \in W$ is assigned an ordering $\sigma_w$ of its incident edges.
	
	For each node $w \in W$, we create an $(x_w,\eta_w,W_w)$-octopus gadget $G_w$, where the parameters $x_w$, $\eta_w$, and $W_w$ are chosen as follows.
	\begin{itemize}
		\item $x_w = 2^{\floor{\log_2(\deg(w))}} + 1$.
		\item $\eta_w \in \{1,2\}^{x_w-1}$ is a vector satisfying that $\sum \eta_w(i) = \deg(w)$. Such a vector exists, since $(x_w-1) \le \deg(w) \le 2(x_w-1)$.
		\item All elements of $W_w$ are equal to $h$.
	\end{itemize}
	Let \(E_w = \{e_{(w,1)}, \dots, e_{(w,\deg(w))}\}\) be the set of edges that are incident to \(w\), ordered according to $\sigma_w$.
	For each $1 \le i \le \deg(w)$, let \(P_{e_{(w,i)}}\) be the $i$-th port gadget connected to the octopus gadget $G_w$ (according to the natural order defined in the proof of \cref{obs:compression-implies-local-ordering}).
	For each black node \(b \in B\), we create a node \(v_b\) that is an inter-octopus node. If \(b\) is connected to \(w\) through some edge \(e_{(w,i)}\), then \(v_b\) is connected to the left-most leaf of \(P_{e_{(w,i)}}\). 
	It is trivial to see that \(L\) admits a \(\compression\) function to \(G\) with the properties given in \cref{def:lift:bipartite-graph}.

	We now prove that each proper instance \(L = (V,E)\) can be compressed to some bipartite graph \(G = (W \cup B, E')\). We construct $G$ as follows.
	The set of white nodes \(W\) contains a vertex \(v_H\) for each head gadget \(H\).
	The set of black nodes \(B\) contains a node \(u_{v_{\inter}}\) for each inter-octopus node \(v_{\inter}\).
	Then, \(v_H\) is connected to \(u_{v_{\inter}}\) through an edge if and only if \(H\) is connected to a port gadget \(P\) that is, in turn, connected to \(v_\inter\).
\end{proof}

\begin{lemma}\label{lem:mapping-linearizable-promise}
	Let $\Pi^{\lpromise}$ be the problem defined as a function of a linearizable problem $\Pi^{\linearizable}$ as described in \Cref{sec:lcl-problem-pipromise}.
	Let \(L = (V',E')\) be a proper instance, and let $G  = (W \cup B, E)$ be the bipartite graph to which $L$ can be compressed to, according to \Cref{lemma:lift-well-defined}. For each node $w \in W$, let $\sigma_w$ be the ordering defined in \Cref{obs:compression-implies-local-ordering}.

	Then, there is a bijective function that maps solutions of \(\problem^{\lpromise}\) on \(L\) to solutions of \(\problem^{\linearizable}\) on \(G\) (with the aforementioned edge ordering).
\end{lemma}
\begin{proof}
	Let \(\compression : V' \to W \cup B \cup E\) be the function that compresses \(L\) to \(G\).
	First, suppose that we are given a solution for \(\problem^{\lpromise}\) on \(L\).
	Consider any white node \(w\) of \(G\) and its set of incident edges \(e_{(w,1)}, \ldots, e_{(w,\deg(w))}\) ordered according to $\sigma_w$.
	Consider an arbitrary node \(v\) of a port gadget \(P\) such that \(\compression(v) = e_{(w,i)}\).
	Then, the output of $w$ on the edge \(e_{(w,i)}\) is defined to be the same as the output of \(v\).
    This labeling clearly satisfies the constraints of $\Pi^{\linearizable}$ under the ordering of the edges defined by $\sigma = \{\sigma_w \mid w \in W \}$.

	Now, suppose we are given a solution for \(\problem^{\linearizable}\) on \(G\).
	Consider any white node \(w\) of \(G\) and its set of incident edges \(e_{(w,1)}, \ldots, e_{(w,\deg(w))}\),  ordered according to $\sigma_w$.
	For all nodes \(v \in \compression^{-1}(e_{(w,i)})\), we define the output of \(v\) to be the same as the output of $w$ on \(e_{(w,i)}\).
	For all other nodes \(v\) of \(G\), we define the output of \(v\) to be \(\bot\).
	This labeling clearly satisfies the constraints of $\Pi^{\lpromise}$.
\end{proof}

\subsection{The LCL problem \texorpdfstring{\boldmath $\problem$}{Pi}}\label{sec:lcl-problem}
We are now ready to define our problem $\problem$, as a function of the given problem $\problem^{\lpromise}$.
The input labels of $\problem$ are the same as the input labels $\VV^{\lproper}$ and $\EE^{\lproper}$ of $\problem^{\lbadgraph}$.
The set $\VV^{\problem}$ of output labels contains the following labels.
\begin{itemize}[noitemsep]
	\item $(\badgraph,x)$, for all $x \in \VV^{\lbadgraph} \setminus \{\bot\}$.
	\item $(\lpromise,x)$, for all $x \in \Vout^{\lpromise}$.
\end{itemize}
The constraints $\CC^{\problem}$ are defined as follows.
\begin{myframe}{The constraints $\CC^{\problem}$}
	\begin{itemize}
		\item Consider the graph labeling in which the output labels are changed as follows: nodes labeled $(\lpromise,x)$ become labeled $\bot$, while nodes labeled $(\badgraph,x)$ become labeled $x$. Then, this labeling must satisfy the constraint $\CC^{\lbadgraph}$.
		\item Let $G'$ be the subgraph induced by nodes outputting $(\lpromise,x)$, for some $x$. Consider the labeling of $G'$ in which nodes labeled $(\lpromise,x)$ become labeled $x$. Then, all nodes of $G'$, in the subgraph $G'$, must satisfy $\CC^{\lpromise}$.
	\end{itemize}
\end{myframe}
In other words, if nodes output a $(\badgraph,\cdot)$ label, then this labeling must be a valid solution for $\problem^{\lbadgraph}$. This allows nodes to mark invalid parts of the graph. Then, in the parts that are not marked as invalid, nodes need to solve $\problem^{\lpromise}$.

\subsection{Upper bound in quantum-LOCAL}\label{sec:ub-quantum}
In this section, we prove an upper bound on the quantum complexity of $\problem$ as a function of the quantum complexity of $\problem^{\linearizable}$.
More in detail, we prove the following.
\begin{lemma}\label{lem:ub-padding}
	Let $T(n)$ be an upper bound on the quantum complexity of $\problem^{\linearizable}$, that holds also if the given graph contains parallel edges. Then, the quantum complexity of $\problem$ is upper bounded by $O(T(n) \log n)$.
\end{lemma}
\begin{proof}
	Let $G$ be the graph in which we need to solve $\problem$.
	The algorithm is composed of two parts, a classical deterministic part and a quantum part. The first part uses \Cref{lem:badgraph} to compute a solution for $\problem^{\lbadgraph}$ such that the nodes outputting $\bot$ induce a proper instance. This requires $O(\log n)$ classical deterministic rounds. For each $v$, let $x_v$ be its output for $\problem^{\lbadgraph}$.
	Each node $v$ satisfying $x_v \neq \bot$ outputs $(\badgraph,x_v)$ for $\problem$. Note that the graph $G'$ induced by nodes that still did not output anything for $\problem$ is a proper instance.

	The second part of the algorithm works as follows. Nodes that have already produced an output for $\problem$ do nothing. The other nodes simulate the quantum algorithm $\mathcal{A}$ for $\problem^{\linearizable}$ as follows. The goal is to simulate $\mathcal{A}$ in the graph obtained by contracting each octopus gadget of $G'$ into a single node. Let $\hat{G}$ be this graph. Observe that this is the graph that $G$ can be compressed to, according to \Cref{lemma:lift-well-defined}.
	For each octopus gadget $\nu$, the root of its head gadget is the node responsible for storing the quantum state of the node in $\hat{G}$ corresponding to $\nu$. The inter-octopus nodes are unaffected and simulate themselves in $\hat{G}$.
	Since the diameter of a valid octopus gadget is clearly upper bounded by $O(\log n)$, we get that the communication between nodes of $\hat{G}$ can be simulated in $O(\log n)$ quantum rounds in $G'$.
	Hence, after $O(T(n) \log n)$, all the roots of the head gadget computed a solution for $\problem^{\linearizable}$ in $\hat{G}$. With additional $O(\log n)$ steps nodes can solve $\problem^{\lpromise}$ in $G'$, according to the mapping defined in \Cref{lem:mapping-linearizable-promise}. For each node $v$ in $G'$, let $x_v$ be its output for $\problem^{\lpromise}$.
	Each node $v$ in $G'$ outputs $(\lpromise,x_v)$ for $\problem$. 
	
	The output clearly satisfies the constraints $\CC^{\problem}$, and the runtime is clearly upper bounded by $O(T(n) \log n)$.
\end{proof}

\subsection{Lower bound in LOCAL}\label{sec:lb-local}
In this section, we prove a lower bound on the classical randomized complexity of $\problem$ as a function of the classical randomized complexity of $\problem^{\linearizable}$. More in detail, we prove the following.
\begin{lemma}\label{lem:lb-padding}
	Let $T(n)$ be a lower bound on the time required to solve $\problem^{\linearizable}$ with a classical randomized algorithm that has failure probability at most $1/n$. Then, any classical randomized algorithm for $\problem$ with failure probability at most $1/n$ requires $\Omega(T(n^{1/3}) \log n)$ rounds.
\end{lemma}
\begin{proof}
	For a contradiction, assume that there exists an algorithm $\mathcal{A}$ that solves $\problem$ in $T'(n) = o(T(n^{1/3}) \log n)$ rounds with failure probability at most $1/n$. We show that we can construct an algorithm $\mathcal{B}$ that solves $\problem^{\linearizable}$ in $o(T(n))$ with failure probability at most $1/n$, contradicting the hypothesis.

	The algorithm $\mathcal{B}$ works as follows.
	Let $G$ be an instance of $\problem^{\linearizable}$. We construct an instance $G'$ of $\problem$ by applying \Cref{lemma:lift-well-defined} to construct a lift of  $G$ with parameter $h = \Theta(\log n)$. 
	On a high level, the algorithm $\mathcal{B}$ will be defined such that the nodes of $G$ simulate the execution of $\mathcal{A}$ on $G'$.
	
	Recall that $G'$ is constructed as follows. Each hyperedge $e$ of $G$ is replaced with an inter-octopus node $b_e$. 
	Each node $v$ of degree $d_v$ is replaced with an octopus gadget $\nu_v$ that contains exactly $d_v$ ports, such that each port gadget has height $\Theta(\log n)$, and such that the number of nodes of each port gadget is strictly less than $n$. 
	For each node-hyperedge pair $(u,e)$ of $G$, we add an edge to $G'$ that connects the left-most leaf of the $i$-th port of $\nu_v$ to $b_e$, assuming that $e$ is the $i$-th edge incident to $u$ according to the given ordering $\sigma_u$. We label $G'$ such that it is a properly labeled valid instance. Observe that the labeling of the nodes of $\nu_v$ can be computed by $v$ without communication.

	Since each port gadget has size strictly less than $n$, and since the maximum degree of $G$ is $n$, we obtain that $G'$ has at most $n^3$ nodes. Observe that if two nodes $(u,v)$ are neighbors in $G$, then the head gadgets of $\nu_v$ and $\nu_u$ are at distance $\Theta(\log n)$ in $G'$. 
	Hence, each node $v$ of $G$ can spend $O(T'(n^3) / \log n)$ rounds of communication in $G$ to gather its $T'(n^3)$-radius neighborhood in $G'$. 
	Since the runtime of $\mathcal{A}$ on instances of size at most $n^3$ is upper bounded by $T'(n^3)$, and since $G'$ has size at most $n^3$, 
	then each node $v$ is able to compute the output of $\mathcal{A}$ for $\problem$ on $G'$ without further communication. 
	
	Since $G'$ is a properly labeled valid instance, all nodes of $G'$ must output labels of type $(\lpromise,\cdot)$. Hence, $v$ can reconstruct a solution for $\problem^{\linearizable}$ as a function of the output of the nodes in $\nu_v$, without communication, according to \Cref{lem:mapping-linearizable-promise}. 

	The runtime of the algorithm $\mathcal{B}$ for $\problem^{\linearizable}$ is $O(T'(n^3) / \log n) = o(T(n))$, and its failure probability is upper bounded by $1/n^3 \le 1/n$, proving the claim.
\end{proof}

\subsection{Instantiating the construction}
In \cite{balliu2024quantum}, the authors introduced a family of problems $\{\iterghz(\Delta) \mid \Delta \ge 3 \}$. 
While this set can be seen as a family of problems parametrized by the maximum degree $\Delta$ of the graph, for the ease of presentation we will present this family as a single problem that we call \emph{iterated GHZ}.
We now formally define this problem, show that it can be defined as a linearizable problem, and then state its quantum and classical complexities.

\paragraph{The iterated GHZ problem.}
An instance of the problem is a hypergraph of rank $3$. However, for the ease of presentation, in the following, we will describe an input instance as a bipartite graph $(U,V,E)$, where nodes in $V$ have degree exactly $3$ (equivalently, in order to remove this assumption, if a node in $V$ has degree different from $3$, then it is unconstrained). Nodes in $U$ are called \emph{white} nodes, or \emph{players}, while nodes in $V$ are called \emph{black} nodes, or \emph{games}.
Each node $u \in U$ receives as input an ordering on its incident edges. Let $e_{u,i}$ be the $i$-th edge incident to $u$ according to this order, where $1 \le i \le d_u$, and $d_u$ is the degree of $u$.
The problem is defined as follows.
Each node $u \in U$ must output two bits $x(u,i)$ and $y(u,i)$ on each incident edge $e_{u,i}$. The following must hold.
\begin{itemize}
	\item For each white node $u$, it must hold that $x(u,1) = 0$.
	\item For each white node $u$, it must hold that $y(u,i) = x(u,i+1)$, for all $i$.
	\item Let $v$ be a black node, and let $u_1$, $u_2$, $u_3$ be its white neighbors. Let $i_j$ be the position of the edge $\{u_j,v\}$ according to the order of the edges incident to $u_j$, for $1 \le j \le 3$.
	Then, the following must hold.
	\begin{itemize}
		\item If $i_1 = i_2 = i_3 = 1$, then it must hold that the multiset $\{y(u_1,1), y(u_2,1), y(u_3,1)\}$ is equal to $\{0,0,1\}$.
		\item Otherwise, if $x(u_1,i_1) + x(u_2,i_2) + x(u_3,i_3)$ is even, then it must hold that $y(u_1,i_1) \oplus y(u_2,i_2) \oplus y(u_3,i_3) = x(u_1,i_1) \lor x(u_2,i_2) \lor x(u_3,i_3)$.
		\item Otherwise, node $v$ is unconstrained.
	\end{itemize}
\end{itemize}
The original definition of $\iterghz(\Delta)$ of \cite{balliu2024quantum} is equivalent to the definition that we have provided, when restricted to the case in which the maximum degree of the graph is $\Delta$.

\paragraph{Lower bound for classical randomized algorithms.}
In \cite{balliu2024quantum}, the following lower bound is shown.
\begin{theorem}[\cite{balliu2024quantum}]
	Let $\Delta \ge 3$ be an integer. Any classical randomized algorithm that solves $\iterghz(\Delta)$ with high probability requires $\Omega(\min\{\Delta, \log_\Delta \log n\})$ rounds.
\end{theorem}
In order to obtain the best possible lower bound stated solely as a function of $n$, we take $\Delta = \Theta\left(\frac{\log \log n}{\log \log \log n}\right)$, and we obtain the following.
\begin{corollary}[\cite{balliu2024quantum}]\label{lb:iterghz}
	Any classical randomized algorithm that solves iterated GHZ with high probability requires $\Omega\left(\frac{\log \log n}{\log \log \log n}\right)$ rounds.
\end{corollary}

\paragraph{Upper bound for quantum algorithms.}
In \cite{balliu2024quantum}, the lower and upper bounds for the problem $\iterghz(\Delta)$ are proved in a setting in which the input instances are $\Delta$-edge colored. 
In the construction described in this section, however, we would need to provide a different color to each port gadget belonging to the same octopus gadget. This would require $\omega(1)$ input labels, which is not allowed in a proper LCL.

Since the lower bound of \cite{balliu2024quantum} holds in a $\Delta$-edge colored graph, this lower bound clearly holds in graphs that are not $\Delta$-edge colored (having such a coloring can only make the problem potentially easier). However, for an upper bound, we need to adapt the algorithm of \cite{balliu2024quantum} to the case in which the edge coloring is not provided. Moreover, in order to use \Cref{lem:ub-padding}, we need an algorithm that works also in the case of parallel edges.

\begin{lemma}\label{ub:iterghz}
	The iterated GHZ problem can be solved in $O(1)$ rounds with a quantum algorithm, even if parallel edges are present.
\end{lemma}
\begin{proof}
	The algorithm works in two rounds of communications, which are described in the following two points.
	\begin{enumerate}
		\item Each white node $w$ does the following. Let $b_i$ be the $i$-th neighbor of $w$. Node $w$ sends $i$ to $b_i$, for all \(1 \le i \le d_w\), where \(d_w\) is the degree of \(w\).
		\item Each black node \(b\) has now received three port numbers \(p_v,p_u,p_z\) from its three white neighbors \(v,u,z\). 
		It then prepares three qubits \(q_v,q_u,q_z\) in the GHZ state, i.e.,
		\[
			\ket{0}_{q_v}\ket{0}_{q_u}\ket{0}_{q_z} \to \frac{1}{\sqrt{2}}\left(\ket{0}_{q_v}\ket{0}_{q_u}\ket{0}_{q_z}+ \ket{1}_{q_v}\ket{1}_{q_u}\ket{1}_{q_z}\right).
		\]
		If \(p_v = p_u = p_z = 1\), then $b$ sends the pair \((q_v,0)\) to \(v\), the pair \((q_u,0)\) to \(u\), and the pair \((q_z,1)\) to \(z\).
		Otherwise, it just sends the qubit \(q_v\) to \(v\), the qubit \(q_u\) to \(u\), and the qubit \(q_z\) to \(z\).
	\end{enumerate}
	Now, all white nodes decide on their outputs without any further communication.
	First, all white nodes \(w\) set \(x(w,1) = 0\).
	Now, observe that each white node \(w\) has received a qubit \(q_{w,i}\) through its \(i\)-th incident edge. 
	Furthermore, some white nodes $w$ have received a pair \((q_{w,1},j_w)\), where \(q_{w,1}\) is a qubit and \(j_w \in \{0,1\}\): we call such nodes \emph{lucky}.
	Each lucky white node \(w\) sets \(y(w,1) = j_w\) and $x(w,2) = j_w$ (if $d_w \ge 2$).
	Now we describe what each white node \(w\) outputs for the \(j\)-th incident edge. If a node is lucky, it does the following for each $j \ge 2$ iteratively, while if a node is unlucky it does the following for each $j \ge 1$ iteratively. 
	Each node $w$ measures the qubit \(q_{w,j}\) according to the well-known GHZ strategy (with input $x(w,j)$) that wins the game with probability 1 \cite{Brassard2004QuantumP}, and sets the output \(y(w,j)\) accordingly.
	Then, node $w$ sets $x(w,j+1) = y(w,j)$ (if $j+1 \le d_w$).

	We now argue about the correctness of the algorithm.
	White constraints are satisfied by construction of the algorithm.
	Suppose now that the constraints of some black node \(b\) is not satisfied. 
	This means that the three white neighbors \(v,u,z\) of \(b\) have some outputs \((x_v,y_v),(x_u,y_u),(z_u,y_u)\) that do not satisfy the constraints of the problem.
	We have three cases:
	If the three port numbers \(p_v,p_u,p_z\) are such that \(p_v = p_u = p_z = 1\), then, by point 2, the white nodes are lucky and have set the outputs to be \((0,0),(0,0),(0,1)\), which is a valid solution.
	Suppose that at least a port number between \(p_v, p_u,p_z\) is different from 1.
	If \(x_v + x_u + x_z\) is odd, then any solution is valid.
	Suppose now that \(x_v + x_u + x_z\) is even. 
	Then, the outputs of the white nodes are constructed by measuring the three entangled qubits \(q_v,q_u,q_z\) in the GHZ state prepared by \(b\).
	By the well-known GHZ winning strategy \cite{Brassard2004QuantumP}, the white nodes' outputs are such that \(y_v \oplus y_u \oplus y_z = x_v \lor x_u \lor x_z\), which is a valid solution.
\end{proof}

\paragraph{The iterated GHZ problem as a linearizable problem.}
We define a problem $\problem^{\linearizable} = (\Sigma,(F,L,P),B)$ that encodes iterated GHZ as a linearizable problem. On a high level, each label will encode a pair $(x(v,i),y(v,i))$ of bits, and we will use special labels for the case in which $i=1$.
We define the set of possible first labels as $F = \{(\lfirst, y) \mid y \in \{0,1 \}\}$.
Then, we define the set of possible output labels as $\Sigma = F \cup \{(\lother, x, y) \mid x,y \in \{0,1\}\}$, and we define $L = \Sigma$.
The allowed pairs $P$ are defined as all the pairs satisfying the following.
\begin{itemize}
	\item All pairs $((\lfirst, y_1),(\lother,x_2,y_2))$ where $y_1 = x_2$ and $y_1,x_2,y_2 \in \{0,1\}$.
	\item All pairs $((\lother, x_1, y_1),(\lother,x_2,y_2))$ where $y_1 = x_2$ and $x_1,y_1,x_2,y_2 \in \{0,1\}$.
\end{itemize}
Let $\{L_1,L_2,L_3\}$ be a multiset containing three elements from $\Sigma$.
If $L_i = (\lfirst,y_i)$ for some $y_i \in \{0,1\}$, then let $\ell_i = \lfirst$ and let $x_i = 0$.
If $L_i = (\lother,x_i,y_i)$ for some $x_i,y_i \in \{0,1\}$, then let $\ell_i = \lother$.
Observe that now, for all $1 \le i \le 3$, $\ell_i$, $x_i$ and $y_i$ are defined. 
The set $B$ contains the multisets $\{L_1,L_2,L_3\}$ satisfying the following.
\begin{itemize}
	\item All the multisets satisfying $\ell_1 = \ell_2 = \ell_3 = \lfirst$ such that the multiset $\{y_1,y_2,y_3\}$ is equal to $\{0,0,1\}$.
	\item All the multisets where, if $\lother \in \{\ell_1,\ell_2,\ell_3\}$, and $x_1 + x_2 + x_3$ is even, then $y_1 \oplus y_2 \oplus y_3 = x_1 \lor x_2 \lor x_3$.
	\item All the multisets where $\lother \in \{\ell_1,\ell_2,\ell_3\}$ and $x_1 + x_2 + x_3$ is odd.
\end{itemize}
It is clear that $\problem^{\linearizable}$ and iterated GHZ are equivalent:
\begin{itemize}
	\item both problems require outputting two bits per edge;
	\item the black constraint of iterated GHZ has the same requirements of the set $B$ on the pairs incident to a black node;
	\item both the white constraint of iterated GHZ and the triple $(F,L,P)$ requires that the second bit of a port is the same as the first bit of the next one.
\end{itemize}

\paragraph{Putting things together.}
By combining the problem $\problem^{\linearizable}$ obtained as a function of iterated GHZ with \Cref{lb:iterghz,ub:iterghz,lem:lb-padding,lem:ub-padding}, we obtain the following result.
\begin{theorem}\label{thm:main-result}
	There exists an LCL problem $\problem$ with quantum complexity $O(\log n)$ and classical randomized complexity $\Omega\left(\log n \cdot \frac{\log \log n}{\log \log \log n}\right)$.
\end{theorem}
\section*{Acknowledgements}

This work was supported in part by the Research Council of Finland, Grants 359104 and 363558, by the MUR (Italy) Department of Excellence 2023--2027 for GSSI.
Francesco d’Amore is supported
by the project Decreto MUR n. 47/2025, CUP: D13C25000750001.
This project was initiated at the Research Workshop on Distributed Algorithms (RW-DIST 2025) in Freiburg, Germany; we would like to thank all workshop participants and organizers for inspiring discussions.
We would also like to thank Sebastian Brandt, Xavier Coiteux-Roy, François Le Gall, Augusto Modanese, Marc-Olivier Renou, Ronja Stimpert, Lucas Tendick, and Isadora Veeren for discussions that led to this project.
\printbibliography
\newpage
\appendix
\section{Simulation of component-wise \olocal in \local}
\label{app:reduction}

In this section, we extend the result of \cref{sec:quantum-advantage} to hold also for the \emph{component-wise randomized \olocal model}.
We start by briefly discussing the state of the art and motivating the study of yet another variant of the \local model.
We then define formally what we mean by component-wise randomized \olocal.
Finally, we extend the proofs of \cref{sec:quantum-advantage} to hold also in this new model.

\subsection{Motivation}

Informally, component-wise randomized \olocal is just like the component-wise deterministic \olocal model (see \cite{akbari2024,damore2025}), but generalized for the randomized case.
In particular, a deterministic component-wise \olocal algorithm is also a randomized one.

The component-wise deterministic online-\local model was originally introduced because it allows simulation of randomized \olocal in weaker models like \slocal in rooted trees and forests~\cite{akbari2024}.
In particular, randomized \olocal captures the power of \nonsign distributions~\cite{akbari2024}, \slocal~\cite{ghaffari2017} and dynamic-\local~\cite{akbari_et_al:LIPIcs.ICALP.2023.10}.
Hence, proving lower bounds for randomized online-\local allows us to provide strong lower bounds that span across a wide variety of models.

One way to prove these lower bounds is by providing a simulation of stronger models in weaker ones.
One such model is the component-wise randomized \olocal, which is stronger than both the deterministic component-wise \olocal and the bounded-dependence model.

Extending the simulation result of \cref{sec:quantum-advantage} to component-wise randomized \olocal model improves the state of the art in two major ways:
\begin{enumerate}
    \item So far the only result known for component-wise deterministic \olocal is the simulation result of \textcite{akbari2024} in rooted trees and forests.
    Our result works in all graphs, albeit with a worse locality.

    \item We provide a direct reduction from the bounded-dependence model to the component-wise randomized \olocal.
    Previously, the reduction had to go through the \nonsign model, randomized \olocal and component-wise deterministic \olocal.
    This round-trip through randomized \olocal incurs a doubly exponential cost in the locality.
    For localities in the $\Omega(\log \log n)$ region, our result is asymptotically tighter.
\end{enumerate}

\subsection{Definitions}

We now formally define what we mean by component-wise randomized \olocal.
To do this, we first define the notion of \emph{partial \olocal run of length~$\ell$}, as introduced by \textcite{akbari2024}:
\begin{definition}[Partial \olocal run of length~$\ell$]
    \label{def:partial-online-local-run}
    Let $G$ be a graph with an order of nodes~$v_1, v_2, \ldots, v_n$, and let $\localVar$ encode the input data on each vertex.
    Consider the subgraph $G_\ell \subseteq G$ induced by the radius-$T$ neighborhoods of the first $\ell$ nodes $v_1, \ldots, v_\ell$, i.e., $G_\ell = \mathring{G}[\neighborhood_{T}[\{v_1, \ldots, v_\ell\}]]$.
    We call~$(G_\ell, (v_1, \ldots, v_\ell), \localVar \restriction_{\neighborhood_T[\{v_1, \ldots, v_\ell\}]})$ the \emph{partial \(T\)-rounds \olocal run of length~$\ell$ of~$G$}.
    
    We denote by $({G}_\ell, (v_1, \ldots, v_\ell), \localVar \restriction_{\neighborhood_T[\{v_1, \ldots, v_\ell\}]})[v_\ell]$ the tuple $(\bar{G}_\ell, (w_1, \ldots, w_k), \localVar \restriction_{\neighborhood_T[\{w_1, \ldots, w_k\}]})$ where
    $\bar{G}_\ell \subseteq G_\ell$ is the connected component containing~$v_\ell$,
    $(w_1, \ldots, w_k)$ is the maximal subsequence of $(v_1, \ldots, v_\ell)$ of nodes that belong to~$\bar{G}_\ell$,
    and $\localVar \restriction_{\neighborhood_T[\{w_1, \ldots, w_k\}]}$ is the restriction of the local inputs to this component.
    In such case, $k$ is the length of $({G}_\ell, (v_1, \ldots, v_\ell), \localVar \restriction_{\neighborhood_T[\{v_1, \ldots, v_\ell\}]})[v_\ell]$.
    Notice that $w_k = v_\ell$.

\end{definition}

Now we are ready to introduce the \olocal model.
We loosely follow the definitions given in by \textcite{akbari2024}.
\paragraph{The \olocal model.}
The deterministic \olocal model was first introduced in~\cite{akbari_et_al:LIPIcs.ICALP.2023.10}.
It is a centralized model of computing where the algorithm initially knows only the number of nodes~$n$ of the input graph $G$.
The nodes are processed with respect to an adversarial input ordering sequence $\sigma = v_{1}, v_{2}, \dots, v_{n}$.
Like in the \local model, each node \(v\) has some local state variable \(\localVar(v)\) which encodes input data.
In an algorithm with locality \(T\), the output of node~$v_{i}$ depends on the partial \(T\)-rounds \olocal run of length \(i\) of \(G\), that is, \((G_i,(v_1,\dots,v_i),\localVar \restriction_{\NN_T[\{v_1, \dots, v_i\}]})\) (see \cref{def:partial-online-local-run}).
In the \emph{randomized} \olocal model~\cite{akbari2024}, every node has access to a private infinite random bit string, which is independent of the random bit strings of other nodes.
Trivially, thanks to global memory, node \(v_i\) knows the random bit strings of \(\NN_T[\{v_1, \dots, v_i\}]\).
In the randomized case, the solution to a problem should be correct with probability at least~\(1 - 1/\poly(n)\).

We now give a formal definition of the component-wise \olocal model.

\paragraph{The component-wise \olocal model.}
The component-wise \olocal model is similar to the \olocal model above, but with some further restrictions.
More specifically, let $(G, (v_1, \dots, v_n), \localVar_G)$ and $(H, (u_1, \dots, u_n), \localVar_H)$ be two $n$-node input instances.
Let $i, j \in [n]$ be two indices such that
\begin{align*}
    (G_i, (v_1, \ldots, v_i), \localVar_G \restriction_{\neighborhood_T[\{v_1, \ldots, v_i\}]})[v_i] &=
    (\bar{G}_i, (v_1', \ldots, v_k'), \localVar_G \restriction_{\neighborhood_T[\{v_1', \ldots, v_k'\}]}) \quad\text{ and}\\
    (H_j, (u_1, \ldots, u_j), \localVar_H \restriction_{\neighborhood_T[\{u_1, \ldots, u_j\}]})[u_j] &=
    (\bar{H}_j, (u_1', \ldots, u_k'), \localVar_H \restriction_{\neighborhood_T[\{u_1', \ldots, u_k'\}]})
\end{align*}
are isomorphic with the following properties: the isomorphism preserves the order, i.e., it brings $v_h'$ to $u_h'$ for every $h \in [n]$, and the isomorphism also preserves the inputs $\localVar_G$ and $\localVar_H$ in these neighborhoods.
Then, the algorithm must produce the same output for~$v_i$ and~$u_j$.

In the \emph{randomized} component-wise \olocal model, each node is given access to a private infinite random bit string, which is independent of the random bit strings of other nodes.
When the algorithm needs to output a label for node~\(v_i\), it can only use the random bit strings of nodes in \(\bar{G_i}\).
Again, in the randomized model a solution to a problem must be correct with probability at least~\(1 - 1/\poly(n)\).

To see why this matches the intuition of an algorithm being \emph{component-wise}, imagine that $G_i$ is the real input and that~$H_j$ consists of only a single component, that is
\[
    (H_j, (u_1, \ldots, u_j), \localVar_H \restriction_{\neighborhood_T[\{u_1, \ldots, u_j\}]})[u_j]
    =  (H_j, (u_1, \ldots, u_j), \localVar_H \restriction_{\neighborhood_T[\{u_1, \ldots, u_j\}]}) .
\]
By the above definition, the algorithm must produce the same output for~$G_i$ and~$H_j$.
Then, effectively what the algorithm does on~$G_i$ must depend only on what the algorithm has seen in~$H_j$, that is only the current component.

\subsection{Simulation of component-wise randomized \olocal in randomized \local}

It is intuitive that a \boundep outcome implies the existence of a component-wise \olocal algorithm with the same locality: we formalize this intuition through the following lemma.

\begin{lemma}
    Let \(\problem\) be any labeling problem with output node label set \(\Sigma\), and let \(\outcome\) be a \boundept outcome solving \(\problem\) with probability \(p > 0\) and locality \(T(n)\) on graphs of \(n\) nodes.
    Then there exists a component-wise randomized \olocal algorithm solving \(\problem\) with probability \(p\) with the same locality \(T(n)\).
\end{lemma}
\begin{proof}
    Consider any graph \(G\) of \(n\) nodes, and any adversarial processing order of the nodes \((v_1, \ldots, v_n)\).
    We now construct a randomized component-wise \olocal algorithm \(\AA\) with the desired properties.
    When the adversary reveals \(v_1\) and \((G_1, (v_1), \localVar \restriction_{\neighborhood_{T}[v_1]})\) to \(\AA\), \(\AA\) just samples an output \(\oupt(v_1)= \lambda_{v_1}\) for \(v_1\) from \(\outcome(G,\localVar)[\{v_1\}]\).
    Now, when the adversary reveals \(v_i\) and \((G_i, (v_1, \ldots, v_i), \localVar \restriction_{\neighborhood_{T}[\{v_1, \ldots, v_i\}]})[v_i] = (\bar{G}_i, (v_1', \ldots, v_k'), \localVar \restriction_{\neighborhood_{T}[\{v_1', \ldots, v_k'\}]})\) to \(\AA\), it samples an output \(\oupt(v_i)\) for \(v_i\) from the distribution \(\outcome(G,\localVar)[\neighborhood_{T}[\{v_1', \ldots, v_k'\}]]\) conditional on the fact that the outputs of \(v_1', \ldots, v_{k-1}'\) have already been sampled and are equal to \(\oupt(v_1') = \lambda_{v_1'}, \ldots, \oupt(v_{k-1}') = \lambda_{v_{k-1}'}\).
    Let \(\PP_{\AA}\) denote the probability measure induced by the algorithm \(\AA\) on the outputs of the nodes, and let \(\PP_{\outcome}\) denote the probability measure induced by the \boundep outcome \(\outcome\) on the outputs of the nodes.
    Fix any \(k \in [n]\), and suppose
    \[
        (G_{k}, (v_1, \ldots, v_{k}), \localVar \restriction _{\neighborhood_{T}[\{v_1, \ldots, v_{k}\}]})[v_{k}] = (\bar{G}_{k}, (w_1, \ldots, w_h), \localVar \restriction _{\neighborhood_{T}[\{w_1, \ldots, w_h\}]}),
    \]
    with \(w_h = v_k\).
    We first show that, for all \((\lambda_{w_1}, \ldots, \lambda_{w_h}) \in \Sigma^h\), we have that 
    \[
        \PP_\AA(\cap_{i = 1}^{h} \{\oupt(w_i) = \lambda_{w_i}\}) = 
        \PP_{\outcome}(\cap_{i = 1}^{h} \{\oupt(w_i) = \lambda_{w_i}\}).
    \]
    Let us start by induction on \(h\).
    The base case is \(h = 1\).
    Since \(w_1\) is the only node in \(\bar{G_1}\), then it means that there are no nodes in the radius-\((2T)\) neighborhood of \(w_1\) which have been previously labeled (otherwise the partial \olocal run would contain them).
    By construction of \(\AA\) we are freely sampling from \(\outcome(G,\localVar)[\neighborhood_{T}[w_1]]\) and the thesis is trivial.
    Suppose now \(h > 1\).
    We have that 
    \begin{align*}
        &\PP_{\AA}(\cap_{i = 1}^h \{\oupt(w_i) = \lambda_{w_i}\}) \\
        = \ & \PP_{\AA}\left(\oupt(w_h) = \lambda_{v_h} \st \cap_{i = 1}^{h-1} \{w_i = \lambda_{w_i}\}  \right) \PP_{\AA}( \cap_{i = 1}^{h-1} \{w_i = \lambda_{w_i}\}).
    \end{align*}
    By the inductive hypothesis, \(\PP_{\AA}( \cap_{i = 1}^{h-1} \{w_i = \lambda_{w_i}\}) = \PP_{\outcome}( \cap_{i = 1}^{h-1} \{w_i = \lambda_{w_i}\})\) and, by construction of \(\AA\), \(\PP_{\AA}\left(\oupt(w_h) = \lambda_{v_h} \st \cap_{i = 1}^{h-1} \{w_i = \lambda_{w_i}\}  \right)= \PP_{\outcome}\left(\oupt(w_h) = \lambda_{v_h} \st \cap_{i = 1}^{h-1} \{w_i = \lambda_{w_i}\}  \right)\).
    Hence,
    \begin{align*}
        &\PP_{\AA}(\cap_{i = 1}^h \{\oupt(w_i) = \lambda_{w_i}\}) \\
        = \ & \PP_{\AA}\left(\oupt(w_h) = \lambda_{v_h} \st \cap_{i = 1}^{h-1} \{w_i = \lambda_{w_i}\}  \right) \PP_{\AA}( \cap_{i = 1}^{h-1} \{w_i = \lambda_{w_i}\}) \\
        = \ & \PP_{\outcome}\left(\oupt(w_h) = \lambda_{v_h} \st \cap_{i = 1}^{h-1} \{w_i = \lambda_{w_i}\}  \right) \PP_{\outcome}( \cap_{i = 1}^{h-1} \{w_i = \lambda_{w_i}\}) \\
        = \ & \PP_{\outcome}(\cap_{i = 1}^h \{\oupt(w_i) = \lambda_{w_i}\}),
    \end{align*}
    and the induction is complete.
    Let \(S \subseteq V(G)\) with \(S = \{v_{j_1}, \ldots, v_{j_h}\}\) where \(j_1, \ldots, j_h\) are induced by the ordering \(v_1, \ldots, v_n\).
    The above argument implies that, for all \((\lambda_{v_{j_1}}, \ldots, \lambda_{v_{j_h}}) \in \Sigma^{h}\) it holds that 
    \[
        \PP_{\AA}(\cap_{i = 1}^{h} \{\oupt(v_{j_i}) = \lambda_{v_{j_i}}\}) = 
        \PP_{\outcome}(\cap_{i = 1}^{h} \{\oupt(v_{j_i}) = \lambda_{v_{j_i}}\}).
    \]
    In fact, consider the partial \olocal runs \(\{(G_{j_i}, (v_1, \ldots, v_{j_i}), \localVar \restriction _{\neighborhood_{T}[\{v_1, \ldots, v_{j_i}\}]})[v_{j_i}]\}_{i \in [h]}\).
    Consider the maximal subsequence of partial \olocal runs \(\{(G_{j_{i_k}}, (v_1, \ldots, v_{j_{i_k}}), \localVar \restriction _{\neighborhood_{T}[\{v_1, \ldots, v_{j_{i_k}}\}]})[v_{j_{i_k}}]\}_{i_k \in [h^\star]}\) such that \(G_{j_{i_k}}\) has no intersection with \(G_{j_{i}}\) for all \(i > i_k\), \(i \in [h]\).
    We have two properties: 
    First, it trivially holds that \(\cup_{k = 1}^{h^\star}G_{j_{i_k}} =\cup_{i = 1}^{h}G_{j_{i}}\).
    As for the second property, for any two \(k \neq k' \in [h^\star]\), 
    let 
    \[
        (\bar{G}_{j_{i_k}},(w_1, \ldots, w_s),  \localVar \restriction _{\neighborhood_{T}[\{w_1, \ldots, w_{s}\}]}) = (G_{j_{i_k}}, (v_1, \ldots, v_{j_{i_k}}), \localVar \restriction _{\neighborhood_{T}[\{v_1, \ldots, v_{j_{i_k}}\}]})[v_{j_{i_k}}]
    \]
    and
    \[
        (\bar{G}_{j_{i_{k'}}},(w_1', \ldots, w'_{s'}),  \localVar \restriction _{\neighborhood_{T}[\{w_1', \ldots, w_{s'}'\}]}) = (G_{j_{i_{k'}}}, (v_1, \ldots, v_{j_{i_{k'}}}), \localVar \restriction _{\neighborhood_{T}[\{v_1, \ldots, v_{j_{i_{k'}}}\}]})[v_{j_{i_{k'}}}]. 
    \]
    Then, \(\dist_G(\{w_1, \ldots, w_S\},\{w'_1, \ldots, w'_{s'}\}) > 2T\).
    Hence, we can treat the two partial \olocal runs as independent, and we can apply the inductive hypothesis separately to each of them to conclude the proof and get that the overall success probability is at least \(p\).
\end{proof}

\begin{remark}
    The \nonsign model (see \cite{akbari2024} for a definition) cannot be simulated in the component-wise \olocal model because \nonsign can make use of globally shared resources, such as shared randomness, while component-wise \olocal cannot.
\end{remark}

\begin{remark}
    \cref{thm:findep-sim} is almost-tight for the randomized component-wise \olocal model.
    This is because \cite{akbari_et_al:LIPIcs.ICALP.2023.10} provides a deterministic \(O(\log n)\)-round component-wise \olocal algorithm for \(3\)-coloring bipartite graphs, and \cref{thm:findep-sim} would give a \(\tilde{O}(\sqrt{n})\)-round deterministic \local algorithm for the same problem.
    However, we know that \({\Omega}(\sqrt{n})\) is a lower bound for the problem in the \local model, as well as in the bounded-dependence model (and even in the \nonsign model).
\end{remark}

\end{document}